\documentclass[acmsmall]{acmart}
\AtBeginDocument{%
  }

\usepackage[utf8]{inputenc} % allow utf-8 input
\usepackage{graphicx}
\usepackage[T1]{fontenc}    % use 8-bit T1 fonts
\usepackage{hyperref}       % hyperlinks
\usepackage{url}            % simple URL typesetting
\usepackage{booktabs}       % professional-quality tables
\usepackage{longtable}      % for latex longtables in the IMP rules
\usepackage{amsfonts}       % blackboard math symbols
\usepackage{amsthm}         % theorem and definition styles
\usepackage{nicefrac}       % compact symbols for 1/2, etc.
\usepackage{microtype}      % microtypography
\usepackage[svgnames,x11names]{xcolor}  % colors
\usepackage{multirow}         % for multirow in tables
\usepackage[most]{tcolorbox}  % to draw customized boxes
\tcbuselibrary{listings}      % to use listing inside tcolorbox
\tcbuselibrary{skins}         % to set title position in tcolorbox
\tcbuselibrary{breakable}     % break across pages
\usepackage{xspace}           % for adding space after commands
\usepackage{wrapfig}          % for adding text wrapped figures
\usepackage{mathpartir}       % for inferrule
\usepackage{amsmath,bm}       % for math symbols
\usepackage[font=footnotesize,labelfont=bf]{caption} % for changing figure and table caption font size
\usepackage[font=footnotesize]{subcaption}           % for subcaption environment
\usepackage{outlines}       % for bulleted list with subitems
\usepackage{tikz}           % for drawing figures with Tikz
\usetikzlibrary{tikzmark,
calc, positioning,
shapes.multipart,
decorations.pathmorphing,
decorations.pathreplacing,
fit,backgrounds,
matrix,
arrows.meta}                % for using polar coords in Tikz pictures
\usepackage{mathrsfs}       % math fonts
\usepackage{varwidth}       % for proper wrapfig with display mode
\usepackage{enumitem}       % for custom spacing in itemize/enumeration envs
\usepackage{colortbl}       % for row coloring
\usepackage[symbol]{footmisc} % for footnote symbols
\usepackage{accents}        % for underlining customization
\usepackage{mathtools}      % for cleanly typesetting '::='
\usepackage{xparse}         % for multiple default args for Latex commands
\usepackage{dsfont}         % for mathbb numbers
\usepackage{expl3}         % for expl3 macros
\usepackage{subdepth}      % for consistent subscripts

\lstdefinelanguage{SQLCustom}{
  keywords={
    SELECT, FROM, WHERE, AND, OR, IN, AS, GROUP, BY, COUNT, JOIN, ON, SUM
  },
  keywords=[2]{
    category, pid, name, spend, totalSpend 
  },
  sensitive=false,
  morecomment=[l]{--},
  morestring=[b]',
}

\lstdefinelanguage{CypherCustom}{
  keywords={
    MATCH, WHERE, AS, AND, OR, NOT, RETURN
  },
  sensitive=false,
  morecomment=[l]{--},
  morestring=[b]',
}

\lstdefinelanguage{QueryPlan}{
  keywords={},
  keywords=[2]{
    category, pid, name, spend, totalSpend 
  },
  sensitive=false,
  morecomment=[l]{//},
  morestring=[b]',
}

\lstdefinestyle{sqlstyle}{
  language=SQLCustom,
  basicstyle=\ttfamily\footnotesize,
  keywordstyle=\bfseries,
  keywordstyle=[2]\itshape,
  commentstyle=\color{gray}\ttfamily,
  stringstyle=\color{brown},
  showstringspaces=false,
  columns=fullflexible,
  keepspaces=true,
  breaklines=true,
  frame=none,
  tabsize=2,
  aboveskip=0.5em,
  belowskip=0.5em,
  literate={*}{{\textbf{*}}}1
}

\lstdefinestyle{cypherstyle}{
  language=CypherCustom,
  basicstyle=\ttfamily\footnotesize,
  keywordstyle=\bfseries,
  keywordstyle=[2]\itshape,
  commentstyle=\color{gray}\ttfamily,
  stringstyle=\color{brown},
  showstringspaces=false,
  columns=fullflexible,
  keepspaces=true,
  breaklines=true,
  frame=none,
  tabsize=2,
  aboveskip=0.5em,
  belowskip=0.5em,
  literate={*}{{\textbf{*}}}1
}

\lstdefinestyle{queryplanstyle}{
  language=QueryPlan,
  basicstyle=\ttfamily\footnotesize,
  keywordstyle=\bfseries,
  keywordstyle=[2]\itshape,
  commentstyle=\color{gray}\ttfamily,
  stringstyle=\color{teal},
  showstringspaces=false,
  columns=fullflexible,
  keepspaces=true,
  breaklines=true,
  breakatwhitespace=false,
  frame=none,
  tabsize=2,
  aboveskip=0.5em,
  belowskip=0.5em,
  escapeinside={(*@}{@*)},
  literate={*}{{\textbf{*}}}1
}

\newtheoremstyle{myThm}% name of the style to be used
  {2mm}% measure of space to leave above the theorem. E.g.: 3pt
  {2mm}% measure of space to leave below the theorem. E.g.: 3pt
  {}% name of font to use in the body of the theorem
  {0pt}% measure of space to indent
  {\itshape}% name of head font
  {.  }% punctuation between head and body
  { }% space after theorem head; " " = normal interword space
  {\thmname{#1}\thmnumber{ #2}{ (\thmnote{#3})}}
\theoremstyle{myThm}
\newtheorem{myDef}{Definition}[section]
\newtheorem{myExample}{Example}[section]
\newtheorem{myNotation}{Notation}[section]
\newtheorem{myLemma}{Lemma}[section]
\newtheorem{myTheorem}{Theorem}[section]
\newtheorem{myCorol}{Corollary}[section]

\newenvironment{CypherBox}{%
  \tcblisting{
    enhanced, rounded corners=all,
    hbox,
    boxrule=0.5pt,
    colback=white,
    colframe=black,
    arc=0.5pt,
    left=-1pt,right=7.6cm,top=-1pt,bottom=-2pt,
    listing only,
    listing options={style=cypherstyle,basicstyle=\bfseries\footnotesize},
  }
}{\endtcblisting}

\newtcolorbox{GrammarBox}{%
    enhanced, rounded corners=all,
    hbox,
    boxrule=0.5pt,
    colback=white,
    colframe=black,
    arc=0.5pt,
    left=-2pt,right=5mm,top=-7.5pt,bottom=-2pt,
}

\newtcolorbox{SemanticsBox}{%
    enhanced, rounded corners=all,
    boxrule=0.5pt,
    colback=white,
    colframe=black,
    arc=1pt,
    left=-1.2mm,right=-3.5pt,top=-2.2pt,bottom=0pt,
    fontupper=\FormulaFont,
}
\newtcolorbox{TypingDisplay}[1][\normalsize]{
  enhanced,
  boxrule=0pt,
  frame hidden,
  colback=white,
  left=0pt,
  right=0pt,
  top=-1mm,
  bottom=0pt,
  before skip=0pt,
  after skip=1mm,
  fontupper=#1,
}

\newtcolorbox{RulesDisplay}{
  enhanced,
  breakable=false,
  boxrule=0.4pt,
  colback=white,
  colframe=white,
  arc=0.6pt,
  left=0pt,
  right=0pt,
  top=-0.4mm,
  bottom=1.3mm,
  before skip=3mm,
  after skip=3mm,
  fontupper=\FormulaFont,
}

\newtcolorbox{RulesDisplayCustomFont}[1]{
  enhanced,
  breakable=false,
  boxrule=0.4pt,
  colback=white,
  frame hidden,
  arc=0.6pt,
  left=0pt,
  right=0pt,
  top=-0.4mm,
  bottom=1.3mm,
  before skip=1mm,
  after skip=1mm,
  fontupper=#1,
}

\newtcolorbox{RulesDisplayFrameless}{
  enhanced,
  breakable=false,
  boxrule=0.4pt,
  colback=white,
  frame hidden,
  arc=0.6pt,
  left=0pt,
  right=0pt,
  top=-0.4mm,
  bottom=1.3mm,
  before skip=1mm,
  after skip=1mm,
  fontupper=\FormulaFont,
}
\newcommand{\DefMacro}[2]{\expandafter\newcommand\csname rmk-#1\endcsname{#2}}
\newcommand{\UseMacro}[1]{\csname rmk-#1\endcsname}

\newcommand{\TableFont}{\scriptsize}
\newcommand{\FormulaFont}{\scriptsize}
\newcommand{\ExampleFont}{\small}
\newcommand{\RuleNameFont}{\footnotesize}
\newcommand{\CyFigCommentColor}{gray!70}
\newcommand{\CyGrammarRowColor}{gray!10}
\newcommand{\CyIsoRefColor}{ACMPurple!60}

\newcommand{\ExamplePatMatchNodeColor}{red}
\newcommand{\ExamplePatKnowsEdgeColor}{Salmon}
\newcommand{\ExamplePatPersonNodeColor}{blue}
\newcommand{\ExamplePatLeadsEdgeColor}{DarkCyan}
\newcommand{\ExamplePatProjectNodeColor}{purple}
\newcommand{\GqlTokenColor}{RoyalBlue3}

\definecolor{GqlQryRGB}{RGB}{0,125,188}      % query syntax elements
\definecolor{GqlPatRGB}{RGB}{166,89,25}      % pattern syntax elements
\definecolor{GqlExpRGB}{RGB}{85,119,15}      % expression syntax elements
\definecolor{GqlMetaRGB}{RGB}{107,197,189}    % metavariables
\definecolor{GqlTypeRGB}{RGB}{189,61,111}    % types
\definecolor{GqlSchemaRGB}{RGB}{0,130,108}   % schemas
\definecolor{GqlOpRGB}{RGB}{75,75,75}        % operators
\newcommand{\GqlQryColor}{GqlQryRGB}
\newcommand{\GqlPatColor}{GqlPatRGB}
\newcommand{\GqlExpColor}{GqlExpRGB}

\newcommand{\GqlTypeColor}{GqlTypeRGB}
\newcommand{\GqlSchemaColor}{GqlSchemaRGB}

\newcommand{\StxClr}[1]{\textcolor{\GqlTokenColor}{#1}}
\newcommand{\QryClr}[1]{\textcolor{\GqlQryColor}{#1}}
\newcommand{\PatClr}[1]{\textcolor{\GqlPatColor}{#1}}
\newcommand{\ExpClr}[1]{\textcolor{\GqlExpColor}{#1}}
\newcommand{\MetaClr}[1]{\textcolor{black}{#1}}
\newcommand{\TyClr}[1]{\textcolor{\GqlTypeColor}{#1}}
\newcommand{\SchClr}[1]{\textcolor{\GqlSchemaColor}{#1}}

\newcommand{\lang}{GQL\xspace}
\newcommand{\iso}{ISO/IEC~39075\xspace}
\newcommand{\isoShort}{\textcolor{\CyIsoRefColor}{$\mathrm{ISO}$}\xspace}
\newcommand{\isoSectionShort}[1]{\textcolor{\CyIsoRefColor}{$\mathrm{ISO}$~\S#1}}

\newcommand{\isoSyntxRulePg}[2]{\textcolor{\CyIsoRefColor}{$\mathrm{SR}$:#1} ($\mathrm{Page}$~#2)}

\newcommand{\isoSection}[1]{\isoSectionShort{#1}}
\newcommand{\isoFeature}[1]{#1}

\newcommand{\tool}{MGQL}
\newcommand{\lean}{Lean4\xspace}
\newcommand{\coq}{Coq\xspace}
\newcommand{\Title}{\tool: An Executable, Small-Step Semantics of \lang}
\newcommand{\sql}{SQL\xspace}
\newcommand{\sparql}{SPARQL\xspace}
\newcommand{\rdf}{RDF\xspace}

\newcommand{\cypher}{Cypher\xspace}

\newcommand{\pgql}{PGQL\xspace}
\newcommand{\gcore}{G-CORE\xspace}

\newcommand{\threeVL}{3VL\xspace}
\newcommand{\pathMode}[1]{\texttt{#1}\xspace}
\newcommand{\pathModeTrail}{\pathMode{TRAIL}}
\newcommand{\pathModeWalk}{\pathMode{WALK}}
\newcommand{\pathModeAcyclic}{\pathMode{ACYCLIC}}
\newcommand{\pathModeSimple}{\pathMode{SIMPLE}}
\newcommand{\matchMode}[2]{\texttt{#1} \texttt{#2}\xspace}
\newcommand{\matchModeDiff}{\matchMode{DIFFERENT}{EDGES}}
\newcommand{\matchModeRepeat}{\matchMode{REPEATABLE}{ELEMENTS}}
\newcommand{\isoAnnotate}[4]{%
  \textit{\textcolor{\CyFigCommentColor}{%
    #1 (\isoSectionShort{#2}%
    \ifx&#3&%
    \else%
      : \isoFeature{#3}%
    \fi%
    )%
    \ifx&#4&%
    \else%
      ~[\textcolor{\CyFigCommentColor}{#4}]%
    \fi%
  }}%
}

\newcommand{\isoAnnotatePlain}[1]{%
  \textcolor{\CyFigCommentColor}{\emph{#1}}%
}

\newcommand{\isoAnnotateBaseline}[2]{%
  \textcolor{\CyFigCommentColor}{\emph{#1}}\;%
  \textcolor{\CyFigCommentColor}{(}{\tiny\textbf{\textcolor{\CyIsoRefColor}{\isoSectionShort{#2}}}}\textcolor{\CyFigCommentColor}{)}%
}

\newcommand{\featAnnotateComposite}[2]{%
  \textcolor{\CyFigCommentColor}{\emph{#1}}\;%
  \textcolor{\CyFigCommentColor}{[}{\tiny\text{#2}}\textcolor{\CyFigCommentColor}{]}%
}

\newcommand{\isoAnnotateComposite}[3]{%
  \textcolor{\CyFigCommentColor}{\emph{#1}}\;%
  \textcolor{\CyFigCommentColor}{(}{\tiny\textbf{\textcolor{\CyIsoRefColor}{\isoSectionShort{#2}}}}\textcolor{\CyFigCommentColor}{)}~%
  \textcolor{\CyFigCommentColor}{[}{\tiny\text{#3}}\textcolor{\CyFigCommentColor}{]}%
}

\newcommand{\isoGType}{4.13.2}
\newcommand{\isoGraphType}{4.13.2.2}
\newcommand{\isoNodeType}{4.13.2.3}
\newcommand{\isoEdgeType}{4.13.2.4}
\newcommand{\isoPropertyType}{4.13.2.5}

\newcommand{\isoPatternMatching}{4.11}
\newcommand{\isoPatternMatchingPathMode}{4.11.7}
\newcommand{\isoPatternMatchingMatchMode}{4.11.9}

\newcommand{\isoFocusedLinQuery}{14.3}

\newcommand{\isoCompositeQuery}{14.2}
\newcommand{\isoQueryProjection}{14.11}

\newcommand{\isoRefContext}{4.11.5}
\newcommand{\isoSyntxPrjAlias}{8}
\newcommand{\isoPgPrjAlias}{186}
\newcommand{\isoSyntxCompOp}{3}
\newcommand{\isoPgCompOp}{164}

\newcommand{\isoPatternExpression}{16.4}
\newcommand{\isoPathPattern}{16.7}
\newcommand{\isoNodePattern}{16.7}
\newcommand{\isoEdgePattern}{16.7}
\newcommand{\isoQuantifiers}{16.11}
\newcommand{\isoPropertyMap}{16.7}
\newcommand{\isoLabelExpression}{16.8}
\newcommand{\isoPathEval}{22.3}

\newcommand{\isoSyntxPatEdgeCount}{8}
\newcommand{\isoPgPatEdgeCount}{229}

\newcommand{\isoSyntxPatNodeCount}{16}
\newcommand{\isoPgPatNodeCount}{231}

\newcommand{\isoPredicates}{19}
\newcommand{\isoComparisonPredicate}{19.3}
\newcommand{\isoNullPredicate}{19.5}
\newcommand{\isoValueExpressions}{20}

\newcommand{\isoPropertyAccess}{20.11}
\newcommand{\isoAggregation}{20.9}
\newcommand{\isoLiterals}{21}
\newcommand{\isoPatternVarName}{21.1}

\newcommand{\isoFeatureStyle}[1]{$\mathrm{#1}$}
\newcommand{\isoFeatFocusedLinQuery}{\isoFeatureStyle{GQ01}}
\newcommand{\isoFeatCompOtherwise}{\isoFeatureStyle{GQ02}}
\newcommand{\isoFeatCompUnion}{\isoFeatureStyle{GQ03}}
\newcommand{\isoFeatCompExceptDstnct}{\isoFeatureStyle{GQ04}}
\newcommand{\isoFeatCompExceptAll}{\isoFeatureStyle{GQ05}}
\newcommand{\isoFeatCompIntersectDstnct}{\isoFeatureStyle{GQ06}}
\newcommand{\isoFeatCompIntersectAll}{\isoFeatureStyle{GQ07}}

\newcommand{\isoFeatPathMultisetAlt}{\isoFeatureStyle{G030}}

\newcommand{\isoFeatQuantifiedPaths}{\isoFeatureStyle{G035}}
\newcommand{\isoFeatQuantifiedEdges}{\isoFeatureStyle{G036}}
\newcommand{\isoFeatParenthesizedPath}{\isoFeatureStyle{G038}}

\newcommand{\isoFeatBoundedQuantifiers}{\isoFeatureStyle{G060}}
\newcommand{\isoFeatUnboundedQuantifiers}{\isoFeatureStyle{G061}}

\newcommand{\isoFeatWildcardLabel}{\isoFeatureStyle{G074}}
\newcommand{\isoPageCount}{600+\xspace}
\newcommand{\Code}[1]{{\ifmmode{\mathtt{#1}}\else$\mathtt{#1}$\fi}}

\newcommand{\eg}{e.g.,\xspace}
\newcommand{\ie}{i.e.,\xspace}
\newcommand{\etc}{etc.,\xspace}

\newcommand{\MyPara}[1]{\vspace{1pt}\noindent\textbf{#1}.\xspace}

\DeclareRobustCommand{\circledCaption}[1]{\tikz[baseline=(char.base)]{\node[shape=circle, draw, minimum size=4pt, inner sep=1pt, fill=black, text=white] (char) {\tiny #1};}}

\DeclareRobustCommand{\circledColor}[5]{\tikz[baseline=(char.base)]{\node[shape=circle, draw={#2}, minimum size=4pt, inner sep=1pt, fill={#2}, text={#3}] (char) {\raisebox{#5}{\scalebox{#4}{#1}}};}}
\DeclareRobustCommand{\circledWhite}[1]{\tikz[baseline=(char.base)]{\node[shape=circle, draw=black, minimum size=4pt, inner sep=0.6pt, fill=white, text=black, thick] (char) {\scalebox{0.8}{#1}};}}

\newcommand{\redcircled}[1]{H\xspace}
\newcommand{\enumParen}[1]{%
\footnotesize\textcolor{ACMPurple!80}{\textbf{[}}#1\textcolor{ACMPurple!80}{\textbf{].}}}
\newcommand{\underboxcap}[2]{%
\tikz[baseline=(X.base)]{
  \node[inner sep=0pt, outer sep=0pt] (X) {$#1$};
  \def\yoff{1ex}   % distance below the text
  \def\thick{0.5pt}   % line thickness
  \def\cap{0.8ex}     % cap height
  \def\labsep{0.3ex}  % label depth

  \draw[line width=\thick, color=black!70]
    ($(X.south west)+(0,-\yoff)$) -- ($(X.south east)+(0,-\yoff)$);

  \draw[line width=\thick, color=black!70]
    ($(X.south west)+(0,-\yoff)$) -- ($(X.south west)+(0,-\yoff+\cap)$);

  \draw[line width=\thick, color=black!70]
    ($(X.south east)+(0,-\yoff)$) -- ($(X.south east)+(0,-\yoff+\cap)$);

  \draw[line width=\thick, color=black!70]
    ($(X.south)+(0,-\yoff)$) -- ($(X.south)+(0,-\yoff-\cap)$);

  \node[overlay, anchor=north, color=black!70] at
    ($(X.south)+(0,-\yoff-\labsep)$)
    {#2};
}%
}

\newcommand{\tikzGEdge}[5]{
\tikz[baseline=(a.base)]{
  \node[shape=circle, draw=white, text=black, minimum size=6pt, inner sep=1.5pt, fill={#2}, line width=1pt] (a) {#1};
  \node[shape=circle, draw=white, text=black, minimum size=6pt, inner sep=1.5pt, fill={#4}, line width=1pt] (b) at ($(a.east)+(1.5cm,0)$) {#3};
  \draw[-{Latex[length=1.6mm,width=1.1mm]},thick] (a) -- node[fill=white,inner sep=1.2pt,font=\footnotesize, midway, xshift=-2pt] {\texttt{#5}} (b);
}%
}

\newcommand{\tikzGEdgeOpa}[1]{
\tikz[baseline=(a.base)]{
  \node[draw=white, inner sep=0pt, line width=0pt, draw opacity = 0.0, fill opacity = 0.0] (a) {\phantom{i}};
  \node[draw=white, inner sep=0pt, line width=0pt, draw opacity = 0.0, fill opacity = 0.0] (b) at ($(a.east)+(1.5cm,0)$) {\phantom{i}};
  \draw[-{Latex[length=1.6mm,width=1.1mm]},thick] (a) -- node[fill=white,inner sep=1.2pt,font=\footnotesize, midway, xshift=-2pt] {\texttt{#1}} (b);
}%
}

\newcommand{\tikzGTwoEdge}[8]{
\tikz[baseline=(a.base)]{
  \node[shape=circle, draw=white, text=black, minimum size=6pt, inner sep=1.5pt, fill={#2}, line width=1pt] (a) {#1};
  \node[shape=circle, draw=white, text=black, minimum size=6pt, inner sep=1.5pt, fill={#4}, line width=1pt] (b) at ($(a.east)+(1.5cm,0)$) {#3};
  \node[shape=circle, draw=white, text=black, minimum size=6pt, inner sep=1.5pt, fill={#6}, line width=1pt] (c) at ($(b.east)+(1.5cm,0)$) {#5};
  \draw[-{Latex[length=1.6mm,width=1.1mm]},thick] (a) -- node[fill=white,inner sep=1.2pt,font=\footnotesize, midway, xshift=-2pt] {\texttt{#7}} (b);
  \draw[-{Latex[length=1.6mm,width=1.1mm]},thick] (b) -- node[fill=white,inner sep=1.2pt,font=\footnotesize, midway, xshift=-2pt] {\texttt{#8}} (c);
}%
}

\newcommand{\tikzGThreeEdge}[8]{
\tikz[baseline=(a.base)]{
  \node[shape=circle, draw=white, text=black, minimum size=6pt, inner sep=1.5pt, fill={#2}, line width=1pt] (a) {#1};
  \node[shape=circle, draw=white, text=black, minimum size=6pt, inner sep=1.5pt, fill={#4}, line width=1pt] (b) at ($(a.east)+(1.5cm,0)$) {#3};
  \node[shape=circle, draw=white, text=black, minimum size=6pt, inner sep=1.5pt, fill={#6}, line width=1pt] (c) at ($(b.east)+(1.5cm,0)$) {#5};
  \node[shape=circle, draw=white, text=black, minimum size=6pt, inner sep=1.5pt, fill={#8}, line width=1pt] (d) at ($(c.east)+(1.5cm,0)$) {#7};
  \draw[-{Latex[length=1.6mm,width=1.1mm]},thick] (a) -- node[fill=white,inner sep=1.2pt,font=\footnotesize, midway, xshift=-2pt] {\texttt{KNOWS}} (b);
  \draw[-{Latex[length=1.6mm,width=1.1mm]},thick] (b) -- node[fill=white,inner sep=1.2pt,font=\footnotesize, midway, xshift=-2pt] {\texttt{KNOWS}} (c);
  \draw[-{Latex[length=1.6mm,width=1.1mm]},thick] (c) -- node[fill=white,inner sep=1.2pt,font=\footnotesize, midway, xshift=-2pt] {\texttt{LEADS}} (d);
}%
}

\newcommand{\tikzGEdgeUDir}[5]{
\tikz[baseline=(a.base)]{
  \node[shape=circle, draw=white, text=black, minimum size=6pt, inner sep=1.5pt, fill={#2}, line width=1pt] (a) {#1};
  \node[shape=circle, draw=white, text=black, minimum size=6pt, inner sep=1.5pt, fill={#4}, line width=1pt] (b) at ($(a.east)+(1.5cm,0)$) {#3};
  \draw[thick] (a) -- node[fill=white,inner sep=1.2pt,font=\footnotesize, midway] {\texttt{#5}} (b);
}%
}
\newcommand{\fPropGraph}{\mathcal{G}}

\newcommand{\fPGNode}{\mathcal{N}}
\newcommand{\fPGEdge}{\mathcal{E}}
\newcommand{\fPGEdgeDir}{\mathcal{E}^{\scalebox{0.6}{$\!\fPGDir\!$}}}
\newcommand{\fPGEdgeUDir}{\mathcal{E}^{\scalebox{0.6}{$\!\fPGUDir\!$}}}
\newcommand{\fPGNodeElem}{\mathrm{n}}
\newcommand{\fPGEdgeElem}{\mathrm{e}}
\newcommand{\fPGElem}{\iota}
\newcommand{\fPGNodeElems}[1]{\fPGNodeElem_{#1}}

\newcommand{\fPGNodelabel}[1][]{\MetaClr{\text{$l_{\scalebox{0.7}{#1}}$}}}

\newcommand{\fPGNodeSort}{\mathsf{N}}
\newcommand{\fPGEdgeSort}{\mathsf{E}}

\newcommand{\fPGSrcDstFunc}{\Theta}
\newcommand{\fPGLabelFunc}{\Lambda}
\newcommand{\fPGPropFunc}{\Xi}
\newcommand{\fPGDirFunc}{\Delta}
\newcommand{\fPGDirBase}[1]{%
  \mathrel{%
    \tikz[baseline=-0.65ex]{%
      \node (a) {\scalebox{1.2}{$\!\circlearrowright\!$}};%
      \ifnum#1=0\relax
      \else
        \draw[line width=0.5pt] (0.15,0.1) -- (-0.15,-0.1);%
        \draw[line width=0.5pt] (-0.15,0.1) -- (0.15,-0.1);%
      \fi
    }%
  }%
}
\newcommand{\fPGDir}{\fPGDirBase{0}}
\newcommand{\fPGUDir}{\fPGDirBase{1}}

\newcommand{\fPGProp}[1][]{\Pi_{#1}}
\newcommand{\fCyElemNodePattern}[1]{\StxClr{\textbf{(}}#1\StxClr{\textbf{)}}}
\newcommand{\fCyElemEdgePattern}[1]{\StxClr{\textbf{[}}#1\StxClr{\textbf{]}}}
\newcommand{\straightarrowR}{
\tikz[x=1ex,y=1ex,baseline=-0.6ex]{
\draw[-{Latex[length=0.4em,width=0.4em]},inner sep=0pt, outer sep=0pt,  line width=0.45pt]
(0,0) -- (3.3,0);
}
}
\newcommand{\straightarrowL}{
\tikz[x=1ex,y=1ex,baseline=-0.6ex]{
\draw[-{Latex[length=0.4em,width=0.4em]},inner sep=0pt, outer sep=0pt,  line width=0.45pt]
(0,0) -- (-3.3,0);
}
}
\newcommand{\straightarrow}{
\tikz[x=1ex,y=1ex,baseline=-0.6ex]{
\draw[inner sep=0pt, outer sep=0pt, line width=0.45pt]
(0,0) -- (-3.3,0);
}
}

\newcommand{\squigpic}[5]{%
  \tikz[x=1ex,y=1ex,baseline=-0.6ex]{
    \path[use as bounding box] (#3,0) rectangle (#2,0);
    \draw[#1, inner sep=0pt, outer sep=0pt, line width=0.45pt, decorate, decoration={snake,amplitude=.1em,segment length=0.26em,pre length = #4, post length = #5}] (#3,0) -- (#2,0);
  }%
}

\newcommand{\squig}{%
  \squigpic{}{1.5em}{0em}{0.13em}{0.12em}%
}

\newcommand{\squigarrowL}{%
  \squigpic{-{Latex[length=0.4em,width=0.4em]}}{-1.5em}{0em}{0.15em}{0.5em}%
}

\newcommand{\squigarrowR}{%
  \squigpic{-{Latex[length=0.4em,width=0.4em]}}{1.5em}{0em}{0.15em}{0.5em}%
}

\newcommand{\fCyEdgeRight}[1]{\,\straightarrow{#1}\straightarrowR\,}
\newcommand{\fCyEdgeLeft}[1]{\,\straightarrowL{#1}\straightarrow\,}
\newcommand{\fCyEdgeLeftRight}[1]{\,\straightarrowL{#1}\straightarrowR\,}
\newcommand{\fCyEdgeNone}[1]{\,\straightarrow{#1}\straightarrow\,}
\newcommand{\fCyEdgeRightU}[1]{\,\squig{#1}\squigarrowR\,}
\newcommand{\fCyEdgeLeftU}[1]{\,\squigarrowL{#1}\squig\,}
\newcommand{\fCyEdgeNoneU}[1]{\,\squig{#1}\squig\,}

\newcommand{\fCyQuantStar}{\StxClr{*}}
\newcommand{\fCyQuantPlus}{\StxClr{+}}
\newcommand{\fCyQuantQues}{\StxClr{?}}
\newcommand{\fCyQuantExact}[1]{\StxClr{\{}#1\StxClr{\}}}
\newcommand{\fCyQuantBound}[2]{\StxClr{\{}#1,#2\StxClr{\}}}

\ExplSyntaxOn
\cs_new:Npn \fGqAtom:nnnn #1#2#3#4 {
  \str_case:nnF {#1} {
    {vlp} { \textbf{#4}\;\fCyColon {#2}\;\fCyQuantExact{#3} }
    {vl} { \textbf{#4}\;\fCyColon {#2} }    
    {vp} { \textbf{#4}\;\fCyQuantExact{#3} }
    {v} { \textbf{#4} }
    {l} { \fCyColon\texttt{#2}}   
  }
  { \textbf{#4} }
}

\cs_new:Npn \fGqElemAtom:nnnnn #1#2#3#4#5 {
  \str_case:nnF {#1} {
    {node} { \fCyElemNodePattern{\fGqAtom:nnnn {#2} {#3} {#4} {#5}}} 
    {edge} { \fCyElemEdgePattern{\fGqAtom:nnnn {#2} {#3} {#4} {#5}}} 
  } {#1} 
}

\cs_new:Npn \fGqEdge:nnnnn #1#2#3#4#5 {
  \str_case:nnF {#1} {
    {->} { \fCyEdgeRight{\fGqElemAtom:nnnnn {edge} {#2} {#3} {#4} {#5}}}
    {<-} { \fCyEdgeLeft{\fGqElemAtom:nnnnn {edge} {#2} {#3} {#4} {#5}}}
    {<->} { \fCyEdgeLeftRight{\fGqElemAtom:nnnnn {edge} {#2} {#3} {#4} {#5}}}
    {;>} { \fCyEdgeRightU{\fGqElemAtom:nnnnn {edge} {#2} {#3} {#4} {#5}}}
    {<;} { \fCyEdgeLeftU{\fGqElemAtom:nnnnn {edge} {#2} {#3} {#4} {#5}}}
    {;} { \fCyEdgeNoneU{\fGqElemAtom:nnnnn {edge} {#2} {#3} {#4} {#5}}}
    {-} { \fCyEdgeNone{\fGqElemAtom:nnnnn {edge} {#2} {#3} {#4} {#5}}}
    {->e} { \fCyEdgeRight{#5}}
    {<-e} { \fCyEdgeLeft{#5}}
    {<->e} { \fCyEdgeLeftRight{#5}}
    {;>e} { \fCyEdgeRightU{#5}}
    {<;e} { \fCyEdgeLeftU{#5}}
    {;e} { \fCyEdgeNoneU{#5}}
    {-e} { \fCyEdgeNone{#5}}
  } { \fCyEdgeRight{\fGqElemAtom:nnnnn {edge} {#2} {#3} {#4} {#5}}}
}

\cs_new:Npn \fGqEdgeQuant:nnnnnnnn #1#2#3#4#5#6#7#8 {
  \fGqEdge:nnnnn {#1} {#2} {#3} {#4} {#8} 
  \str_case:nnF {#5} {
    {*} { \fCyQuantStar }
    {+} { \fCyQuantPlus }
    {?} { \fCyQuantQues }
    {exact} { \fCyQuantExact{#6} }
    {bound} { \fCyQuantBound{#6}{#7} }
  } { \fCyQuantStar }
}

\NewDocumentCommand{\fGqNode}{O{v} O{} O{} m}{\fGqElemAtom:nnnnn {node} {#1} {#2} {#3} {#4}}
\NewDocumentCommand{\fGqEdgeAtom}{O{v} O{} O{} m}{\fGqElemAtom:nnnnn {edge} {#1} {#2} {#3} {#4}}
\NewDocumentCommand{\fGqEdge}{O{->} O{v} O{} O{} m}{\fGqEdge:nnnnn {#1} {#2} {#3} {#4} {#5}}
\NewDocumentCommand{\fGqEdgeQuant}{O{->} O{v} O{} O{} O{*} O{} O{} m}{\fGqEdgeQuant:nnnnnnnn {#1} {#2} {#3} {#4} {#5} {#6} {#7} {#8}}
\ExplSyntaxOff

\newcommand{\fMatch}{\textbf{\textcolor{\GqlTokenColor}{\textsc{Match}}}}
\newcommand{\fWhere}{\textbf{\textcolor{\GqlTokenColor}{\textsc{Where}}}}
\newcommand{\fReturn}{\textbf{\textcolor{\GqlTokenColor}{\textsc{Return}}}}
\newcommand{\fAs}{\textbf{\textcolor{\GqlTokenColor}{\textsc{As}}}}

\newcommand{\fCyQueryExp}{\QryClr{\text{$q$}}}
\newcommand{\fCyQuery}{\QryClr{\text{$Q$}}}

\newcommand{\fCyNodeVar}{\MetaClr{\text{\emph{v}}}}

\newcommand{\fCyNodeVars}[1]{\fCyNodeVar_{\scalebox{0.7}{#1}}}

\newcommand{\fCyVar}{\MetaClr{\text{$x$}}}
\newcommand{\fCyVars}[1]{\MetaClr{\text{$x$}_{#1}}}
\newcommand{\fCyColon}{\,\StxClr{\textbf{:}}\,}
\newcommand{\fCyComma}{\StxClr{\textbf{,}}\,}
\newcommand{\fCyVarProp}[2]{#1\,\StxClr{\textbf{.}}\,#2}
\newcommand{\fCyProp}{\MetaClr{\text{$k$}}}
\newcommand{\fCyNodeLabel}{\text{$L$}}
\newcommand{\fCyNodeLabels}{\overline{\fCyNodeLabel}}

\newcommand{\fCyQuantifier}{{\text{$K$}}}

\newcommand{\fCyPatternNode}{\PatClr{\text{$N$}}}
\newcommand{\fCyPatternEdge}{\PatClr{\text{$E$}}}
\newcommand{\fCyPatternAtom}{\PatClr{\text{$A$}}}
\newcommand{\fCyPatternNodes}[1]{\PatClr{\text{$N_{#1}$}}}
\newcommand{\fCyPatternEdges}[1]{\PatClr{\text{$E_{#1}$}}}

\newcommand{\fCyVarSyn}{\MetaClr{\textsf{VarDcl}}}
\newcommand{\fCyDscSyn}{\MetaClr{\textsf{Dsc}}}
\newcommand{\fCyLabelSyn}{\MetaClr{\textsf{LblDcl}}}
\newcommand{\fCyPropMap}{\MetaClr{\textsf{PrpDcl}}}
\newcommand{\fCyExp}{\ExpClr{\vartheta}}
\newcommand{\fCyExpConst}[1]{\ExpClr{\text{$c_{#1}$}}}
\newcommand{\fCyExpBConst}[1]{\ExpClr{\text{$b_{#1}$}}}
\newcommand{\fCyExps}[1]{\ExpClr{\text{$\vartheta_{#1}$}}}
\newcommand{\fCyAggExp}{\ExpClr{\text{$\alpha$}}}

\newcommand{\fCyPred}{\ExpClr{\phi}}

\newcommand{\fCyExpLogic}[2]{#1\bm{\ExpClr{\,\otimes\,}}#2}
\newcommand{\fCyExpRel}[2]{#1\bm{\ExpClr{\,\odot\,}}#2}
\newcommand{\fCyExpRelOp}[3]{#1\bm{\textcolor{\GqlTokenColor}{\,#3\,}}#2}
\newcommand{\fCyExpArith}[2]{#1\bm{\ExpClr{\,\oplus\,}}#2}
\newcommand{\fCyExpArithOp}[3]{#1\bm{\textcolor{\GqlTokenColor}{\,#3\,}}#2}
\newcommand{\fCyExpAnd}[2]{#1\,\StxClr{\textbf{\textsc{And}}}\,#2}
\newcommand{\fCyExpOr}[2]{#1\,\StxClr{\textbf{\textsc{Or}}}\,#2}
\newcommand{\fCyExpNeg}[1]{\StxClr{\textbf{\textsc{Not}}}\,#1}
\newcommand{\fCyAExpCount}{\textbf{\textcolor{\GqlTokenColor}{\textsc{Count}}}}
\newcommand{\fCyAExpSum}{\textbf{\textcolor{\GqlTokenColor}{\textsc{Sum}}}}
\newcommand{\fCyAExpMax}{\textbf{\textcolor{\GqlTokenColor}{\textsc{Max}}}}
\newcommand{\fCyAExpMin}{\textbf{\textcolor{\GqlTokenColor}{\textsc{Min}}}}

\newcommand{\fCyLWild}{\textbf{\textcolor{\GqlTokenColor}{\%}}}
\newcommand{\fCyLNeg}[1]{\bm{\textcolor{\GqlTokenColor}{!\,}}#1}
\newcommand{\fCyLAnd}[2]{#1\bm{\textcolor{\GqlTokenColor}{\,\&\,}}#2}
\newcommand{\fCyLOr}[2]{#1\bm{\textcolor{\GqlTokenColor}{\,\mid\,}}#2}
\newcommand{\fCyExpIsNull}[1]{#1\;\textbf{\textcolor{\GqlTokenColor}{\textsc{isNull}}}}
\newcommand{\fCyExpTrue}{\fTrue}
\newcommand{\fCyExpFalse}{\fFalse}
\newcommand{\fCyExpNull}{\fNull}
\newcommand{\fCyProjection}{{\mu}}

\newcommand{\fCyPatternExp}{\PatClr{\text{$p$}}}
\newcommand{\fCyPattern}{\PatClr{\text{$P$}}}
\newcommand{\fCyLabelExp}{{\text{$l$}}}
\newcommand{\fCyLabel}{{\text{$L$}}}
\newcommand{\fCyPatterns}[1]{\PatClr{\text{$P_{#1}$}}}
\newcommand{\fCyPatternAnd}[2]{{#1}\fCyComma {#2}}

\newcommand{\fCyDirection}{\fCyPatternEdge}
\newcommand{\fCyDir}{\MetaClr{\partial}}

\newcommand{\fCyNode}[1]{\textbf{\text{#1}}}
\newcommand{\fCyTuple}{\mathsf{R}}
\newcommand{\fCyTuples}[1]{\fCyTuple_{#1}}
\newcommand{\fCyParenthesis}[1]{\textbf{\textcolor{\GqlTokenColor}{(}}\,#1\,\textbf{\textcolor{\GqlTokenColor}{)}}}
\newcommand{\fCySqParenthesis}[1]{\textbf{\textcolor{\GqlTokenColor}{[}}\,#1\,\textbf{\textcolor{\GqlTokenColor}{]}}}
\newcommand{\fUse}{\textbf{\textcolor{\GqlTokenColor}{\textsc{Use}}}}
\newcommand{\fUnion}{\textbf{\textcolor{\GqlTokenColor}{\textsc{Union}}}}
\newcommand{\fDstnct}{\textbf{\textcolor{\GqlTokenColor}{\textsc{Distinct}}}}
\newcommand{\fAll}{\textbf{\textcolor{\GqlTokenColor}{\textsc{All}}}}
\newcommand{\fIntersectD}{\textbf{\textcolor{\GqlTokenColor}{\textsc{Intersect Distinct}}}}
\newcommand{\fExceptD}{\textbf{\textcolor{\GqlTokenColor}{\textsc{Except Distinct}}}}

\newcommand{\fIntersectA}{\textbf{\textcolor{\GqlTokenColor}{\textsc{Intersect All}}}}
\newcommand{\fIntersectJ}{\textbf{\textcolor{\GqlTokenColor}{\textsc{Intersect}}}}
\newcommand{\fExceptJ}{\textbf{\textcolor{\GqlTokenColor}{\textsc{Except}}}}
\newcommand{\fExceptA}{\textbf{\textcolor{\GqlTokenColor}{\textsc{Except All}}}}
\newcommand{\fOtherwise}{\textbf{\textcolor{\GqlTokenColor}{\textsc{Otherwise}}}}
\newcommand{\fCyCompOp}{\QryClr{\text{$\circledast$}}}

\newcommand{\FTyProjAs}{\textsc{Prj-Atom-ValExp}}
\newcommand{\FTyProj}{\textsc{Prj-Atom-VarRef}}
\newcommand{\FTyProjList}{\textsc{Prj-Compose}}
\newcommand{\FTyMatchRet}{\textsc{Qry-Match}}
\newcommand{\FTyMatchWhereRet}{\textsc{Qry-Match-Filter}}
\newcommand{\FTyCompQueryLift}{\textsc{CompQry-Lift}}
\newcommand{\FTyCompQuery}{\textsc{CompQry-Compose}}
\newcommand{\FSubRefl}{\textsc{S-Reflexive}}
\newcommand{\FSubTrans}{\textsc{S-Transitive}}
\newcommand{\FSubRefineN}{\textsc{S-Ref-Node}}
\newcommand{\FSubRefineE}{\textsc{S-Ref-Edge}}
\newcommand{\FSubRefineNE}{\textsc{S-Ref-Node-Empty}}
\newcommand{\FSubRefineEE}{\textsc{S-Ref-Edge-Empty}}

\newcommand{\FSubUnionL}{\textsc{S-Union-Left}}
\newcommand{\FSubUnionR}{\textsc{S-Union-Right}}
\newcommand{\FSubUnionElim}{\textsc{S-Union}}
\newcommand{\FSubUnionCon}{\textsc{S-Union-Congruence}}
\newcommand{\FSubList}{\textsc{S-List}}
\newcommand{\FSubAny}{\textsc{S-Any}}
\newcommand{\FSubEmpty}{\textsc{S-Empty}}
\newcommand{\FInterL}{\textsc{Intersect-Left}}
\newcommand{\FInterR}{\textsc{Intersect-Right}}
\newcommand{\FInter}{\textsc{Intersect}}
\newcommand{\FPropEmpty}{\textsc{Prp-Empty}}
\newcommand{\FPropAtom}{\textsc{Prp-Atom}}
\newcommand{\FPropInsert}{\textsc{Prp-Insert}}
\newcommand{\FTyPNodeLabelPropLift}{\textsc{Pat-Node}}
\newcommand{\FTyPNodeLabelPropOpen}{\textsc{Atom-Node-Open}}

\newcommand{\FTyPNodeLabelPropClosed}{\textsc{Atom-Node-Closed}}

\newcommand{\FTySAtomLabelPropFresh}{\textsc{Sch-Atom}}

\newcommand{\FTyPEdgeLabelPropOpen}{\textsc{Atom-Edge-Open}}

\newcommand{\FTyPEdgeLabelPropClosed}{\textsc{Atom-Edge-Closed}}

\newcommand{\FTyPPatRefine}{\textsc{Refine-Closed}}

\newcommand{\FTyPPatRefineOpen}{\textsc{Refine-Open}}
\newcommand{\FTyPPatEdge}{\textsc{Pat-Edge}}
\newcommand{\FTyPPatPattern}{\textsc{Pat-Step}}

\newcommand{\FTyPPatQuantEdge}{\textsc{Pat-Quant-Edge}}
\newcommand{\FTyPPatParenPath}{\textsc{Pat-Paren-Path}}
\newcommand{\FTyPPatQuantPath}{\textsc{Pat-Quant-Path}}
\newcommand{\FSPatAtomNode}{\textsc{Atom-Node}}
\newcommand{\FSPatAtomEdge}{\textsc{Atom-Edge}}

\newcommand{\FSPatExpSingle}{\textsc{PatLst-Lift-Step}}
\newcommand{\FSPatExpDone}{\textsc{PatLst-Lift}}
\newcommand{\FSPatExpAndL}{\textsc{PatLst-Conj-L}}
\newcommand{\FSPatExpAndR}{\textsc{PatLst-Conj-R}}
\newcommand{\FSPatExpAnd}{\textsc{PatLst-Conj}}

\newcommand{\FSSegNode}{\textsc{Path-Node}}
\newcommand{\FSSegEdgeStepL}{\textsc{Path-Edge-L}}
\newcommand{\FSSegEdgeStepR}{\textsc{Path-Edge-R}}
\newcommand{\FSSegEdge}{\textsc{Path-Edge}}

\newcommand{\FSPathConcatStepL}{\textsc{Path-Concat-L}}
\newcommand{\FSPathConcatStepR}{\textsc{Path-Concat-R}}
\newcommand{\FSPathConcat}{\textsc{Path-Concat}}

\newcommand{\FSPathQPathStep}{\textsc{Path-QPath-Step}}
\newcommand{\FSPathQPathEval}{\textsc{Path-QPath}}
\newcommand{\FSQPathInit}{\textsc{QPath-Init}}
\newcommand{\FSQPathIter}{\textsc{QPath-Iter}}
\newcommand{\FSQPathFinish}{\textsc{QPath-Finish}}

\newcommand{\FSLabelEmpty}{\textsc{L-Empty}}
\newcommand{\FSLabelAtom}{\textsc{L-Atom}}
\newcommand{\FSLabelAtomFail}{\textsc{L-Atom-Fail}}
\newcommand{\FSLabelWild}{\textsc{L-Wild}}
\newcommand{\FSLabelWildFail}{\textsc{L-Wild-Fail}}
\newcommand{\FSLabelNeg}{\textsc{L-Neg}}
\newcommand{\FSLabelAnd}{\textsc{L-And}}
\newcommand{\FSLabelOr}{\textsc{L-Or}}
\newcommand{\FSCompQueryLift}{\textsc{CQ-Lift}}
\newcommand{\FSCompQueryL}{\textsc{CQ-L}}
\newcommand{\FSCompQueryR}{\textsc{CQ-R}}
\newcommand{\FSCompQueryUnion}{\textsc{CQ-Union}}

\newcommand{\FSQueryUse}{\textsc{Q-Use}}

\newcommand{\FSQueryMatch}{\textsc{Q-Match}}
\newcommand{\FSQueryMatchDone}{\textsc{Q-Where}}
\newcommand{\FSQueryProject}{\textsc{Q-Return}}
\newcommand{\FTyPPatExpSingle}{\textsc{PatList-Lift}}
\newcommand{\FTyPPatExpAnd}{\textsc{PatList-Conjunction}}

\newcommand{\FTySLWildcard}{\textsc{Sch-Lbl-Wildcard}}
\newcommand{\FTySLAtom}{\textsc{Sch-Lbl-Atom}}
\newcommand{\FTySLEmpty}{\textsc{Sch-Lbl-Empty}}
\newcommand{\FTySLNot}{\textsc{Sch-Lbl-Negation}}
\newcommand{\FTySLAnd}{\textsc{Sch-Lbl-Conjunction}}
\newcommand{\FTySLOr}{\textsc{Sch-Lbl-Disjunction}}
\newcommand{\FTySPAtom}{\textsc{Sch-Prp-Atom}}
\newcommand{\fDbWorld}{\Sigma;\fDbGNames}

\newcommand{\fDbWorldS}{\Sigma}
\newcommand{\fDbCatalog}{\Omega}
\newcommand{\fDbSchema}{\Upsilon}
\newcommand{\fDbRSchema}{\SchClr{\Gamma}}
\newcommand{\fSingletonRef}{\mathds{1}}
\newcommand{\fGroupRef}{\bigstar}
\newcommand{\fDbRSchemas}[1]{\SchClr{\Gamma_{#1}}}
\newcommand{\fDbVarUse}{\mathcal{X}}

\newcommand{\fQuantVars}{\mathcal{Y}}

\newcommand{\fDbNodeType}[1][\fDbGNames]{\TyClr{\bm{\mathsf{N}\langle}}{#1}\TyClr{\bm{\rangle}}}
\newcommand{\fDbEdgeType}[1][\fDbGNames]{\TyClr{\bm{\mathsf{E}\langle}}{#1}\TyClr{\bm{\rangle}}}

\newcommand{\fDbGraphSchema}[1][]{\SchClr{\Psi_{\scalebox{0.6}{$#1$}}}}
\newcommand{\fDbGraphSchemas}[1][]{\overline{\fDbGraphSchema[#1]}}
\newcommand{\fDbNodeSchemaSymbol}{\SchClr{\zeta}}
\newcommand{\fDbEdgeSchemaSymbol}{\SchClr{\xi}}
\newcommand{\fDbElemSchemaSymbol}{\SchClr{\varsigma}}
\newcommand{\fDbNodeSchemaElem}[1][]{\fDbNodeSchemaSymbol_{\SchClr{\scalebox{0.6}{$#1$}}}}
\newcommand{\fDbEdgeSchemaElem}[1][]{\fDbEdgeSchemaSymbol_{\SchClr{\scalebox{0.6}{$#1$}}}}
\newcommand{\fDbElemSchemaElem}[1][]{\fDbElemSchemaSymbol_{\SchClr{\scalebox{0.6}{$#1$}}}}
\newcommand{\fDbNodeSchema}[1][]{\overline{\fDbNodeSchemaSymbol}_{\SchClr{\scalebox{0.6}{$#1$}}}}
\newcommand{\fDbEdgeSchema}[1][]{\overline{\fDbEdgeSchemaSymbol}_{\SchClr{\scalebox{0.6}{$#1$}}}}
\newcommand{\fDbElemSchema}[1][]{\overline{\fDbElemSchemaSymbol}_{\SchClr{\scalebox{0.6}{$#1$}}}}
\newcommand{\fSchLbls}[1]{\fElemAccess{1}{#1}}
\newcommand{\fSchProps}[1]{\fElemAccess{2}{#1}}
\newcommand{\fSchEnd}[2]{\fElemAccess{\the\numexpr#2+2\relax}{#1}}
\newcommand{\fSchDir}[1]{\fElemAccess{5}{#1}}
\newcommand{\fDbPropSchema}[1][]{\SchClr{\Phi_{\scalebox{0.6}{$#1$}}}}
\newcommand{\fDbEdgeDir}[1][]{\MetaClr{\delta_{\scalebox{0.6}{$#1$}}}}
\newcommand{\fDbGNames}[1][]{\textsf{G}_{#1}}
\newcommand{\fDbBindTable}[1][]{\textsc{T}_{#1}}
\newcommand{\fDbRVal}{\omega}
\newcommand{\fDbRVals}[1]{\fDbRVal_{#1}}
\newcommand{\fConformSym}[2]{\models^{\scalebox{0.5}{$#1$}}_{\scalebox{0.6}{$#2$}}}

\newcommand{\fConform}[2]{#1\!\fConformSym{}{}\!#2}

\newcommand{\fConformGraph}[4]{#1\!\fConformSym{#3}{#4}\!#2}

\newcommand{\fCompatU}[2]{#1\!\sim_{\scalebox{0.7}{$\bm{\cup}$}}\!#2}

\newcommand{\fCompatJ}[2]{#1\!\sim_{\scalebox{0.7}{$\bm{\Join}$}}\!#2}

\newcommand{\fCompatC}[2]{#1\!\sim_{\scalebox{0.7}{$\fCyCompOp$}}\!#2}

\newcommand{\fPropJudge}[3]{{#1}\vdash_{\scalebox{0.7}{\textsf{Prp}}}{#2}:{#3}}

\newcommand{\fRuntimeStx}[1]{\hat{#1}}
\NewDocumentCommand{\fExpJudgeParam}{O{\fDbWorld;\fDbRSchema} O{\fCyExp} O{\fSortVar} O{\Box} O{\Diamond} O{\fDbVarUse}}{%
{#1}\vdash^{\scalebox{0.6}{$\hspace{1pt}\!#4\!$}}_{\scalebox{0.6}{$\!#5\!$}}{#2}:{#3}\triangleright{#6}%
}
\NewDocumentCommand{\fExpJudgeParamC}{O{\fCyExp} O{\fSortVar} O{\fDbVarUse}}{%
\fExpJudgeParam[\fDbWorld;\fDbRSchema][#1][#2][\Box][\Diamond][#3]
}

\newcommand{\fSAtomJudge}[3]{{#1}\vdash_{\scalebox{0.7}{\textsf{Atom}}}{#2}\rightsquigarrow{#3}}

\newcommand{\fPatAtomJudge}[3]{{#1}\vdash_{\scalebox{0.7}{\textsf{Atom}}}{#2}:{#3}}

\newcommand{\fPatJudgeM}[4]{{#1}\vdash^{\raisebox{0ex}{\scalebox{0.6}{$#4$}}}_{\raisebox{0.1ex}{\scalebox{0.7}{\textsf{Pat}}}}{#2}:{#3}}
\newcommand{\fPatJudgeP}[4]{{#1}\vdash^{\raisebox{0ex}{\scalebox{0.6}{$#4$}}}_{\raisebox{0.1ex}{\scalebox{0.7}{\textsf{PatStep}}}}{#2}:{#3}}
\newcommand{\fLabelSJudge}[3]{{#1}\vdash_{\scalebox{0.7}{\textsf{Lbl}}}{#2}\rightsquigarrow{#3}}
\newcommand{\fPropSJudge}[3]{{#1}\vdash_{\scalebox{0.7}{\textsf{Prp}}}{#2}\rightsquigarrow{#3}}
\newcommand{\fPatExpJudge}[3]{{#1}\vdash_{\scalebox{0.7}{\textsf{PatLst}}}{#2}:{#3}}

\newcommand{\fVisitEdge}[1][]{\mathcal{W}_{#1}}
\newcommand{\fVisitPath}[1][]{\mathcal{F}_{#1}}

\newcommand{\fLiftQuant}[2]{{#2}\!\uparrow_{\!\scalebox{0.7}{$#1$}}}

\newcommand{\fBaseType}[1]{\lfloor{#1}\rfloor}
\newcommand{\fRefineType}[2]{{#1}\!\!\upharpoonright_{\!#2}}
\newcommand{\fEndpointCond}[3]{\bm{\theta_{\fCyDirection}}\textbf{(}#1,#2,#3\textbf{)}}

\newcommand{\fOrientR}{\rightarrow}
\newcommand{\fOrientL}{\leftarrow}
\newcommand{\fOrientU}{\sim}
\newcommand{\fDirDenote}[1]{\bm{\mathsf{dir}(}{#1}\bm{)}}

\newcommand{\fRefineJudge}[3]{{#1}\vdash_{\scalebox{0.7}{\textsf{Rfn}}}{#2}\rightsquigarrow{#3}}
\newcommand{\fProjJudge}[3]{{#1}\vdash{#2}:{#3}}
\newcommand{\fQueryJudge}[3]{{#1}\vdash{#2}:{#3}}
\newcommand{\fPatConcat}[2]{{#1}\,\circ\,{#2}}
\newcommand{\fCompQueryJudge}[3]{{#1}\vdash_{\fCyCompOp}{#2}:{#3}}

\newcommand{\fSAtomJudgeN}[2]{\fPropGraph\vdash_{\scalebox{0.7}{\textsf{Atom}}}{#1}\fOpBigStep{#2}}
\newcommand{\fSAtomJudgeE}[2]{\fPropGraph\vdash_{\scalebox{0.7}{\textsf{Atom}}}{#1}\fOpBigStep{#2}}

\newcommand{\fSPathJudge}[2]{\fPropGraph\vdash_{\scalebox{0.7}{\textsf{Pat}}}{#1}\fOpStep{#2}}
\newcommand{\fSPathJudgeC}[2]{\fPropGraph\vdash_{\scalebox{0.7}{\textsf{Pat}}}{#1}\fOpStep{#2}}

\newcommand{\fSQPathJudge}[2]{\fPropGraph\vdash_{\scalebox{0.7}{\textsf{QPat}}}{#1}\fOpStep{#2}}
\newcommand{\fOpStep}{\!\rightarrow\!}
\newcommand{\fOpStepStar}{\longrightarrow^{\!*}}
\newcommand{\fOpBigStep}{\Downarrow}
\newcommand{\fValSet}{\mathcal{V}}
\newcommand{\fQuantIt}{\kappa}

\newcommand{\fEval}[4]{#1;\,#2\vdash#3\,\fOpStep\,#4}
\newcommand{\fLabelOp}[3]{#1\vdash_{\fPGLabelFunc}#2\fOpBigStep#3}

\newcommand{\fPEStep}[2]{\fPropGraph\vdash_{\scalebox{0.7}{\textsf{PatLst}}}{#1}\fOpStep#2}
\newcommand{\fPEStepStar}[2]{\fPropGraph\vdash_{\scalebox{0.7}{\textsf{PatLst}}}{#1}\fOpStepStar#2}
\newcommand{\fQStep}[3]{#1\vdash_{\fCyCompOp}#2\fOpStep#3}
\newcommand{\fQLinStep}[3]{#1\vdash_{}#2\fOpStep#3}

\newcommand{\fValType}[2]{#1\mathbin{:}#2}
\newcommand{\fConfigType}[2]{\fDbWorldS\vdash_{\mathsf{cfg}}#1\mathbin{:}#2}

\newcommand{\fPConfigType}[3]{\fDbWorld\vdash_{\mathsf{pcfg}}#1\mathbin{:}#2\Rightarrow#3}
\newcommand{\fWFWorld}{\models\fDbWorldS}
\newcommand{\fAnonErase}[1]{\fAnonDrop{\fAttrAnon}{\fBTDrop{#1}}}
\newcommand{\fBag}[1]{\left\{\mkern-6mu\left\{#1\right\}\mkern-6mu\right\}}
\newcommand{\fSet}[1]{\left\{#1\right\}}
\newcommand{\fList}[1]{[#1]}

\newcommand{\fSize}[1]{\lvert{#1}\rvert}
\newcommand{\fPowerSetSub}[2]{\bm{\mathcal{P}_{#2}(}#1\bm{)}}
\newcommand{\fPowerSet}[1]{\fPowerSetSub{#1}{}}
\newcommand{\fTimes}[2]{{#1}\!\times\!{#2}}
\newcommand{\fTermDef}{\mathrel{\vcentcolon\vcentcolon=}}
\newcommand{\fAssign}{\mathrel{\vcentcolon=}}

\newcommand{\fSubType}[2]{#1\bm{<}\!\textbf{:}\,#2}
\newcommand{\fNSubType}[2]{#1\bm{\nless}\!\textbf{:}\,#2}
\newcommand{\fElemAccess}[2]{\bm{\pi_{#1}(}#2\bm{)}}

\newcommand{\fUniverse}[1]{\mathscr{U}_{#1}}
\newcommand{\fUniverseScale}[1]{\mathscr{U}_{\scalebox{0.6}{$#1$}}}
\newcommand{\fInt}{\TyClr{\mathbb{Z}}}
\newcommand{\fStr}{\TyClr{\mathbb{S}}}
\newcommand{\fBool}{\TyClr{\mathbb{B}}}

\newcommand{\fAny}{\TyClr{\textbf{\textsf{Any}}}}
\newcommand{\fEmpty}{\TyClr{\bm{\bot}}}
\newcommand{\fEmptyT}[1]{{#1}^{\fEmpty}}
\newcommand{\fUpToNull}[2]{{#1}\simeq_{\scalebox{0.7}{$\fNullable$}}{#2}}
\newcommand{\fNullable}{\TyClr{\textbf{?}}}
\newcommand{\fListT}{\TyClr{\textbf{\textsf{List}}}}
\newcommand{\fListConcat}[2]{#1::#2}

\newcommand{\fAttr}{\mathcal{A}}
\newcommand{\fAttrAnon}{\mathcal{A}_\mathsf{Anon}}
\newcommand{\fMSort}[1]{\TyClr{\mathcal{T}_{#1}}}
\newcommand{\fMSortPG}{\fMSort{0}}
\newcommand{\fMSortGQL}{\fMSort{1}}
\newcommand{\fMSortGQLN}{\fMSort{2}}

\newcommand{\fSortVar}{\TyClr{\tau}}
\newcommand{\fSortVars}[1]{\TyClr{\tau_{\scalebox{0.65}{$#1$}}}}

\newcommand{\fTrue}{\textbf{\textsc{True}}}
\newcommand{\fFalse}{\textbf{\textsc{False}}}
\newcommand{\fNull}{\textbf{\textsc{Null}}}

\newcommand{\fSchemaDUnion}[2]{#1\sqcup#2}
\newcommand{\fSchemaUnion}[2]{#1\cup#2}
\newcommand{\fSchemaJoin}[2]{#1\,{\scalebox{0.8}{$\Join$}}\,#2}

\newcommand{\fTupleJoin}[2]{{#1}\!\Join\!{#2}}
\newcommand{\fBagVar}[1][]{\textsc{B}_{#1}}
\newcommand{\fBTJoinOp}{\bm{\Join}}

\newcommand{\fBTTrailProd}[2]{{#1}\bm{\otimes}{#2}}
\newcommand{\fBTUnit}{\mathbf{1}_{\mathsf{T}}}

\newcommand{\fPathTable}[1][]{\textsc{P}_{#1}}
\newcommand{\fPTNodeLift}[3]{\bm{\eta}_{\scalebox{0.8}{$#1, #2$}}\textbf{(}#3\textbf{)}}
\newcommand{\fPTEdgeLift}[3]{\bm{\eta}_{\scalebox{0.8}{$#1, #2$}}\textbf{(}#3\textbf{)}}
\newcommand{\fEdgeEnds}[3]{\bm{\mathsf{ends}}_{#1,#2}\textbf{(}#3\textbf{)}}

\newcommand{\fPTJoinOp}{\bm{\diamond}}
\newcommand{\fPTJoinRec}[1]{\bm{\diamond}_{\scalebox{0.85}{$#1$}}}
\newcommand{\fPTDrop}[1]{\bm{\pi}_{\mathsf{T}}\textbf{(}#1\textbf{)}}
\newcommand{\fBTDrop}[1]{\bm{\pi}_{\mathsf{B}}\textbf{(}#1\textbf{)}}
\newcommand{\fQFrame}[3]{\bm{\lbrack}#1\bm{\rbrack}^{#3}_{#2}}
\newcommand{\fQZero}[2]{\mathbf{0}_{#1}^{#2}}

\newcommand{\fQRecExt}{\bm{\star}_{\fCyQuantifier}}

\newcommand{\fQJoinOp}[1]{\bm{\diamond}_{#1}}
\NewDocumentCommand{\fQState}{O{\fPathTable} O{\fQuantVars} O{\fCyQuantifier} m m m}{%
  \fQFrame{#1}{#2}{#3}\bm{\langle}#4,\,#5,\,#6\bm{\rangle}%
}
\newcommand{\fAnonDrop}[2]{\bm{\pi}_{\neg #1}\textbf{(}#2\textbf{)}}
\newcommand{\fVars}[1]{\bm{\mathsf{vars}(}#1\bm{)}}
\newcommand{\fNormOf}[1]{\bm{\mathsf{norm}(}#1\bm{)}}
\newcommand{\fAnonFill}[1]{\bm{\mathsf{anon}(}#1\bm{)}}

\newcommand{\fDom}[1]{\text{$\bm{\mathsf{dom}(}#1\bm{)}$}}

\newcommand{\fQLo}[1]{\text{$\bm{\mathsf{lo}(}#1\bm{)}$}}
\newcommand{\fQHi}[1]{\text{$\bm{\mathsf{hi}_{\fPropGraph}(}#1\bm{)}$}}
\newcommand{\fDir}[1]{\text{$\bm{\mathsf{ornt}(}#1\bm{)}$}}
\newcommand{\fVar}[1]{\text{$\bm{\mathsf{var}(}#1\bm{)}$}}
\newcommand{\fTail}[1]{\text{$\bm{\mathsf{tail}(}#1\bm{)}$}}

\newcommand{\fLbl}[1]{\text{$\bm{\mathsf{lbl}(}#1\bm{)}$}}
\newcommand{\fProp}[1]{\text{$\bm{\mathsf{prp}(}#1\bm{)}$}}
\newcommand{\fMinE}[1]{\text{$\bm{\mathsf{lo}^{\scalebox{0.7}{$\mathsf{\#E}$}}(}#1\bm{)}$}}
\newcommand{\fMinN}[1]{\text{$\bm{\mathsf{lo}^{\scalebox{0.7}{$\mathsf{\#N}$}}(}#1\bm{)}$}}
\newcommand{\fSup}[1]{\text{$\bm{\mathsf{supp}(}#1\bm{)}$}}

\newcommand{\FigCaptionGqlCalculus}{%
The formalized \lang calculus with the relevant
\textcolor{\CyIsoRefColor}{$\mathrm{ISO}$} sections and optional features (\textcolor{\CyIsoRefColor}{\isoFeatureStyle{GXXX}}), described in gray.
\(
\fCyVar\!\in\!\fAttr;\,
i,j\!\in\!\,\mathbb{N};\,
\fUniverseScale{\fBool}\!=\!\fSet{\fTrue,\fFalse}\!;
\)
and
$\fCyExpConst{}\!\in\!\fUniverse{\fSortVar}$, $\fSortVar\!\in\!\fMSortPG.$
We use the metavariable $\fCyPatternAtom$ to denote $\fCyPatternNode$ and $\fCyPatternEdge$.
\label{fig:gql-calculus}
}
\newcommand{\FigCaptionGqlMetafunctions}{%
Metafunctions for parsing \lang's graph pattern matching
fragment. \emph{Atom Descriptor} metafunctions $\bm{\mathsf{var}}$, $\bm{\mathsf{lbl}}$ and 
$\bm{\mathsf{prp}}$; lift to pattern atoms $\fCyPatternNode$ and $\fCyPatternEdge$ via 
their descriptors ($\fCyDscSyn$). Only the quantifier lower bounds 
$\fQLo{\fCyQuantifier}$ are given since interpreting upper bounds requires a 
specific property graph instance (for \pathModeTrail mode).
\label{fig:gql-metafunctions}
}
\newcommand{\FigCaptionPropGraphInstance}{%
An example of a property graph instance that contains:
nodes (\raisebox{0.2ex}{\circledCaption{1}}),
directed edges
(\raisebox{0.2ex}{\circledCaption{2}}),
undirected edges
(\raisebox{0.2ex}{\circledCaption{3}}),
node property maps
(\raisebox{0.2ex}{\circledCaption{4}}),
and labeled edges
(\raisebox{0.2ex}{\circledCaption{5}}).
\label{fig:pgm-instance}
}
\newcommand{\FigCaptionGqlExample}{%
An example of graph pattern matching with \lang.
\label{fig:gql-example}
}

\newcommand{\FigCaptionGqlRuntime}{%
Runtime values, execution constructs, runtime syntax, and runtime
metafunctions used in our semantics formalization. Runtime terms are 
shown with a hat (\eg $\fRuntimeStx{\fCyPattern}$) to differentiate 
them from the source terms of Figure~\ref{fig:gql-calculus}. $\fPathTable[A]$ and
$\fVisitPath$ are the accumulator and the frontier of a quantified
path, and $\fQuantIt$ counts its repetitions.
\label{fig:runtime-entities}
}

\setcopyright{cc}
\setcctype{by}
\acmDOI{10.1145/3839459}
\acmYear{2026}
\acmJournal{PACMPL}
\acmVolume{10}
\acmNumber{OOPSLA2}
\acmArticle{327}
\acmMonth{10}
\acmSubmissionID{oopslab26main-p337-p}
\received{2026-03-18}
\received[accepted]{2026-06-10}

\begin{document}

\title{\Title}

\author{Aditya Thimmaiah}
\orcid{0009-0002-1917-7386}
\affiliation{%
  \institution{University of Texas at Austin}
  \city{Austin}
  \country{USA}
}
\email{auditt@utexas.edu}

\author{Tong-Nong Lin}
\orcid{0009-0005-0403-654X}
\affiliation{%
  \institution{University of Texas at Austin}
  \city{Austin}
  \country{USA}
}
\email{tong-nong@utexas.edu}

\author{Milos Gligoric}
\orcid{0000-0002-5894-7649}
\affiliation{%
  \institution{University of Texas at Austin}
  \city{Austin}
  \country{USA}
}
\email{gligoric@utexas.edu}

\begin{abstract}
  Research and development of graph query languages
  has been gaining traction with the increase in popularity of graph
  databases, specifically due to the flexible schema
  and other rich semantic offerings of the latter's most common underlying
  data model: the property graph. 
  This has culminated in the standardization of the ISO Graph
  Query Language (\lang) as \iso in 2024,
  the first international standard for property graph-based
  graph query languages.
  However, \iso codifies its semantics informally across
  \isoPageCount pages of prose, making it difficult to
  formally reason about the standard or for a standard-faithful implementation.
  
  Existing formalizations are not adequate because they
  either: (1)~significantly reduce the semantic complexity by omitting
  bag semantics, schemas, and composite queries on multiple graphs;
  (2)~or significantly reduce the syntactic complexity by only considering isolated
  fragments such as pattern-matching, leaving the full query pipeline
  unformalized. Yet it is these semantic--syntactic features that 
  make formalizing \lang non-trivial.

  We present \tool, the first mechanized, small-step operational
  semantics for a substantial read-only fragment of \lang that
  is grounded in the \iso standard. Our formalization
  models multi-graph property graphs with mixed edge directionality and
  supports a large fraction of \lang pattern constructs: quantified
  paths and edges, directional and undirected matching, label
  expressions, pattern lists, and composite queries.
  The semantics is supported by a schema-aware type system
  that refines variable types via closed-graph schemas, tracks
  nullability, supports multiple composite
  query operators, and models quantified-path bindings with list types.
  We prove that our type system is sound, ensuring an end-to-end
  guarantee of well-formed queries yielding results that conform
  to their declared schemas. \tool{} provides the first bridge
  between \lang's informal specification and a mechanized implementation,
  enabling formal reasoning about correctness.
\end{abstract}

  


\maketitle

% 2mm gap between table and captions
\captionsetup[table]{skip=1mm}
% 2mm gap between figure and captions
\captionsetup[figure]{skip=2mm}
% 2mm gap between figure and captions
\captionsetup[subfigure]{skip=2mm}
\setlength{\columnsep}{2mm}%
\setlength{\textfloatsep}{8pt plus 1pt minus 1pt}
\setlength{\intextsep}{6pt plus 1pt minus 1pt}
\setlength{\parskip}{4pt}

\section{Introduction}
\label{sec:intro}

Graph databases have become foundational for modern data management
and power many applications from social networks and fraud
detection to drug discovery and knowledge 
graphs~\cite{angles2018propertygraph, sakr2021future}.
The Graph Query Language (\lang)~\cite{gheerbrant2025gql,deutsch2022graph}, 
standardized in 2024 as \iso, is the first international standard 
of graph query languages that query property graphs~\cite{angles2018propertygraph}. 
The standard unified language features
from other graph query languages such as
\cypher~\cite{francis2018cypher}, PGQL~\cite{van2016pgql}, 
and G-CORE~\cite{angles2018gcore}.
It defines the semantics specification which implementations 
are expected to conform.

However, the semantics is specified informally, spanning more than 
600 pages of prose that interleave data flow, typing,
query evaluation pipelines, and pattern-matching, \etc
making formally reasoning about \lang, such as verifying implementations or 
proving query equivalences, difficult.

Existing formalizations do not address this adequately.
\citet{Francis2023gpc} distilled \lang's pattern matching into a Graph
Pattern Calculus (GPC) with typing rules and a denotational semantics,
and \citet{gheerbrant2025gql} defined Core \lang and Core PGQ to study
its expressive power, but both leave out graph schemas~\cite{angles2023pgschema} 
and bag semantics.
\citet{ye2025flexible} formalized a gradually typed calculus for \lang
path patterns in isolation from the rest of the language, implementing 
it in Python.
No mechanized semantics exists for a fragment of \lang
rich enough to model the interaction between graph schemas, 
bag semantics, and composite queries, \etc 
some of the key unique features of \lang.

We present \tool{}, the first mechanized, executable, small-step operational
semantics for a large read-only fragment of \lang, together with
a schema-aware type system and a type soundness theorem, all
machine-checked in \lean~\cite{deMouraLean2021}; providing a foundation 
for rigorous formal reasoning.

We handle the complexity of \lang's semantics by \emph{stratifying} 
our static typing rules and small-step operational semantics into 
compositional layers---expressions, patterns, and queries. This allows
us to capture the key non-trivial features of \lang.
First, graph element types in \lang are graph-scoped and refined 
by graph schemas, allowing for label and property constraints 
to enable statically detecting patterns that cannot match, but requiring a 
typing system that propagates refined types across the query pipeline.
Second, \lang's three-valued semantics identifies null with the truth 
value \emph{Unknown}, requiring every operator to participate in 
Kleene-style null propagation and every type to track nullability.
Third, \lang's \emph{path modes} and \emph{match modes} impose 
additional constraints on pattern matching such as the \pathModeTrail{} 
path mode which constrains a pattern-match by forbidding repeated edges, 
thus breaking naïve path composition. We
handle this by introducing auxiliary \emph{execution constructs} besides those
specified in the standard, which carry the edges visited during pattern-matching 
alongside the matched elements, making it trivial to enforce edge-disjointness.
Fourth, quantified path patterns iterate over graph structure and bind
their variables to sequences rather than single elements. We compute
their repetitions with a frontier.
Finally, composite queries combine bags whose records carry the same
attributes but not necessarily the same types, yielding heterogeneous
unions that the type system must track to formalize well-formedness of queries.

\smallskip
\noindent
We make the following contributions:

\begin{itemize}[leftmargin=1.5em]
\item We define a \textbf{formal calculus} (\S\ref{sec:prelims-gql}) for a rich read-only
  fragment of \lang, covering graph resolution, pattern matching,
  filtering, projection, and composite queries over multiple graphs.

\item We develop the first \textbf{schema-aware type system} (\S\ref{sec:type-system}) for \lang,
  supporting graph-scoped refinement, nullability tracking,
  heterogeneous unions, and list-typed quantified path bindings.

\item We give the first \textbf{small-step operational semantics} (\S\ref{sec:semantics-operational}) for
  \lang, capturing \pathModeTrail-aware pattern matching,
  quantified patterns, nulls, and composite queries; all under bag semantics.

\item We prove a \textbf{layered type soundness theorem} (\S\ref{sec:metatheory}) guaranteeing that well-formed queries
(on complete evaluation) produce results that conform to the
types assigned to them statically.

\item We \textbf{mechanize} (\S\ref{sec:mechanization}) our semantics in \lean{}, and validate it on queries of the Linked Data Benchmark 
  Council's (LDBC) Social Network Benchmark (SNB) Interactive v2 workload~\cite{ldbc-snb}.
\end{itemize}

%\smallskip
% Every component of \tool{} is derived from the \iso standard.
% %  yet
% % restructured into a compositional and mechanized framework.
% %
% Our formalization provides a foundation for reasoning about
% correctness of \lang implementations, validating query rewrites, and
% guiding future evolution of the standard.

%% \noindent
%% The paper is structured as follows:
%% Section~\ref{sec:prelims} discusses the preliminaries.
%% Sections~\ref{sec:type-system} and~\ref{sec:typing-rules} develop the
%% type system and well-formedness rules for \lang queries.
%% Section~\ref{sec:semantics-operational} formalizes the small-step 
%% operational semantics for \lang, while
%% Section~\ref{sec:metatheory} uses it for proving metatheoretic results.
% Section~\ref{sec:related-work} discusses related work, and
% Section~\ref{sec:conclusion} concludes.

\noindent
All references to the \lang standard are with respect to the \iso First Edition 
2024, which we henceforth refer to as just \isoShort. 
We also abbreviate normative \emph{syntax rules} from the standard as 
\textcolor{\CyIsoRefColor}{$\mathrm{SR}$} when referring to them. 
\tool{} is available at \url{https://github.com/EngineeringSoftware/mgql}.

\section{Preliminaries}
\label{sec:prelims}

% This section establishes the formal foundations on which the rest of
% the paper builds.
% %
We now formally introduce the two core components our 
work is built around: the property graph data model (\S\ref{sec:prelims-pgm}) and the 
\lang graph query language (\S\ref{sec:prelims-gql}) that operates on them.
%
% We first define the property graph data model , 
% followed by the \lang calculus fragment 
% that we consider (\S\ref{sec:prelims-gql}).
%
% Together they establish the vocabulary and notation used
% throughout the remainder of the paper.

\begin{myNotation}[Metafunctions and Shorthands]\label{not:metafunc}
  We use the metafunctions $\fDom{\cdot}$ and $\fSup{\cdot}$ to denote the
  \emph{domain} and the \emph{support} of a function; the uppercase calligraphic 
  letters (\eg $\mathcal{A},\,\mathcal{B},\,\mathcal{C}$) as well as accents 
  (\eg \smash{$\overline{L},\,\overline{\zeta}$}\,) to denote sets; and when convenient, 
  the corresponding unaccented symbols denote elements
  ranging over those sets (\eg \smash{$\zeta\in\overline{\zeta}$}\,).
  For a set $\mathcal{S}$ and $m \in \mathbb{N}$,
  $\fPowerSetSub{\mathcal{S}}{\leq m}$ denotes the set of subsets of
  $\mathcal{S}$ whose cardinality is at most $m$, \ie
  \smash{\(
  \fPowerSetSub{\mathcal{S}}{\leq m}
  \triangleq
  \{\, \mathcal{C} \subseteq \mathcal{S} \mid \fSize{\mathcal{C}} \le m \,\}.
  \)}
  The metafunction $\fElemAccess{k}{\cdot}$ for a natural number $k$, denotes access to the
  \smash{$k^{\text{th}}$} component of a tuple. When applied to a set or a bag of tuples,
  it lifts pointwise and returns the set or bag of the \smash{$k^{\text{th}}$}
  components of those tuples, respectively. \qed
\end{myNotation}

% subsections
\subsection{The Property Graph Data Model}
\label{sec:prelims-pgm}

The property graph data
model~\cite{angles2018propertygraph, rodriguez2010constructions} is
richer than the edge-labeled graphs typically used in theoretical work
on regular path queries~\cite{barcelo2013querying}: nodes and edges
carry \emph{labels} describing their domain role, \emph{properties}
stored as key--value pairs, and \emph{directionality} markers that
distinguish directed from undirected relationships.
This additional structure makes \lang's type system nontrivial with
label expressions, property constraints, join conditions, \etc all
inspecting graph structure absent in simpler models and abstractions.
We formalize it in full without reducing to a minimal abstraction.
%% ---and why we formalize it in full rather than reducing to a
%% minimal graph abstraction.
%
%
%% Our type system exploits closed-graph schemas for static refinement,
%% a capability absent from all prior graph-query formalizations known to
%% us.
%
%% Edges also carry a descriptor that identifies whether they are
%% directed or undirected.
%% %
%% A property graph may enforce a schema that constrains the
%% admissibility of nodes/edges in which case it is referred to
%% as \emph{closed} while a schemaless graph is said to be
%% \emph{open}.

\begin{wrapfigure}{r}{0.38\linewidth}%
  \vspace{-2mm}%
  \centering%
  \resizebox{0.85\linewidth}{!}{\begin{tikzpicture}[
  font=\footnotesize,
  >=Latex,
  node distance=15mm and 20mm,
  vertex/.style={
    circle,
    draw,
    thick,
    white,
    minimum size=6mm,
    inner sep=0pt,
    align=center
  },
  vertexD/.style={
    circle,
    draw,
    thick,
    white,
    minimum size=5.5mm,
    inner sep=0pt,
    align=center
  },
  numberIndicate/.style={
    circle,
    draw=white,
    fill=black,
    text=white,
    thick,
    font=\small,
    minimum size=3.5mm,
    inner sep=0pt,
    align=center
  },
  person/.style={vertex, fill=blue!20},
  company/.style={vertex, fill=orange!20},
  proj/.style={vertex, fill=purple!20},
  personD/.style={vertexD, fill=blue!20},
  companyD/.style={vertexD, fill=orange!20},
  projD/.style={vertexD, fill=purple!20},
  indicate/.style={numberIndicate},
  vname/.style={font=\bfseries},
  vlabel/.style={font=\scriptsize, yshift=-5mm},
  vlabelD/.style={font=\scriptsize, yshift=-4mm},
  ilabel/.style={font=\normalsize, xshift=-0.2pt},
  ilabel2/.style={font=\normalsize, xshift=-0.2pt, yshift=-1pt},
  callout/.style={
    draw=black!80,
    rounded corners=1pt,
    %thin,
    fill=black!6,
    font=\scriptsize,
    inner sep=2pt,
    align=left
  },
  rel/.style={
    -{Latex[length=2mm,width=1.6mm]},
    thick
  },
  relUDir/.style={
    thick
  },
  rlabel/.style={
    fill=white,
    inner sep=1.2pt,
    font=\scriptsize
  },
  frame/.style={draw, rounded corners=3pt, thin, inner sep=7pt}
]

% Persons
\node[person] (a) at (0.5,-0.2) {};
\node[vname] at (a) {a};
\node[vlabel] at (a) {};

\node[person] (o) at (2,1.0) {};
\node[vname] at (o) {o};
\node[vlabel] at (o) {};

\node[person] (c) at (2,-1.4) {};
\node[vname] at (c) {c};
\node[vlabel] at (c) {};

\node[person] (z) at (3.5,-0.2) {};
\node[vname] at (z) {z};
\node[vlabel] at (z) {};

\node[person] (e) at (4.0,-2.0) {};
\node[vname] at (e) {e};
\node[vlabel] at (e) {};

% Companies
\node[company] (x) at (0.0,-2.0) {};
\node[vname] at (x) {x};
\node[vlabel] at (x) {};

\node[company] (u) at (4.0,1.6) {};
\node[vname] at (u) {u};
\node[vlabel] at (u) {};

% Project
\node[proj] (s) at (0.0,1.6) {};
\node[vname] at (s) {s};
\node[vlabel] at (s) {};

% Node types
\node[vlabelD] (spacing) at (0,-2.8) {};
\node[personD] (f) at (0,-3.7) {};
\node[vlabelD] (flab) at (f) {\footnotesize\textbf{\texttt{Person}}};
\node[companyD] (n) at (2,-3.7) {};
\node[vlabelD] (nlab) at (n) {\footnotesize\textbf{\texttt{Company}}};
\node[projD] (q) at (4,-3.7) {};
\node[vlabelD] (qlab) at (q) {\footnotesize\textbf{\texttt{Project}}};
\node[
  draw,
  thick,
  rectangle,
  rounded corners=1pt,
  fit=(spacing)(f)(n)(q)(flab)(nlab)(qlab),
  inner sep=-1.5pt,
] (infoBox) {};

\node[
  draw,
  thick,
  rounded corners=0.5pt,
  fill=white,
  inner sep=2pt,
] at (infoBox.north) {\footnotesize \textbf{Property Graph Node Labels}};

% KNOWS edges
\draw[rel] (a) to[bend left=12] node[rlabel, midway] {\texttt{KNOWS}} (o);
\draw[rel] (o) to[bend left=12] node[rlabel, midway] {\texttt{KNOWS}} (z);
\draw[rel] (a) to[bend right=12] node[rlabel, midway] {\texttt{KNOWS}} (c);
\draw[rel] (c) to[bend right=14] node[rlabel, midway] {\texttt{KNOWS}} (z);
\draw[rel] (c) to[bend right=0] node[rlabel, midway] {\texttt{KNOWS}} (o);
%\draw[rel] (o) to[bend right=30] node[rlabel, midway] {\texttt{KNOWS}} (c);
\draw[rel] (c) to[bend right=12] node[rlabel, midway] {\texttt{KNOWS}} (e);
\draw[relUDir] (z) to[bend left=12] node[rlabel, midway, xshift=1pt, yshift=6pt] (uedge) {\texttt{KNOWS}} (e);

% WORKS edges
\draw[rel] (a)  to[bend right=12] node[rlabel, midway, yshift=2pt] {\texttt{WORKS}} (x);
\draw[rel] (c) to[bend left=14] node[rlabel, midway, xshift=4pt, yshift=1pt] {\texttt{WORKS}} (x);
\draw[rel] (o) -- node[rlabel, midway, xshift=-4pt, yshift=-1pt] {\texttt{WORKS}} (u);
\draw[rel] (z)  to[bend right=14] node[rlabel, midway, xshift=-1pt, yshift=-4pt] (workedge) {\texttt{WORKS}} (u);

\draw[rel] (o) -- node[rlabel, midway, xshift=3pt, yshift=-1pt] {\texttt{LEADS}} (s);

% FUNDS edges
\draw[rel] (u) to[bend right=20] node[rlabel, midway] {\texttt{FUNDS}} (s);
\draw[rel] (x) to[bend left=20] node[rlabel, midway] (dedge) {\texttt{FUNDS}} (s);

% INVESTS edges
\draw[rel] (e) to[bend right=20] node[rlabel, midway] {\texttt{INVESTS}} (u);
\draw[rel] (e) to[bend left=20] node[rlabel, midway] (invstEdge) {\texttt{INVESTS}} (x);

% Properties
\node[callout, anchor=south] (yprop) at ($(u.north)+(-0.2,0.47)$)
{\texttt{name: "Cypher"}\\ \texttt{location: "Earth"}};
\draw[-{Latex[length=2mm,width=1.2mm]}, thin, draw=black!60] ($(yprop.south)+(0.2,0)$) to[bend left=1] node[midway] (yproparrow) {} ($(u.north)+(0,0.1)$);

\node[callout, anchor=south] (pprop) at ($(s.north)+(0.2,0.47)$)
{\texttt{name: "Atlas"}\\ \texttt{projectID: 451}};
\draw[-{Latex[length=2mm,width=1.2mm]}, thin, draw=black!60] ($(pprop.south)+(-0.2,0)$) -- ($(s.north)+(0,0.1)$);

\node[indicate, above=0.2mm] at (o.north) {1};
\node[indicate, below=7pt] at (dedge.north) {2};
\node[indicate, below left=8pt and -4pt of uedge] at (uedge.north) {3};
\node[indicate, below left=-3pt and 7.5pt of yprop] at (yproparrow.east) {4};
\node[indicate, above=0.2mm] at (invstEdge.north) {5};

%% \node[indicate, above right=5mm and 1cm of workedge] (nodeInfo) at (workedge.east) {1};
%% \node[indicate, below=4mm] (dEdgeInfo) at (nodeInfo.south) {2};
%% \node[indicate, below=4mm] (uEdgeInfo) at (dEdgeInfo.south) {3};
%% \node[indicate, below=4mm] (propInfo) at (uEdgeInfo.south) {4};
%% \node[indicate, below=4mm] (labelInfo) at (propInfo.south) {5};

%% \node[ilabel, anchor=west] at (nodeInfo.east) {Node};
%% \node[ilabel2, anchor=west] at (dEdgeInfo.east) {Directed Edge};
%% \node[ilabel2, anchor=west] at (uEdgeInfo.east) {Undirected Edge};
%% \node[ilabel2, anchor=west] at (propInfo.east) {Property Map};
%% \node[ilabel2, anchor=west] at (labelInfo.east) {Edge Label};

\end{tikzpicture}}%
  \vspace{1pt}%
  \caption{\FigCaptionPropGraphInstance}%
  \vspace{-6mm}%
\end{wrapfigure}%
Figure~\ref{fig:pgm-instance} shows an example concrete property graph
instance. The graph contains eight nodes carrying one of three node
labels: five \texttt{Person} nodes
\(%
\{
  \circledColor{\textbf{a}}{blue!20}{black}{1}{0ex},
  \circledColor{\textbf{c}}{blue!20}{black}{1}{0ex},
  \circledColor{\textbf{e}}{blue!20}{black}{1}{0ex},
  \circledColor{\textbf{o}}{blue!20}{black}{1}{0ex},
  \circledColor{\textbf{z}}{blue!20}{black}{1}{0ex}
\},
\)
two \texttt{Company} nodes
\(%
\{
  \circledColor{\textbf{x}}{orange!20}{black}{1}{0ex},
  \circledColor{\textbf{u}}{orange!20}{black}{1}{0ex}
\}
\),
and one \texttt{Project} node
\(%
\{
  \circledColor{\textbf{s}}{purple!20}{black}{1}{0ex}
\}.
\)
Edges carry one of five labels: \texttt{KNOWS},
\texttt{WORKS}, \texttt{LEADS}, \texttt{FUNDS}, and
\texttt{INVESTS}, forming a multigraph with parallel edges
and cycles.
The graph is mixed and contains directed 
edges (\eg the \texttt{WORKS} edge in
\!\tikzGEdge{\textbf{o}}{blue!20}{\textbf{u}}{orange!20}{WORKS})
as well as undirected 
edges (\eg the \texttt{KNOWS} edge in
\!\tikzGEdgeUDir{\textbf{e}}{blue!20}{\textbf{z}}{blue!20}{KNOWS}).
Edge labels denote the semantic role of the relationship connecting
two nodes, and two nodes may be connected by more than one edge
(multigraphs).
Many nodes also carry property maps, like the \texttt{Project}
node \(\circledColor{\textbf{s}}{purple!20}{black}{1}{0ex}\):
\(\fSet{\texttt{name}\mapsto\texttt{``Atlas''},\;
  \texttt{projectID}\mapsto 451}\).
These properties illustrate the key feature of property graphs: nodes
and edges may carry arbitrary and finite (possibly empty) key--value 
maps whose values range over primitive domains such as integers, 
strings, and booleans; they are used for capturing the
attributes of the data being modeled.

\MyPara{Graph Instance}%
We assume a multi-sorted universe of values $\fUniverse{}$\ given as the
disjoint union of sub-universes of the types
$\fMSortPG \triangleq \fInt \mid \fStr \mid \fBool$;
with $\fInt$, $\fStr$, and $\fBool$ denoting the universes of integers, strings,
and booleans, respectively.
We also assume an infinite universe of named identifiers $\fAttr.$
A \emph{property map}
\(
\,\fPGProp\!: \fAttr\rightharpoonup\fUniverse{}\,
\)
is a finite partial map from elements (\emph{keys}) in $\fAttr$ to those (\emph{values})
in $\fUniverse{}.$

\begin{myDef}[Property Graph Instance]\label{def:data-model}%
  A \emph{property graph} $\fPropGraph$ is a tuple
  \(
  (\fPGNode, \fPGEdge, \fPGDirFunc, 
   \fPGSrcDstFunc, \fPGLabelFunc, \fPGPropFunc)
  \)
  where:
  \begin{enumerate}[nosep, leftmargin=7mm, labelsep=0.4em]
    \item $\fPGNode$ and $\fPGEdge$ are finite sets of \emph{nodes} and \emph{edges},
          respectively.                
    \item \smash{\(\fPGDirFunc\!: \fPGEdge \to \{\!\fPGDir\!, \fPGUDir\!\},\)} gives the
          \emph{directionality} of edges; whether directed \smash{($\!\fPGDir\!$)} or undirected
          \smash{($\!\fPGUDir\!$)}.
    \item \(\fPGEdge=\fPGEdgeDir\sqcup\fPGEdgeUDir,\) the set of edges
          $\fPGEdge$ can be partitioned into sets $\fPGEdgeDir$
          and $\fPGEdgeUDir$ based on directionality.
    \item \(
          \fPGSrcDstFunc\!: \fPGEdgeDir\sqcup\fPGEdgeUDir \to
          (\fTimes{\fPGNode}{\fPGNode}) \sqcup \binom{\fPGNode}{2},
          \)
          maps edges to their endpoint nodes.
          For $\fPGEdgeElem\in\fPGEdgeDir$,
          \(
          \fPGSrcDstFunc(\fPGEdgeElem)\!=\!(\fPGNodeElems{1}, \fPGNodeElems{2})
          \in\fTimes{\fPGNode}{\fPGNode}
          \)
          is an \emph{ordered pair} with $\fPGNodeElems{1}$ and $\fPGNodeElems{2}$
          referred to as its \emph{source} and \emph{target} nodes.
          For $\fPGEdgeElem\in\fPGEdgeUDir$,
          \(
          \,\fPGSrcDstFunc(\fPGEdgeElem)=\fSet{\fPGNodeElems{1}, \fPGNodeElems{2}}
          \in\binom{\fPGNode}{2}
          \)
          is an \emph{unordered pair}.
          %
          %% For directed edges $\fPGEdgeElem\in\fPGEdgeDir,$ the nodes
          %% $\fPGNodeElems{1}$ and $\fPGNodeElems{2}$ are referred
          %% to as the \emph{source} and \emph{target} of $\fPGEdgeElem$,
          %% respectively.
          %
          Self-loops are permitted, \ie $\fPGNodeElems{1}=\fPGNodeElems{2}$.
    \item \(\fPGLabelFunc\!: \fPGNode\!\cup\!\fPGEdge\!\to\!\fPowerSetSub{\fAttr}{\leq m},\)
          $m\!\in\!\mathbb{N}$;
          maps nodes/edges to finite (maybe empty) sets of
          \emph{labels} from $\fAttr$.
    \item \(\fPGPropFunc\!: \fPGNode\cup\fPGEdge \to \fPGProp,\)
          maps nodes/edges to their properties. $\fPGProp$ may be empty,
          \ie $\fDom{\fPGPropFunc(\cdot)}=\varnothing$.
  \end{enumerate}
\end{myDef}

\MyPara{Graph Schemas}%
The \lang standard permits property graphs to be associated with
\emph{graph schemas}~\cite{angles2023pgschema}, which constrain
the admissible labels and properties on nodes and edges (\isoSection{\isoGType}).
%
%% These schemas serve as the primary static mechanism for constraining
%% the structure of property graph instances and form the basis for the
%% typing discipline developed in Section~\ref{sec:type-system}.

\begin{myDef}[Property, Node, Edge, and Graph Schemas]\label{def:prop-graph-schemas}
  Let $\fCyNodeLabels\subset\fAttr$ be a set of labels.
  \begin{enumerate}[nosep, leftmargin=8mm, labelsep=0.4em]
  \item[\enumParen{$\fDbPropSchema$}] A \emph{property schema} (\isoSectionShort{\isoPropertyType})
    \(
    \,\fDbPropSchema\!:\fAttr\rightharpoonup\fMSortPG\,
    \)
    is a finite partial, mapping keys to scalar types.
  \item[\enumParen{$\hspace{0.8pt}\fDbNodeSchemaElem\hspace{0.8pt}$}] A \emph{node schema} (\isoSectionShort{\isoNodeType})
    $\,\fDbNodeSchemaElem\!\triangleq\!(\fCyNodeLabels, \fDbPropSchema)$ 
    is a pair of a set of labels $\fCyNodeLabels$ and a property schema $\fDbPropSchema$.
  \item[\enumParen{$\fDbEdgeSchemaElem\,$}] An \emph{edge schema} (\isoSectionShort{\isoEdgeType})
    \(
    \,\fDbEdgeSchemaElem\triangleq(\fCyNodeLabels, \fDbPropSchema,
    \fDbNodeSchemaElem[1], \fDbNodeSchemaElem[2], \fDbEdgeDir)
    \)
    is a tuple of a set of labels $\fCyNodeLabels$, a property schema $\fDbPropSchema$,
    its endpoint node schemas $\fDbNodeSchemaElem[1]$ and $\fDbNodeSchemaElem[2]$,
    and edge directionality \smash{$\fDbEdgeDir\in\{\!\fPGDir\!,\fPGUDir\!\}$}.
    \begin{itemize}[nosep, leftmargin=4mm, labelsep=0.4em]
    \item For \smash{$\fDbEdgeDir=\fPGDir$}, $\fDbNodeSchemaElem[1]$ and $\fDbNodeSchemaElem[2]$
    correspond to the source and target node schemas, respectively.
    \item For \smash{$\fDbEdgeDir=\fPGUDir$}, $\fDbNodeSchemaElem[1]$ and $\fDbNodeSchemaElem[2]$
    are unordered endpoint schemas.
    \end{itemize}
  \item[\enumParen{$\fDbGraphSchema$}] A \emph{graph schema} (\isoSectionShort{\isoGraphType})
    $\,\fDbGraphSchema \triangleq (\fDbNodeSchema, \fDbEdgeSchema\,)$
     is a pair of finite sets of node $\fDbNodeSchema$ and edge $\fDbEdgeSchema$ schemas.
  \end{enumerate}
\end{myDef}

\begin{myDef}[Conformance Relations]\label{def:conformance-relation}
  Let $\fPropGraph\!=\!(\fPGNode, \fPGEdge, \fPGDirFunc,
  \fPGSrcDstFunc, \fPGLabelFunc, \fPGPropFunc)$ be a property graph,
  $\fDbGraphSchema = (\fDbNodeSchema, \fDbEdgeSchema\,)$ a graph
  schema, and $\fDbPropSchema\!:\fAttr\rightharpoonup\fMSortPG$ a property schema.
  Suppose $\fDbNodeSchemaElem\!\in\!\fDbNodeSchema$ and
  $\fDbEdgeSchemaElem\!\in\!\fDbEdgeSchema$.
  \begin{enumerate}[nosep, leftmargin=1.2cm, labelsep=0.4em]
  \item[\enumParen{$\fConform{\fPGProp}{\fDbPropSchema}$}] A property map $\fPGProp$ \emph{conforms} to $\fDbPropSchema$ when
    \(
      \fDom{\fPGProp}\!=\!\fDom{\fDbPropSchema}\land
      \forall \fCyProp\!\in\!\fDom{\fPGProp}.\fPGProp(\fCyProp)\!\in\!\fUniverse{\fDbPropSchema(\fCyProp)}.
    \)
    %It is denoted as $\fConform{\fPGProp}{\fDbPropSchema}$.  
  \item[\enumParen{\hspace{-0.7pt}$\fConformGraph{\fPGNodeElem}{\fDbNodeSchemaElem}{\fPGNodeSort}{\fPropGraph}$}] A node $\fPGNodeElem\!\in\!\fPGNode$ \emph{conforms} to a node
    schema $\fDbNodeSchemaElem\!=\!(\fCyNodeLabels, \fDbPropSchema)$,
    when
    \(
      \fPGLabelFunc(\fPGNodeElem)\!=\!\fCyNodeLabels\;\text{and}\;
      \fConform{\fPGPropFunc(\fPGNodeElem)}{\fDbPropSchema}.
    \)  
  \item[\enumParen{$\fConformGraph{\fPGEdgeElem}{\fDbEdgeSchemaElem}{\fPGEdgeSort}{\fPropGraph}$}] An edge $\fPGEdgeElem\!\in\!\fPGEdge$ \emph{conforms} to an edge schema
    \(\fDbEdgeSchemaElem\!=\!(\fCyNodeLabels, \fDbPropSchema,\fDbNodeSchemaElem[1], \fDbNodeSchemaElem[2], \fDbEdgeDir)\)
    when $\fPGLabelFunc(\fPGEdgeElem)\!=\!\fCyNodeLabels$, $\fConform{\fPGPropFunc(\fPGEdgeElem)}{\fDbPropSchema}$, and:
    \begin{itemize}[nosep, leftmargin=4mm, labelsep=0.4em]
      \item $\fPGEdgeElem\!\in\!\fPGEdgeDir$ for $\fDbEdgeDir\!=\!\fPGDir$ with
        \(
          \fConformGraph{\fPGNodeElems{1}}
                       {\fDbNodeSchemaElem[1]}{\fPGNodeSort}{\fPropGraph}\;\text{and}\;
          \fConformGraph{\fPGNodeElems{2}}
                       {\fDbNodeSchemaElem[2]}{\fPGNodeSort}{\fPropGraph}
        \)
        where $\fPGSrcDstFunc(\fPGEdgeElem)\!=\!(\fPGNodeElems{1},\fPGNodeElems{2})$;                
        or               
      \item $\fPGEdgeElem\!\in\!\fPGEdgeUDir$ for $\fDbEdgeDir\!=\!\fPGUDir$ with
        \(
          \exists i,j\!\in\!\fSet{1,2}.i\!\neq\!j.
          \fConformGraph{\fPGNodeElems{i}}
                       {\fDbNodeSchemaElem[1]}{\fPGNodeSort}{\fPropGraph}\;\text{and}\;
          \fConformGraph{\fPGNodeElems{j}}
                       {\fDbNodeSchemaElem[2]}{\fPGNodeSort}{\fPropGraph}
        \)
        where $\fPGSrcDstFunc(\fPGEdgeElem)\!=\!\fSet{\fPGNodeElems{1},\fPGNodeElems{2}}.$
    \end{itemize}
  \item[\enumParen{$\fConformGraph{\fPropGraph}{\fDbGraphSchema}{}{}$}] The graph $\fPropGraph$ \emph{conforms} to
    $\fDbGraphSchema$ when
    \(
      \forall\fPGNodeElem\!\in\!\fPGNode.\,\exists\fDbNodeSchemaElem\!\in\!\fDbNodeSchema.
      \,\fConformGraph{\fPGNodeElem}{\fDbNodeSchemaElem}{\fPGNodeSort}{\fPropGraph}
    \)
    and  
    \(
      \forall\fPGEdgeElem\!\in\!\fPGEdge.\,\exists\fDbEdgeSchemaElem\!\in\!\fDbEdgeSchema.
      \,\fConformGraph{\fPGEdgeElem}{\fDbEdgeSchemaElem}{\fPGEdgeSort}{\fPropGraph}.
    \)
  \end{enumerate}
\end{myDef}

\begin{myDef}[Closed and Open Property Graphs]\label{def:closed-open-graphs}
Let $\fDbGraphSchemas$ be a set of graph schemas and $\fPropGraph$ a property graph.
$\fPropGraph$ is defined as \emph{closed} or \emph{open}
under $\fDbGraphSchemas$ based on whether
$\exists\fDbGraphSchema\!\in\!\fDbGraphSchemas.\,\fConformGraph{\fPropGraph}{\fDbGraphSchema}{}{}$
or $\not\exists\fDbGraphSchema\!\in\!\fDbGraphSchemas.\,\fConformGraph{\fPropGraph}{\fDbGraphSchema}{}{}$,
respectively, \ie closed graphs conform to atleast one
schema, while open graphs conform to none.
%
%the graph conforms to at least one schema or none.
\end{myDef}
%% there exists at least one graph schema $\fDbGraphSchema\!\in\!\fDbGraphSchemas$
%% such that $\fConformGraph{\fPropGraph}{\fDbGraphSchema}{}{}$ 
%% %
%% A property graph for which

%% is defined as \emph{open} under $\fDbGraphSchemas$.

The \isoShort standard refers to the node, edge, and graph schemas from 
Definition~\ref{def:prop-graph-schemas} as their respective \emph{types}.
This alias is evident from the conformance relations (Definition~\ref{def:conformance-relation}),
where a graph or a graph element, \ie node or edge, \emph{conforms} to a
schema only when it satisfies the schema's constraints on all of its semantic
attributes, \ie labels, properties, directionality, \etc which can be likened to fields.
Therefore, when a graph conforms to a schema, \ie closed graphs, the schema
can be leveraged to type \lang queries operating on that graph more precisely,
\ie assign narrower types. We design our type system
(\S\ref{sec:type-system}) to exploit the type information inferrable from the schemas.

\subsection{\lang}
\label{sec:prelims-gql}

\begin{figure}[t]%
  \centering%
  \begin{GrammarBox}%
    \begin{minipage}{0.95\linewidth}%
      \input{formalizations/gql-calculus}%
    \end{minipage}%
  \end{GrammarBox}%
  \vspace{-3mm}%
  \caption{\FigCaptionGqlCalculus}%
\end{figure}%

\lang~\cite{francis2023digest,iso2024gql} is the first international
standard for a property-graph query language, standardized in 2024 as
\iso; and aimed at unifying ideas and concepts from other popular graph query
languages such as \cypher~\cite{francis2018cypher},
\pgql~\cite{van2016pgql}, and \gcore~\cite{angles2018gcore}, \etc
into a single and vendor-neutral specification.
Similar to relational query languages like \sql, 
a \lang query specifies \emph{what} to retrieve from
data sources, and optionally transform them before retrieving; while
the query engine determines \emph{how} to evaluate the query. 
Broadly, data sources in query languages can be divided
into constructed sources and stored sources. Property graphs 
serve as stored data sources in graph query languages, while 
\emph{base tables}~\cite{iso2023sql} are an example of stored sources in \sql.
Since base tables are highly structured data sources with rigid
schemas, a fragment supporting just column projections is often
sufficient purely for data retrieval in \sql lineage of languages.
Property graphs on the other hand have far more flexible schemas, and so
graph query languages include a rich graph pattern matching fragment
specifically for data retrieval. The fragment for the optional data transformations
before retrieving, largely overlap between relational and graph query languages.

The interplay between the graph pattern matching fragment for data retrieval
and the remaining largely relational fragment for data transformations, make
\lang's semantics non-trivial. The semantics is further complicated by \lang
allowing querying multiple graphs within a single query; with different ways of
combining their results, via composite queries and operators.
We formalize a \lang fragment that captures the difficult semantic
challenges arising from the interplay.

\MyPara{Formalized \lang Fragment}%
Figure~\ref{fig:gql-calculus} shows the \lang fragment we formalize.
A \lang query $\fCyQueryExp$ is either a focused linear query
$\fCyQuery$, \ie a query operating on a single graph instance, without
any composite query operators; or a composite query
$\fCyQueryExp\;\fCyCompOp\;\fCyQuery$ that combines the results of
two or more sub-queries using composite query operators
$\,\fCyCompOp\in\{\fOtherwise$, $\fUnion$, $\fExceptJ$, $\fIntersectJ\}$.

A focused linear query $\fCyQuery$ first declares the graph being queried
using the $\fUse\;\fCyVar$ clause, where $\fCyVar\in\fAttr$
refers to the name assigned to the graph instance in the database's catalog.
We use $\fDbGNames$ specifically to refer to graph names, \ie
$\fUse\;\fDbGNames$ selects the graph instance named $\fDbGNames$
in the catalog as the graph being queried (also called the working graph).
Declaration of the working graph is followed by a $\fMatch$
$\fCyPatternExp$ clause which is the heart of \lang's graph pattern
matching fragment. It specifies a pattern list
$\fCyPatternExp$.
Next, an optional $\fWhere$ $\fCyPred$ clause for filtering the
matched bindings using the predicate expression $\fCyPred$.
Finally, a $\fReturn$ $\fCyProjection$ clause for returning the
query result after projecting the filtered matched bindings
via the list of projections $\fCyProjection$.

A pattern list $\fCyPatternExp$ is built from path patterns
$\fCyPattern$ composed via
conjunction $\fCyPatternAnd{\fCyPatternExp}{\fCyPattern}$: a
comma-separated list requiring \emph{all} patterns to match. 
% or
% disjunction~($\fCyPatternOr{\fCyPatternExp}{\fCyPattern}$, a
% pipe-separated alternative).
%
A path pattern $\fCyPattern$ is a sequence of node atoms
$\fCyPatternNode$ connected by---only directed
($\fGqEdge[->e]{\cdot}$, $\fGqEdge[<-e]{\cdot}$,
$\fGqEdge[<->e]{\cdot}$), right directed 
or undirected ($\fGqEdge[;>e]{\cdot}$),
left directed 
or undirected ($\fGqEdge[<;e]{\cdot}$),
only undirected ($\fGqEdge[;e]{\cdot}$), and unconstrained
($\fGqEdge[-e]{\cdot}$)---edge atoms $\fCyPatternEdge$.
Edges and patterns can be quantified by quantifiers $\fCyQuantifier$
whose semantics parallels that of regular-expression repetition.
The quantifier $\fCyQuantifier$ includes: Kleene star
($\fCyQuantStar$), Kleene plus ($\fCyQuantPlus$), optionality
($\fCyQuantQues$), exact repetition ($\fCyQuantExact{i}$), and bounded
repetition ($\fCyQuantBound{i}{j}$), where
$i,j\in\mathbb{N}$ and $i\leq j$.
Node atoms $\fCyPatternNode$ may declare binding variables
$\fCyVar$, specify a label expression $\fCyLabelExp$, and impose property
constraints via a property map $\fCyPropMap$. Same holds for the edge atoms
$\fCyPatternEdge$.
Label expressions $\fCyLabelExp$ are formed from label names ($\fCyLabel\in\fAttr$) and the
wildcard ($\fCyLWild$) under: conjunction ($\fCyLAnd{}{}$), disjunction
($\fCyLOr{}{}$), and negation ( $\fCyLNeg{\!}$ ). This allows for fine-grained control 
over which graph elements a pattern atom matches.

Value expressions $\fCyExp$ include constants, variables, property
accesses $\fCyVarProp{\fCyVar}{\fCyVar}$, arithmetic, and aggregate
function applications $\fCyAggExp$.
Predicates $\fCyPred$ are formed from Boolean constants $\fTrue$/$\fFalse$, relational
operations $\fCyExpRel{\fCyExp}{\fCyExp}$, the null test
$\fCyExpIsNull{\fCyExp}$, and the logical connectives
($\fCyExpAnd{\!}{\!}$, $\fCyExpOr{\!}{\!}$, $\fCyExpNeg{\!}$).
Aggregate functions $\fCyAExpCount$, $\fCyAExpSum$, $\fCyAExpMax$, and
$\fCyAExpMin$ may be applied with a $\fDstnct$ or $\fAll$ (default)
qualifier, which decides if the aggregating operation should be 
performed under set or bag semantics.

\begin{figure}[t]%
  \centering%
  \begin{minipage}{0.52\linewidth}%
    \begin{minipage}{1.0\linewidth}%
\hspace*{-1cm}
\begin{CypherBox}        
(*@$\fUse\;\fDbGNames\;\fMatch\;\;$@*)(*@{$\textcolor{\ExamplePatMatchNodeColor}{\fCyElemNodePattern{\fCyNode{m}}}\textcolor{\ExamplePatKnowsEdgeColor}{\fCyEdgeRight{\fCyElemEdgePattern{\text{\texttt{\!:KNOWS}}}}\!\fCyQuantBound{1}{2}}\,\textcolor{\ExamplePatPersonNodeColor}{\fCyElemNodePattern{\phantom{i}}}\textcolor{\ExamplePatLeadsEdgeColor}{\fCyEdgeRight{\fCyElemEdgePattern{\text{\texttt{\!:LEADS}}}}}\,\textcolor{\ExamplePatProjectNodeColor}{\fCyElemNodePattern{\phantom{i}}}$}@*)
\end{CypherBox}%
    \end{minipage}%
    \vspace{3mm}%
    \\
    \begin{minipage}{1.0\linewidth}%
    \resizebox{0.9\linewidth}{!}{\begin{tikzpicture}[
  font=\footnotesize,
  >=Latex,
  node distance=15mm and 20mm,
  vertex/.style={
    circle,
    draw,
    thick,
    white,
    minimum size=6mm,
    inner sep=0pt,
    draw opacity=0,
    fill opacity=0.5,
    text opacity=0.5,
    align=center
  },
  vertexMatch/.style={
    circle,
    draw,
    very thick,
    white,
    minimum size=6mm,
    inner sep=0pt,
    align=center,
  },
  numberMatch/.style={
    circle,
    draw=black,
    fill=white,
    text=black,
    thick,
    font=\bfseries\normalsize,
    minimum size=3.6mm,
    inner sep=0pt,
    align=center
  },
  person/.style={vertex, fill=blue!20},
  personMatch/.style={vertexMatch, fill=blue!20, draw=\ExamplePatMatchNodeColor},
  personMatchInt/.style={vertexMatch, fill=blue!20, draw=\ExamplePatPersonNodeColor},
  personMatchIntAnon/.style={vertexMatch, fill=blue!20},
  pattern/.style={numberMatch},
  company/.style={vertex, fill=orange!20},
  proj/.style={vertexMatch, fill=purple!20, draw=purple!90!black},
  projMatch/.style={vertexMatch, fill=purple!20, draw=\ExamplePatProjectNodeColor},
  vname/.style={font=\bfseries, text opacity=0.5},
  vnameMatch/.style={font=\bfseries},
  vlabel/.style={font=\scriptsize, yshift=-5mm, text opacity=0},
  callout/.style={
    draw=black!80,
    rounded corners=1pt,
    %thin,
    fill=black!6,
    font=\scriptsize,
    inner sep=2pt,
    align=left
  },
  rel/.style={
    -{Latex[length=2mm,width=1.6mm]},
    thin,
    draw=black!30
  },
  relMatch/.style={
    -{Latex[length=2mm,width=1.8mm]},
    very thick,
    draw=purple!90!black
  },
  rlabel/.style={
    fill=white,
    inner sep=1.2pt,
    font=\scriptsize,
    text opacity=0.35
  },
  relMatchKnows/.style={
    -{Latex[length=2mm,width=1.8mm]},
    very thick,
    draw=\ExamplePatKnowsEdgeColor
  },
  relMatchLeads/.style={
    -{Latex[length=2mm,width=1.8mm]},
    very thick,
    draw=\ExamplePatLeadsEdgeColor
  },
  rlabelMatch/.style={
    fill=white,
    inner sep=1.2pt,
    text=purple!90!black,
    font=\scriptsize
  },
  rlabelMatchKnows/.style={
    fill=white,
    inner sep=1.2pt,
    text=\ExamplePatKnowsEdgeColor,
    font=\bfseries\footnotesize
  },
  rlabelMatchLeads/.style={
    fill=white,
    inner sep=1.2pt,
    text=\ExamplePatLeadsEdgeColor,
    font=\bfseries\footnotesize
  },
  frame/.style={draw, rounded corners=3pt, thin, inner sep=7pt}
]

% Persons
\node[personMatch] (a) at (0.8,0) {};
\node[vnameMatch] at (a) {a};
\node[vlabel] at (a) {};

\node[pattern, right=3.5mm] at (a.east) {1};

\node[personMatchInt] (b) at (2.3,1.4) {};
\node[vnameMatch] at (b) {o};
\node[vlabel] at (b) {};

\node[person] (c) at (2.3,-1.4) {};
\node[vname] at (c) {c};
\node[vlabel] at (c) {};

% Project
\node[projMatch] (p) at (0.3,2.0) {};
\node[vnameMatch] at (p) {s};
\node[vlabel] at (p) {};

% KNOWS edges
\draw[relMatchKnows] (a) to[bend left=12] node[rlabelMatchKnows, midway] {\texttt{KNOWS}} (b);
\draw[rel] (a) to[bend right=12] node[rlabel, midway] {\texttt{KNOWS}} (c);
\draw[rel] (c) -- node[rlabel, midway] {\texttt{KNOWS}} (b);

\draw[relMatchLeads] (b) -- node[rlabelMatchLeads, midway, xshift=3pt, yshift=-1pt] {\texttt{LEADS}} (p);

\node[person] (a1) at (7.8,0) {};
\node[vname] at (a1) {a};
\node[vlabel] at (a1) {};

\node[personMatchInt] (b1) at (9.3,1.4) {};
\node[vnameMatch] at (b1) {o};
\node[vlabel] at (b1) {};

\node[pattern, right=3.5mm] at (a1.east) {3};

\node[personMatch] (c1) at (9.3,-1.4) {};
\node[vnameMatch] at (c1) {c};
\node[vlabel] at (c1) {};

% Project
\node[projMatch] (p1) at (7.3,2.0) {};
\node[vnameMatch] at (p1) {s};
\node[vlabel] at (p1) {};

% KNOWS edges
\draw[rel] (a1) to[bend left=12] node[rlabel, midway] {\texttt{KNOWS}} (b1);
\draw[rel] (a1) to[bend right=12] node[rlabel, midway] {\texttt{KNOWS}} (c1);
\draw[relMatchKnows] (c1) -- node[rlabelMatchKnows, midway] {\texttt{KNOWS}} (b1);

\draw[relMatchLeads] (b1) -- node[rlabelMatchLeads, midway, xshift=3pt, yshift=-1pt] {\texttt{LEADS}} (p1);

\node[personMatch] (a2) at (4.3,0) {};
\node[vnameMatch] at (a2) {a};
\node[vlabel] at (a2) {};

\node[pattern, right=3.5mm] at (a2.east) {2};

\node[personMatchInt] (b2) at (5.8,1.4) {};
\node[vnameMatch] at (b2) {o};
\node[vlabel] at (b2) {};

\node[personMatchIntAnon] (c2) at (5.8,-1.4) {};
\node[vnameMatch] at (c2) {c};
\node[vlabel] at (c2) {};

% Project
\node[projMatch] (p2) at (3.8,2) {};
\node[vnameMatch] at (p2) {s};
\node[vlabel] at (p2) {};

% KNOWS edges
\draw[rel] (a2) to[bend left=12] node[rlabel, midway] {\texttt{KNOWS}} (b2);
\draw[relMatchKnows] (a2) to[bend right=12] node[rlabelMatchKnows, midway] {\texttt{KNOWS}} (c2);
\draw[relMatchKnows] (c2) -- node[rlabelMatchKnows, midway] {\texttt{KNOWS}} (b2);

\draw[relMatchLeads] (b2) -- node[rlabelMatchLeads, midway, xshift=3pt, yshift=-1pt] {\texttt{LEADS}} (p2);

% Node types
\node[vlabel] (spacing) at (3,-2.6) {};
\node[personMatch] (f) at (3,-3.5) {};
\node[vnameMatch] at (f) {a};
\node[personMatch] (z) at (5,-3.5) {};
\node[vnameMatch] at (z) {a};
\node[personMatch] (q) at (7,-3.5) {};
\node[vnameMatch] at (q) {c};
\node[
  draw,
  thick,
  rectangle,
  rounded corners=1pt,
  fit=(spacing)(f)(z)(q),
  inner sep=6pt,
] (infoBox) {};

\node[
  draw,
  thick,
  rounded corners=0.5pt,
  fill=white,
  inner sep=3pt,
] at (infoBox.north) {\normalsize $\fReturn$ \textcolor{\ExamplePatMatchNodeColor}{\fCyNode{m}}; \textbf{(Bag Semantics)}};

\end{tikzpicture}}%
    \end{minipage}%
    %\hspace*{-6mm}%
  \end{minipage}%
  \caption{\FigCaptionGqlExample}
\end{figure}%

\MyPara{Execution Constructs}%
A \lang query executes within a \emph{database world}
$\fDbWorldS = (\fDbCatalog, \fDbSchema)$.
The \emph{graph catalog}
$\fDbCatalog\!:\fAttr\rightharpoonup\fPropGraph$ is a fixed
collection that resolves graph names ($\fDbGNames\in\fAttr$) to property graph instances
(\isoSection{4.2.5}).
The \emph{schema mapper}
$\fDbSchema\!:\fAttr\rightharpoonup\fDbGraphSchema$ associates each
graph name with a graph schema (if defined). Schema information is
available only for closed graphs.
The \emph{working graph site} $\fDbGNames$ designates the graph
currently in scope and is selected from the catalog via the $\fUse$ clause. 
All pattern matching in the query after graph selection operates against that graph 
(\isoSection{4.7.5}).

We formally define \lang's execution constructs when formalizing 
the semantics (\S\ref{sec:semantics-operational}). Here, we
introduce them informally for contextual reasons.
A \emph{record} is a finite set of fields, each pairing a variable
name with a runtime value (\S\ref{sec:semantics-runtime-dom})---a scalar, a graph-element identity, or the
null value (\isoSection{4.15.4}).
A \emph{binding table} is a bag of records having compatible schema (Definition~\ref{def:record-schema}).
It serves as the primary execution construct in query evaluation pipeline,
by holding intermediary results such as the matches found by
graph pattern matching, and constituting the execution result returned
after query evaluation (\isoSection{4.3.6}).
Every clause in a query transforms a binding table:
$\fMatch$ produces one, $\fWhere$ filters it, and $\fReturn$ projects
it into the query's result.

Evaluating $\fMatch\;\fCyPatternExp$ against the working graph
enumerates all graph homomorphisms
$h\!:\fCyPatternExp\!\to\!\fPropGraph$ that satisfy the structural,
label, and property constraints of $\fCyPatternExp$.
Each such homomorphism yields a record that binds the pattern variables
to the matched graph elements, and the collection of all such records forms
the binding table produced by the $\fMatch$ clause.
Pattern conjunction combines binding tables via a join on shared
variables.

The $\fWhere$ clause retains only those records for which the
predicate $\fCyPred$ evaluates to true, and the $\fReturn$ clause
maps each surviving record to its projected expressions, producing the
output binding table of the query.
Composite queries compose queries by combining their binding tables
under different composite query operators: $\fUnion$---union,
$\fIntersectJ$---intersection, $\fExceptJ$---difference,
and $\fOtherwise$ returns the left operand's table when non-empty otherwise the
right's.

\begin{myExample}[\lang Graph Pattern Matching]\label{ex:prelims-gql}%
Consider the example \lang query shown in Figure~\ref{fig:gql-example}:

{  
  \ExampleFont%
  \setlength{\abovedisplayskip}{-3mm}%
  \setlength{\belowdisplayskip}{0mm}%
  \begin{align*}
    \fUse\;\fDbGNames\;
    \fMatch\;
    \fGqNode{m}
    \fGqEdgeQuant[->][l][KNOWS][][bound][1][2]{}
    \fGqNode{\phantom{i}}
    \fGqEdge[->][l][LEADS]{}
    \fGqNode{\phantom{i}}
    \;\;\fReturn\; \fCyNode{m};
  \end{align*}
}  

%In our core calculus from Figure~\ref{fig:gql-calculus}, 
\noindent%
This query has the shape
\(
  \fMatch\; \fCyPattern \;\fReturn\; \fCyProjection,
\)
from the core calculus in Figure~\ref{fig:gql-calculus}.
The projection item is $\fCyProjection=\fCyNode{m}$ and a single
path pattern
\(
  \fCyPattern=
  \fGqNode{m}
  \fGqEdgeQuant[->][l][KNOWS][][bound][1][2]{}
  \fGqNode{\phantom{i}}
  \fGqEdge[->][l][LEADS]{}
  \fGqNode{\phantom{i}}.
\)
The pattern $\fCyPattern$ requires a node $\fCyNode{m}$ from which
there exists a path of one or two \texttt{KNOWS} edges to some
intermediate node, followed by a single outgoing \texttt{LEADS} edge
to a final node. Evaluation against the property graph instance
$\fPropGraph$ in Figure~\ref{fig:pgm-instance} proceeds by
enumerating all graph homomorphisms
$h\!:\fCyPattern\!\to\!\fPropGraph$ satisfying the path constraint,
yielding three records in the binding table:
\begin{enumerate}[nosep, leftmargin=7mm]
  \item[\circledWhite{\textbf{1}}]
    \(\fCyNode{m}=\circledColor{\textbf{a}}{blue!20}{black}{1}{0ex}\), via the path
    \tikzGTwoEdge{\textbf{a}}{blue!20}{\textbf{o}}{blue!20}{\textbf{s}}{purple!20}{KNOWS}{LEADS}.
  \item[\circledWhite{\textbf{2}}]
    \(\fCyNode{m}=\circledColor{\textbf{a}}{blue!20}{black}{1}{0ex}\), via the
    longer path
    \tikzGThreeEdge{\textbf{a}}{blue!20}{\textbf{c}}{blue!20}{\textbf{o}}{blue!20}{\textbf{s}}{purple!20} from quantifier $\{1,2\}$.
  \item[\circledWhite{\textbf{3}}]
    \(\fCyNode{m}=\circledColor{\textbf{c}}{blue!20}{black}{1}{0ex}\), via the path
    \tikzGTwoEdge{\textbf{c}}{blue!20}{\textbf{o}}{blue!20}{\textbf{s}}{purple!20}{KNOWS}{LEADS}.
\end{enumerate}

\noindent%
Since no $\fWhere$ clause is present, all three records survive
filtering. The $\fReturn\;\fCyNode{m}$ clause then projects each
record onto the variable $\fCyNode{m}$, and under bag semantics the
result is the bag
\(
\fBag{%
  \circledColor{\textbf{a}}{blue!20}{black}{1}{0ex},
  \circledColor{\textbf{a}}{blue!20}{black}{1}{0ex},
  \circledColor{\textbf{c}}{blue!20}{black}{1}{0ex}
}%
\)
containing one entry per homomorphism, including the duplicate binding
of~$\fCyNode{m}$ to~$\circledColor{\textbf{a}}{blue!20}{black}{1}{0ex}$.
\qed\end{myExample}

\begin{figure}[t]%
  \centering%
  %\begin{GrammarBox}%
    \begin{minipage}{0.95\linewidth}%
      \input{formalizations/gql-metafunctions}%
    \end{minipage}%
  %\end{GrammarBox}%
  \caption{\FigCaptionGqlMetafunctions}%
\end{figure}%

\MyPara{Metafunctions for Parsing Patterns}%
Our formalization uses several metafunctions for parsing \lang query patterns
as shown in Figure~\ref{fig:gql-metafunctions}.
$\fVar{\cdot}$, $\fLbl{\cdot}$, and $\fProp{\cdot}$
return the binding variable, label expression, and property map of a
pattern atom (node or edge).
We interpret the property map $\fPGProp$ returned by $\fProp{\cdot}$ 
as a partial, mapping property names to values.
$\fQLo{\cdot}$ returns the lower bound of a quantifier.
$\fTail{\cdot}$ returns the
binding variable of a pattern's rightmost node atom.
$\fMinE{\cdot}$ returns the minimum number of edge atoms in a pattern, while 
$\fMinN{\cdot}$ returns the minimum number of node atoms in a pattern. Both
account for quantifiers.
Finally, $\fDir{\cdot}$ maps an edge atom based on its direction,
to a set of edge \emph{orientations}: left directed
($\fOrientL$), undirected ($\fOrientU$), and right directed
($\fOrientR$); while $\fDirDenote{\cdot}$ maps orientations to directionalities: directed
(\smash{$\fPGDir$}) or undirected (\smash{$\fPGUDir$}).

\section{Type System}
\label{sec:type-system}

The graph schema support intrinsic to \lang's standard allows
for more precise typing. 
We design our type system to exploit the availability of graph schemas
to propagate their graph-structural invariants through query patterns into
query results.
Broadly, our design includes four features of the \lang specification,
some of which are absent from simpler formalisms~\cite{Francis2023gpc,deutsch2022graph}:
(1)~\emph{graph scoping}---elements of different graphs cannot be conflated with each other;
(2)~\emph{schema refinement}---closed graphs carry schemas that constrain semantic
attributes of their elements allowing for more precise typing;
(3)~\emph{nullability}---missing properties and failed matches produce
null values which obey three-valued logic (\threeVL);
(4)~and \emph{heterogeneous unions}---composite queries and property accesses
produce bindings whose runtime type varies across records within the same binding table.

{
\setlength{\abovedisplayskip}{-2mm}%
\setlength{\belowdisplayskip}{1mm}%
\begin{align*}
  \small
  \begin{array}{@{}c@{\qquad\quad\;}c@{}@{\quad\qquad\;}c@{}}
    \begin{array}{@{}r@{\;\;}c@{\;\;}l@{\;\;}l@{}}
      \fMSortPG
      & \fTermDef
      & \fInt
      & \text{\emph{integer}}
      \\[1pt]
      & \mid
      & \fBool
      & \text{\emph{boolean}}
      \\[1pt]
      & \mid
      & \fStr
      & \text{\emph{string}}
      \\[1pt]
      \\[1pt]
      \\[1pt]
    \end{array}
    &
    \begin{array}{@{}r@{\;\;}c@{\;\;}l@{\;\;}l@{}}
      \fMSortGQL
      & \fTermDef
      & \fMSortPG
      & \text{\emph{scalars}}
      \\[1pt]
      & \mid
      & \fDbNodeType[-]
      & \multirow{2}{*}{\text{\emph{node}}}
      \\[1pt]
      & \mid
      & \fDbNodeType[-,-] 
      & 
      \\[1pt]
      & \mid
      & \fDbEdgeType[-] 
      & \multirow{2}{*}{\text{\emph{edge}}}
      \\[1pt]
      & \mid
      & \fDbEdgeType[-,-]
      &
    \end{array}
    &
    \begin{array}{@{}r@{\;\;}c@{\;\;}l@{\quad}l@{}}
      \fMSortGQLN
      & \fTermDef
      & \fMSortGQL
      & \text{\emph{value types}}
      \\[1pt]
      & \mid
      & \fAny
      & \text{\emph{top type}}
      \\[1pt]
      & \mid
      & \fEmpty
      & \text{\emph{bottom type}}
      \\[1pt]
      & \mid
      & \fNullable
      & \text{\emph{null type}}
      \\[1pt]
      & \mid
      & \fMSortGQLN \cup \fMSortGQLN \mid \fListT\;\fMSortGQLN \mid \fEmptyT{\fMSortGQLN}
      & \text{\emph{type formers}}
    \end{array}
  \end{array}
\end{align*}
}

\noindent%
Our type system is organized into three tiers.
$\fMSortPG$ gives \emph{scalar types} (property-value primitives).
$\fMSortGQL$ extends scalars with \emph{graph-scoped element
types}---node and edge identity types indexed by a graph site and
optionally refined by a schema entry.
$\fMSortGQLN$ gives the full \emph{value type} language: element types
$\fMSortGQL$, the top type $\fAny$, the bottom type $\fEmpty$, the null type
$\fNullable$, and is closed under the types formed from the union,
the list, and the empty~$\fEmptyT{(\cdot)}$ formers. Unless otherwise stated, we use
$\fSortVar$ to range over $\fMSortGQLN$.

We now motivate our types and type formers, starting with
graph-indexed and schema-refined types
(\S\ref{subsec:typing-refinement-type}), followed by the union former
(\S\ref{subsec:typing-dut}), then the list former
(\S\ref{subsec:typing-list-types}), and finally, the null type and the empty former
(\S\ref{subsec:typing-nullability}). We conclude with the subtyping
rules for our type system (\S\ref{subsec:typing-subtyping}).

%% dynamic refinement
%% (Appendix~\ref{subsec:typing-coercion}), and typing environments
%% (\S\ref{subsec:typing-auxiliary}).

\subsection{Graph-Indexed and Schema-Refined Types}
\label{subsec:typing-refinement-type}

%\MyPara{Graph Scoping}
\lang supports composite queries where the individual queries may
query different property graphs. However, graph elements are
inherently \emph{graph-local}, \ie nodes and edges of one property
graph cannot be conflated with another.
We prevent cross-graph element conflation through \emph{graph-indexed
types}: for a graph site $\fDbGNames$, the types $\fDbNodeType$ and
$\fDbEdgeType$ classify node and edge identities scoped
to~$\fDbGNames$.

%\MyPara{Schema Refinement}
For graph sites referring to closed graphs
% , it is \emph{closed} under a graph schema
% $\fDbGraphSchema = (\fDbNodeSchema, \fDbEdgeSchema)$
(Definition~\ref{def:closed-open-graphs}), the type system refines
graph-indexed types with schema entries.
The schema-refined types
$\fDbNodeType[\fDbGNames,\fDbNodeSchemaElem]$ and
$\fDbEdgeType[\fDbGNames,\fDbEdgeSchemaElem]$ classify elements whose
labels, properties, \etc conform (Definition~\ref{def:conformance-relation}) to their respective schema entries.
Refined types carry strictly more
static information, enabling more precise type assignments and propagations.

% and the subtyping relation
% (\S\ref{subsec:typing-subtyping}) makes every schema-refined type a
% subtype of its graph-indexed counterpart.

\subsection{\lang's Dynamic Union Types: Union Type Former}
\label{subsec:typing-dut}

\begin{myExample}[Heterogeneous Unions]\label{ex:heterogenous-union}
  Consider the composite \lang query,

  {
  \ExampleFont%
  \setlength{\abovedisplayskip}{-2mm}%
  \setlength{\belowdisplayskip}{4mm}%
  \begin{align*}
  \underboxcap{
  \fUse\;\fDbGNames\;
  \fMatch\;%
  \fCyElemNodePattern{\fCyNode{$\fCyNodeVars{1}$}\textbf{:}\texttt{COMPANY}}\;%
  \;\fReturn\; \fCyNode{$\fCyNodeVars{1}$}\; \fAs\; \fCyNode{x}}{{\scriptsize \text{Returns a binding table with column \texttt{\textbf{x}} typed as $\fDbNodeType[{\fDbGNames,\fDbNodeSchemaElem[1]}]$}}}\;
  \;\fUnion\;
  \;\underboxcap{
  \fUse\;\fDbGNames\;
  \fMatch\;
  \fCyElemNodePattern{\fCyNode{$\fCyNodeVars{2}$}\textbf{:}\texttt{PROJECT}}\;
  \;\fReturn\; \fCyNode{$\fCyNodeVars{2}$}\; \fAs\; \fCyNode{x}}{{\scriptsize \text{Returns a binding table with column \texttt{\textbf{x}} typed as $\fDbNodeType[{\fDbGNames,\fDbNodeSchemaElem[2]}]$}}}
  \end{align*}
  }

  \noindent%
  evaluated against the property graph (assuming closed) from
  Figure~\ref{fig:pgm-instance}.
  Suppose the two individual focused linear queries assign different
  types ($\fDbNodeType[{\fDbGNames,\fDbNodeSchemaElem[1]}]$ and
  $\fDbNodeType[{\fDbGNames,\fDbNodeSchemaElem[2]}]$) to the
  column~$\fCyNode{x}$ in their respective binding tables; the
  combined binding table must classify $\fCyNode{x}$ as inhabiting
  \emph{either} type.
\qed\end{myExample}

\noindent%
Example~\ref{ex:heterogenous-union} illustrates the need for
supporting heterogeneous type unions due to \lang's composite queries.
Similar situations may arise from property accesses. We model
heterogeneous unions by closing our type system under the union former.
For any two types $\fSortVars{1}$ and $\fSortVars{2}$, the dynamic
union $\fSortVars{1}\cup\fSortVars{2}$ is a static type that
admits values whose type inhabits either branch, \ie the union
former does not introduce a new kind of runtime value, every
inhabiting value belongs to either $\fSortVars{1}$ or $\fSortVars{2}$.

%\ie every inhabiting value belongs to one of the constituent
%base types.

\begin{myNotation}[Schema-Refined Union Types]\label{not:schema-refined}
  Consider a graph site $\fDbGNames$ with graph schema
  $(\fDbNodeSchema,\fDbEdgeSchema)$, and some \emph{non-empty} set of
  node \smash{$\fDbNodeSchema'\!\subseteq\fDbNodeSchema$} and edge
  \smash{$\fDbEdgeSchema'\!\subseteq\fDbEdgeSchema$} schemas. We denote the
  node and the edge schema refined types formed from these sets of
  node and edge schemas, closed under the union type former, using the
  syntactic sugars:
  \smash{\(
    \fDbNodeType[\fDbGNames,\fDbNodeSchema']\triangleq
    \bigcup_{\fDbNodeSchemaElem\in\fDbNodeSchema'}\fDbNodeType[\fDbGNames,\fDbNodeSchemaElem]
  \)}
  and
  \smash{\(
    \fDbEdgeType[\fDbGNames,\fDbEdgeSchema']\triangleq
    \bigcup_{\fDbEdgeSchemaElem\in\fDbEdgeSchema'}\fDbEdgeType[\fDbGNames,\fDbEdgeSchemaElem].
  \)}
\qed\end{myNotation}

\subsection{List Types and Reference Modes}
\label{subsec:typing-list-types}

Binding variables declared inside quantified path patterns in \lang can contribute
multiple bindings to a single match (\isoSection{\isoRefContext}).
Consider Example~\ref{ex:prelims-gql} with the quantified
\texttt{KNOWS} segment annotated with edge variable~$\fCyNode{x}$:
\(
\,\fCyElemNodePattern{\fCyNode{m}}\fCyEdgeRight{\fCyElemEdgePattern{\fCyNode{x}\text{\texttt{:KNOWS}}}}\!\fCyQuantBound{1}{2}\fCyElemNodePattern{\phantom{i}}\fCyEdgeRight{\fCyElemEdgePattern{\text{\texttt{\!:LEADS}}}}\!\fCyElemNodePattern{\phantom{i}}\,.
\)
The variable $\fCyNode{m}$ (outside the quantifier) remains
\emph{singleton}, \ie each match binds it to exactly one node.
The variable $\fCyNode{x}$ (inside the quantifier $\fCyQuantBound{1}{2}$)
is \emph{group-reference}, \ie contributes multiple bindings per match:
a one-hop match yields [\!\!\!\tikzGEdgeOpa{KNOWS}\!\!], a two-hop
match yields [\!\!\!\tikzGEdgeOpa{KNOWS}\!,\!\tikzGEdgeOpa{KNOWS}\!\!].
The binding exposed outside the quantifier is not a single edge but
the finite list of all edges matched by that segment.

We model this statically via $\fListT\,\fSortVar$: if a binding variable has type
$\fSortVar$ inside the quantifier, crossing the quantifier boundary
lifts it to $\fListT\,\fSortVar$.
List types thus serve as the static representation of group-reference
bindings. In the example above, $\fCyNode{x}$ is assigned
$\fListT\;\fDbEdgeType$ outside the quantified segment.

\begin{myDef}[Component Type]%
\label{def:base-type}
The \emph{component type} $\fBaseType{\cdot}$ strips the outermost
$\fListT$ wrapper, \ie $\fBaseType{\fListT\;\fSortVar} = \fSortVar$. 
The component type of a non-list type is the type itself.
% {
%   \small
%   \setlength{\abovedisplayskip}{4pt}%
%   \setlength{\belowdisplayskip}{4pt}%
%   \[
%   \fBaseType{\fSortVar}\;\triangleq\;
%   \begin{cases}
%     \fSortVar' & \text{if } \fSortVar = \fListT\;\fSortVar'\\
%     \fSortVar  & \text{else}
%   \end{cases}
%   \]
% }
\end{myDef}

\subsection{\lang's Immaterial Types: Null and Empty}
\label{subsec:typing-nullability}

%\MyPara{Null Type}
Null values are pervasive in \lang, as in \sql~\cite{guagliardo2017formal};
arising from missing properties, optional or failed matches, and undefined
intermediate results, etc.
The null type $\fNullable$ is inhabited solely by the value $\fNull$.
We use the syntactic sugar $\fSortVar$ $\!\fNullable$ to denote the
static type $\fSortVar$ $\cup$ $\fNullable$ formed using the union
former.

\lang's standard specifies an empty type $\fEmpty$ which no value
inhabits, and motivates it as an opportunity to embed richer
statically inferrable information such as knowledge that a binding
does not produce any runtime values (\eg failed matches). The empty
type is the bottom type.

We retain \lang's empty type $\fEmpty$, and use it as the bottom type
of our type system, but in addition, we also introduce the empty
type former $\fEmptyT{\fSortVar}$ to denote a binding
that is statically known to produce no runtime value while retaining the
underlying type $\fSortVar$ as an annotation (\lang's bare empty type
$\fEmpty$ erases that information). We motivate this design choice
using the Example~\ref{ex:empty-type-former} given below.

\begin{myExample}[Empty Former]\label{ex:empty-type-former}
Consider the query with patterns statically known to fail matching.

{
  \ExampleFont
  \setlength{\abovedisplayskip}{-3mm}%
  \setlength{\belowdisplayskip}{0mm}%
  \begin{align*}
    \begin{array}{@{}l@{\hspace{12mm}}l@{}}
      \fUse\;\fDbGNames\;
      \fMatch\;
      {
      \fGqNode[vl][\texttt{COMPANY}]{x}\fCyComma\hspace{1mm}
      \fGqNode[v]{x}
      \fGqEdge[->][vl][\texttt{KNOWS}]{}
      \fGqNode[v]{\phantom{i}}
      }\;\;
      \fReturn\;\fCyNode{x}
      &\emph{Type Check Pass}\;\textcolor{DarkGreen}{\bm{\checkmark}}
    \end{array}
  \end{align*}
}

\noindent%
Typing $\fCyNode{x}$ as $\fDbNodeType$ preserves the base type $\fDbNodeType$ but 
loses the emptiness guarantee, while typing it as $\fEmpty$ captures emptiness but
erases the base type $\fDbNodeType$.
To understand why erasing the base type is problematic, consider 
a small change to the query where $\fCyNode{x}$ binds to an edge in
the right pattern:

{
\ExampleFont%
\setlength{\abovedisplayskip}{-3mm}%
\setlength{\belowdisplayskip}{-1pt}%
\newcommand{\mDel}{ \mathop{ \tcboxmath[ colback=red!15,
      colframe=black, boxrule=0.4pt, arc=0.5pt,
      left=-1pt,right=-1pt,top=-1pt,bottom=-1pt]{\textbf{x}}}\limits_{---}
}%
\newcommand{\mAdd}{ \mathop{ \tcboxmath[ colback=green!15,
      colframe=black, boxrule=0.4pt, arc=0.5pt,
      left=-1pt,right=-1pt,top=-1pt,bottom=-1pt]{\textbf{x}}}\limits_{+++}
}%
\begin{align*}%
  \begin{array}{@{\hspace{-5pt}}ll@{}}%
    \fUse\;\fDbGNames\;
    \fMatch\;%
    {
      \fGqNode[vl][\texttt{COMPANY}]{x}\fCyComma\hspace{1mm}
      \fGqNode[v]{$\mDel$}
      \fGqEdge[->][vl][\texttt{KNOWS}]{$\mAdd$}
      \fGqNode[v]{\phantom{i}}
    }\;%
    \;\fReturn\; \fCyNode{x}
    &\emph{Type Check Fail}\;\textcolor{red}{\bm{\times}}
  \end{array}
\end{align*}
}

\noindent%
This query must be rejected since $\fCyNode{x}$ cannot bind both a node
and an edge at the same depth of graph pattern matching.
But typing $\fCyNode{x}$ as $\fEmpty$ erases the base type
($\fDbNodeType$ in the left pattern vs.\ $\fDbEdgeType$ in the right) 
required for type checking to detect the inconsistency and reject the query.
\qed\end{myExample}

\noindent%
The empty type former $\fEmptyT{(\cdot)}$ resolves the problem in
Example~\ref{ex:empty-type-former} by encoding emptiness while
preserving the underlying type as an annotation.  This allows
emptiness to propagate without sacrificing the ability to reject
ill-formed queries.
We restrict $\fEmptyT{(\cdot)}$ to the element types
$\fDbNodeType[-]$ and $\fDbEdgeType[-]$. We use
$\fDbNodeType[-,\varnothing]$ and $\fDbEdgeType[-,\varnothing]$
as syntactic sugars to denote $\fEmptyT{\fDbNodeType}$ and\,
$\fEmptyT{\fDbEdgeType}$, respectively.

\subsection{Subtyping}
\label{subsec:typing-subtyping}

The subtyping relation
$\fSubType{\fSortVars{1}}{\fSortVars{2}}$ is the least reflexive and
transitive relation closed under:

\begin{RulesDisplayCustomFont}{\small}
\begin{mathpar}
  \infer[\FSubRefl]
  {\phantom{a}}
  {\fSubType{\fSortVar}{\fSortVar}}
  \and
  \infer[\FSubTrans]
  {\fSubType{\fSortVars{1}}{\fSortVars{2}}\\
   \fSubType{\fSortVars{2}}{\fSortVars{3}}}
  {\fSubType{\fSortVars{1}}{\fSortVars{3}}}
  \and
  \infer[\FSubAny]
  {\phantom{a}}
  {\fSubType{\fSortVar}{\fAny}}
  \and
  \infer[\FSubEmpty]
  {\phantom{a}}
  {\fSubType{\fEmpty}{\fSortVar}}
  \and
  \infer[\FSubList]
  {\fSubType{\fSortVars{1}}{\fSortVars{2}}}
  {\fSubType{\fListT\;\fSortVars{1}}{\fListT\;\fSortVars{2}}}
  \\
  \infer[\FSubUnionL]
  {\phantom{a}}
  {\fSubType{\fSortVars{1}}{\fSortVars{1}\cup\fSortVars{2}}}
  \and
  \infer[\FSubUnionR]
  {\phantom{a}}
  {\fSubType{\fSortVars{2}}{\fSortVars{1}\cup\fSortVars{2}}}
  \and
  \infer[\FSubUnionElim]
  {\fSubType{\fSortVars{1}}{\fSortVar}\\
   \fSubType{\fSortVars{2}}{\fSortVar}}
  {\fSubType{\fSortVars{1}\cup\fSortVars{2}}{\fSortVar}}
  \and
  \infer[\FSubUnionCon]
  {\fSubType{\fSortVars{1}}{\fSortVars{1}'}\\
   \fSubType{\fSortVars{2}}{\fSortVars{2}'}}
  {\fSubType{\fSortVars{1}\cup\fSortVars{2}}{\fSortVars{1}'\cup\fSortVars{2}'}}
   \\
  \infer[\FSubRefineN]
  {\phantom{a}}
  {\fSubType{\fDbNodeType[\fDbGNames,-]}{\fDbNodeType}}
  \and
  \infer[\FSubRefineNE]
  {\phantom{a}}
  {\fSubType{\fEmptyT{\fDbNodeType}}{\fDbNodeType[\fDbGNames,-]}}
  \and
  \infer[\FSubRefineE]
  {\phantom{a}}
  {\fSubType{\fDbEdgeType[\fDbGNames,-]}{\fDbEdgeType}}
  \and
  \infer[\FSubRefineEE]
  {\phantom{a}}
  {\fSubType{\fEmptyT{\fDbEdgeType}}{\fDbEdgeType[\fDbGNames,-]}}
\end{mathpar}

\end{RulesDisplayCustomFont}

\noindent
\FSubRefl{}/\FSubTrans{} makes the subtyping relation a preorder.
\FSubAny{}/\FSubEmpty{} bound the type system with top $\fAny$ and bottom $\fEmpty$ types.
\FSubList{} lifts types to their list types.
The union rules---\FSubUnionL{}, \FSubUnionR{}, \FSubUnionElim{},
and \FSubUnionCon{}---make $\fSortVars{1}\cup\fSortVars{2}$ operate as
least-upper-bound.
The schema rules---\FSubRefineN{} and \FSubRefineNE{}---connect graph-indexed 
and schema-refined node element types while bounding this sublattice with the annotated 
empty type $\fEmptyT{\fDbNodeType}$ at the bottom. \FSubRefineE{} and 
\FSubRefineEE{} do the same for edge types.

\begin{myDef}[Type Intersection]\label{def:type-intersect}
  Let $\fSortVars{1}$ and $\fSortVars{2}$ be any two types. We define
  their intersection, written $\fSortVars{1}\cap\fSortVars{2}$, as the
  least reflexive and transitive relation closed under the following judgements:
  \begin{RulesDisplayCustomFont}{\small}
    \begin{mathpar}
  \infer[\FInterL]
  {
    \fSubType{\fSortVars{1}\!}{\fSortVars{2}}
  }
  {
    \fSortVars{1}\cap\fSortVars{2}:\fSortVars{1}
  }
  \and
  \infer[\FInterR]
  {
    \fSubType{\fSortVars{2}\!}{\fSortVars{1}}
  }
  {
    \fSortVars{1}\cap\fSortVars{2}:\fSortVars{2}
  }
  \and
  \infer[\FInter]
  {
    \fSubType{\fSortVar\!}{\fSortVars{1}}\\
    \fSubType{\fSortVar\!}{\fSortVars{2}}\\
    \fNSubType{\fSortVars{1}\!}{\fSortVars{2}}\\
    \fNSubType{\fSortVars{2}\!}{\fSortVars{1}}
  }
  {
    \fSortVars{1}\cap\fSortVars{2}:\fSortVar
  }
\end{mathpar}

  \end{RulesDisplayCustomFont}
\end{myDef}

\begin{myLemma}[Absorption]\label{lemma:top-bot-absorb}
%For any type $\fSortVar$: %
% \begin{enumerate}[nosep,leftmargin=2em, topsep=-1pt]
  (1)~$\fSortVar \cup \fAny = \fAny$,
  (2)~$\fSortVar \cup \fEmpty = \fSortVar$,
  (3)~$\fSortVar \cap \fAny = \fSortVar$, and
  (4)~$\fSortVar \cap \fEmpty = \fEmpty$.
% \end{enumerate}
% \vspace{-1.5mm}%
% \begin{proof}%
%   (1)--(2) from $\FSubAny$/$\FSubEmpty$;
%   (3)--(4) from $\FInterL$/$\FInterR$ on $\fAny$/$\fEmpty$.
% \end{proof}
\end{myLemma}

\begin{myLemma}[Schema-Refined Type Bounds]\label{lemma:top-bot-schema-ref}
  Consider a graph site $\fDbGNames$ with graph schema
  $(\fDbNodeSchema,\fDbEdgeSchema)$.
  Then
  \(
  \fSubType{\fEmptyT{\fDbNodeType}\!}{\fSubType{\fDbNodeType[\fDbGNames,\fDbNodeSchemaElem]\!}{\fDbNodeType}}
  \)
  and
  \(
  \fSubType{\fEmptyT{\fDbEdgeType}\!}{\fSubType{\fDbEdgeType[\fDbGNames,\fDbEdgeSchemaElem]\!}{\fDbEdgeType}}
  \)
  for $\fDbNodeSchemaElem \in \fDbNodeSchema$ and $\fDbEdgeSchemaElem \in \fDbEdgeSchema$,
  respectively.
% \vspace{-2.5mm}%
% \begin{proof}%
%   Immediate from $\FSubRefineN$/$\FSubRefineNE$.
% \end{proof}
\end{myLemma}

\begin{myDef}[Record and Binding Table Schemas]\label{def:record-schema}
  The schema $\fDbRSchema\!: \fAttr \rightharpoonup \fSortVar$ of a record specifies the type for each binding variable in its domain. The schema
  of a binding table is the union (Definition~\ref{def:record-schema-union}) of the
  schemas of all the records in its support.
\end{myDef}

\begin{myDef}[Record Schema Unions and Joins]\label{def:record-schema-compat}\label{def:record-schema-union}
  Let $\fDbRSchemas{1}$ and $\fDbRSchemas{2}$ be two record schemas.
  \begin{enumerate}[nosep, leftmargin=1.5cm]
    \item[\enumParen{$\fCompatU{\fDbRSchemas{1}\hspace{0.3pt}}{\hspace{0.3pt}\fDbRSchemas{2}}$}] They are \emph{union compatible} when
      $\fDom{\fDbRSchemas{1}}=\fDom{\fDbRSchemas{2}}.$
    %% \item They are \emph{disjoint union compatible}, written
    %%   $\fCompatDU{\fDbRSchemas{1}}{\fDbRSchemas{2}}$, when
    %%   $\fDom{\fDbRSchemas{1}}\cap\fDom{\fDbRSchemas{2}}=\varnothing.$
    \item[\enumParen{$\fCompatJ{\fDbRSchemas{1}}{\fDbRSchemas{2}}$}] They are \emph{join compatible} when
      $\forall \fCyVar \in \fDom{\fDbRSchemas{1}}\cap\fDom{\fDbRSchemas{2}}$:
      \begin{itemize}[nosep,leftmargin=2em, topsep=-0.3em]
      \item $\fSubType{\fDbRSchemas{1}(\fCyVar)}{\fDbNodeType[-]}$ (or $\fDbEdgeType[-]$) and
        $\fSubType{\fDbRSchemas{2}(\fCyVar)}{\fDbNodeType[-]}$ (or $\fDbEdgeType[-]$); and
      \item ${\fDbRSchemas{1}(\fCyVar)}\cap{\fDbRSchemas{2}(\fCyVar)}\neq\fEmpty.$
      \end{itemize}
    \item[\enumParen{$\fSchemaUnion{\fDbRSchemas{1}\hspace{0.3pt}\,}{\,\hspace{0.3pt}\fDbRSchemas{2}}$}] If $\fCompatU{\fDbRSchemas{1}}{\fDbRSchemas{2}}$, their
    \emph{union} is:
    \(
    \fSchemaUnion{\fDbRSchemas{1}}{\fDbRSchemas{2}}(\fCyVar) \triangleq\fDbRSchemas{1}(\fCyVar)\cup\fDbRSchemas{2}(\fCyVar), \fCyVar\in\fDom{\fDbRSchemas{1}}.
    \)
  %% \item If $\fCompatDU{\fDbRSchemas{1}}{\fDbRSchemas{2}}$, their
  %%   \emph{disjoint union} is:
  %%   \begin{TypingDisplay}[\small]
  %%   \[
  %%   \fSchemaDUnion{\fDbRSchemas{1}}{\fDbRSchemas{2}}(\fCyVar) \triangleq
  %%   \begin{cases}
  %%     \fDbRSchemas{1}(\fCyVar) &\fCyVar\in\fDom{\fDbRSchemas{1}}\\
  %%     \fDbRSchemas{2}(\fCyVar) &\fCyVar\in\fDom{\fDbRSchemas{2}}
  %%   \end{cases}
  %%   \]
  %%   \end{TypingDisplay}
    \item[\enumParen{$\fSchemaJoin{\fDbRSchemas{1}\hspace{0.9pt}\,}{\,\hspace{0.9pt}\fDbRSchemas{2}}$}] If $\fCompatJ{\fDbRSchemas{1}}{\fDbRSchemas{2}}$, their \emph{join} is:
    \begin{TypingDisplay}[\small]
    \[
    \fSchemaJoin{\fDbRSchemas{1}}{\fDbRSchemas{2}}(\fCyVar) \triangleq
    \begin{cases}
      \fDbRSchemas{1}(\fCyVar) &\fCyVar\in\fDom{\fDbRSchemas{1}}\hspace{0.8pt}\setminus\hspace{0.8pt}\fDom{\fDbRSchemas{2}}\\
      \fDbRSchemas{2}(\fCyVar) &\fCyVar\in\fDom{\fDbRSchemas{2}}\hspace{0.8pt}\setminus\hspace{0.8pt}\fDom{\fDbRSchemas{1}}\\
      \fDbRSchemas{1}(\fCyVar)\cap\fDbRSchemas{2}(\fCyVar) &\fCyVar\in\fDom{\fDbRSchemas{1}}\cap\fDom{\fDbRSchemas{2}}
    \end{cases}
    \]
    \end{TypingDisplay}
  \end{enumerate}
\end{myDef}

\section{Well-formedness of \lang Queries}
\label{sec:typing-rules}

We now use our type system (\S\ref{sec:type-system}) to define 
well-formedness rules for \lang queries and assign
binding table schemas to their results.
\lang's native support for graph schemas (\isoSection{4.13})
allows fine-grained static typing. Schema information flows from 
typing rules for pattern atoms through patterns and queries,
providing static guarantees that are unavailable in schema-free formalisms.

We organize \lang query well-formedness into three layers that mirror 
the syntactic hierarchy of our calculus in Figure~\ref{fig:gql-calculus}: value 
expressions, patterns (\S\ref{sec:typing-rules-patterns}), and 
queries (\S\ref{sec:typing-rules-queries}).
This organization simplifies the type-soundness argument (\S\ref{sec:metatheory}) 
by allowing it to proceed compositionally, lifting guarantees from expressions to 
patterns and ultimately to queries.
We only consider queries where all pattern atoms are named (non-anonymous), \ie 
they declare their binding variables: $\fVar{\fCyPatternAtom}\neq\epsilon$.

\subsection{\lang Value Expressions}
\label{sec:typing-rules-expressions}

%% The complete typing rules for value expressions can be found in the
%% accompanying supplementary material.

Due to space constraints, we only introduce the typing
judgement for \lang value expressions here for clarity when
referencing them in later sections.
We type value expressions under the judgement:
\(
\fExpJudgeParam,
\)
which assigns type~$\fSortVar$ to the
value expression~$\fCyExp$ under the context containing: the 
database world ($\fDbWorldS$), a working graph site
$\fDbGNames$, a record schema ($\fDbRSchema$),
the binding variables used in the expression ($\fDbVarUse$),
the permitted number of aggregate functions ($\Box\in\{0,1\}$), and the reference-mode ($\Diamond\in\{\fSingletonRef,\fGroupRef\}$) 
that captures if the expression is inside or outside a quantifier context.
%
% When an operator is applied to a union-typed argument, the rules use
% dynamic refinement (Appendix~\ref{subsec:typing-coercion}) to check
% that at least one branch is compatible, producing~$\fNull$ for
% incompatible branches---faithfully modeling~\isoSection{\isoDutTest}.
%

\subsection{\lang Patterns}%
\label{sec:typing-rules-patterns}

Typing patterns is one of the most technically challenging aspects of
defining well-formedness of \lang queries because they introduce all
bindings to graph elements, so their typing must statically predict
the schema that the operational semantics produces at runtime.
Moreover, since our \lang formalism includes graphs with schemas,
typing becomes \emph{strictly more precise} than in schema-free
approaches. For instance, a binding variable can be statically
assigned $\fEmptyT{\fSortVar}$ when the schema's label constraints
together with the pattern rule out any valid homomorphism.
These refinements compose in non-trivial ways because:
(a)~\emph{surrounding context}---a pattern list is composed
of several path patterns that may have overlaps in their declared binding
variables, and (b)~\emph{internal context}---a pattern is composed of
several pattern atoms that may similarly have overlaps in their
declared binding variables.
We reconcile the two by structuring our typing rules across different levels
that mirror the grammar of Figure~\ref{fig:gql-calculus}:
(1)~\emph{atom typing} types individual node and edge atoms in
isolation from all context;
(2)~\emph{path pattern typing} composes atoms only under
internal context, staying isolated from surrounding context; and
(3)~\emph{pattern list typing} composes path patterns via
conjunction under surrounding context.
Levels (2) and (3) refine overlapping variables left-to-right in one forward pass without fixed-point
iteration, which may assign wider types than necessary.
We formalize pattern well-formedness using three main judgement
forms.

\begin{TypingDisplay}
\[
\fPatAtomJudge{\fDbWorld}{\fCyPatternAtom}{\fDbRSchema}
\qquad\qquad
\fPatJudgeM{\fDbWorld}{\fCyPattern}{\fDbRSchema}{\Diamond}
\qquad\qquad
\fPatExpJudge{\fDbWorld}{\fCyPatternExp}{\fDbRSchema}
\]
\end{TypingDisplay}

\noindent
The \emph{atom judgement} types a pattern atom $\fCyPatternAtom$, \ie
a node $\fCyPatternNode$ or an edge $\fCyDirection$ atom, in isolation
to produce a singleton binding table schema $\fDbRSchema$.
The \emph{pattern judgement} composes atoms into path patterns. Its
reference mode $\Diamond\in\fSet{\fSingletonRef,\fGroupRef}$ records
whether the pattern sits outside ($\fSingletonRef$) or inside
($\fGroupRef$) a quantifier context.
The \emph{pattern list judgement} composes path patterns using
conjunction ($\fCyPatternAnd{\fCyPatternExp}{\fCyPattern}$).

\MyPara{Atom Typing (Level 1)}\label{sec:atom-typing}%
Atom typing resolves the label expression $\fCyLabelExp$ and property
map $\fPGProp$ of a pattern atom $\fCyPatternAtom$ against the graph schema
\smash{$(\fDbNodeSchema,\fDbEdgeSchema)$}, producing a
schema-refined type.
Two auxiliary judgements perform the filtering and operate uniformly
on node and edge schemas, so we write $\fDbElemSchemaElem$ for a
schema entry of either kind.
The judgement
$\fLabelSJudge{\fDbElemSchema}{\fLbl{\fCyPatternAtom}}{\fDbElemSchema'}$
filters element schemas by structural induction on the label
expression, covering: labels $\fCyNodeLabel$, the wildcard~($\fCyLWild$),
conjunction~($\fCyLAnd{}{}$), disjunction~($\fCyLOr{}{}$), and
negation~($\hspace{2pt}\fCyLNeg{}$). The judgement
$\fPropSJudge{\fDbElemSchema}{\fProp{\fCyPatternAtom}}{\fDbElemSchema'}$
similarly filters element schemas, but by property maps.
Intersecting their results yields the compatible element schemas.

\begin{RulesDisplayCustomFont}{\small}
\begin{mathpar}
  \infer[\RuleNameFont\FTySLEmpty]
  {
    \phantom{i}
  }
  {
    \fLabelSJudge{\fDbElemSchema}{\varepsilon}{\fDbElemSchema}
  }
  \and
  \infer[\RuleNameFont\FTySLAtom]
  {
    \fDbElemSchema[\fCyLabel]\fAssign\{\fDbElemSchemaElem\in\fDbElemSchema\mid\fCyLabel\in\fSchLbls{\fDbElemSchemaElem}\}
  }
  {
    \fLabelSJudge{\fDbElemSchema}{\fCyLabel}{\fDbElemSchema[\fCyLabel]}
  }
  \and
  \infer[\RuleNameFont\FTySLNot]
  {
    \fLabelSJudge{\fDbElemSchema}{\fPGNodelabel}{\fDbElemSchema[l]}
  }
  {
    \fLabelSJudge{\fDbElemSchema}{\fCyLNeg{\fPGNodelabel}}{\fDbElemSchema\setminus\fDbElemSchema[l]}
  }
  \and
  \infer[\RuleNameFont\FTySLAnd]
  {
    \fLabelSJudge{\fDbElemSchema}{\fPGNodelabel[1]}{\fDbElemSchema[1]}\\
    \fLabelSJudge{\fDbElemSchema}{\fPGNodelabel[2]}{\fDbElemSchema[2]}
  }
  {
    \fLabelSJudge{\fDbElemSchema}{\fCyLAnd{\fPGNodelabel[1]}{\,\fPGNodelabel[2]}}{\fDbElemSchema[1]\cap\fDbElemSchema[2]}
  }
  \and
  \infer[\RuleNameFont\FTySLOr]
  {
    \fLabelSJudge{\fDbElemSchema}{\fPGNodelabel[1]}{\fDbElemSchema[1]}\\
    \fLabelSJudge{\fDbElemSchema}{\fPGNodelabel[2]}{\fDbElemSchema[2]}
  }
  {
    \fLabelSJudge{\fDbElemSchema}{\fCyLOr{\fPGNodelabel[1]}{\,\fPGNodelabel[2]}}{\fDbElemSchema[1]\cup\fDbElemSchema[2]}
  }
  \and
  \infer[\RuleNameFont\FTySLWildcard]
  {
    \fDbElemSchema[\%]\fAssign\{\fDbElemSchemaElem\in\fDbElemSchema\mid\fSchLbls{\fDbElemSchemaElem}\neq\varnothing\}
  }
  {
    \fLabelSJudge{\fDbElemSchema}{\fCyLWild}{\fDbElemSchema[\%]}
  }
  \and
  \infer[\RuleNameFont\FTySPAtom]
  {
    \fDbElemSchema[\fDbPropSchema]\fAssign\{\fDbElemSchemaElem\in\fDbElemSchema\mid\fDbPropSchema\subseteq\fSchProps{\fDbElemSchemaElem}\}
  }
  {
    \fPropSJudge{\fDbElemSchema}{\fDbPropSchema}{\fDbElemSchema[\fDbPropSchema]}
  }
\end{mathpar}

\end{RulesDisplayCustomFont}

\noindent
\FTySLEmpty{} returns the schema set as is since there is no label expression. 
\FTySLAtom{} retains schemas whose label set
contains $\fCyLabel$, while \FTySLWildcard{} retains schemas with at
least one label. \FTySLNot{}, \FTySLAnd{}, and \FTySLOr{} filter 
schemas based on logical operations corresponding to their names.
\FTySPAtom{} filters by property maps.

\begin{RulesDisplayCustomFont}{\small}
\begin{mathpar}
  \infer[\RuleNameFont\FPropEmpty]
  {
    \phantom{\fCyProp}      
  }
  {
    \fPropJudge{}{\varnothing}{\varnothing}
  }
  \and
  \infer[\RuleNameFont\FPropAtom]
  {
    \fCyProp\in\fAttr\\
    \fSortVar\in\fMSortPG\\
    \fCyExpConst{}\in\fUniverse{\fSortVar}
  }
  {
    \fPropJudge{}{[\fCyProp\mapsto\fCyExpConst{}]}{[\fCyProp\mapsto\fSortVar]}
  }
  \and
  \infer[\RuleNameFont\FPropInsert]
  {
    \fCyProp\not\in\fDom{\fPGProp}\\
    \fPropJudge{}{\fPGProp}{\fDbPropSchema[\fPGProp]}\\
    \fPropJudge{}{[\fCyProp\mapsto\fCyExpConst{}]}{\fDbPropSchema[k]}
  }
  {
    \fPropJudge{}{\fPGProp\sqcup[\fCyProp\mapsto\fCyExpConst{}]}{\fSchemaDUnion{\fDbPropSchema[\fPGProp]}{\fDbPropSchema[k]}}
  }
\end{mathpar}

\end{RulesDisplayCustomFont}

\noindent
The judgement
$\fPropJudge{}{\fPGProp}{\fDbPropSchema}$ types a property map $\fPGProp$ with 
a property schema $\fDbPropSchema$.
\FPropEmpty{} types the empty map, \FPropAtom{} types a singleton map
by inferring the constant's scalar type, and \FPropInsert{}
combines property schemas via disjoint union provided their domains are also disjoint.

% \begin{RulesDisplayCustomFont}{\small}
% \input{formalizations/gql-node}
% \end{RulesDisplayCustomFont}

\begin{RulesDisplayCustomFont}{\small}
  \begin{mathpar}
  \infer[\RuleNameFont\FTySAtomLabelPropFresh]
  {
    \fVar{\fCyPatternAtom} \in \fAttr\\
    \fPropJudge{}{\fProp{\fCyPatternAtom}}{\fDbPropSchema}\\
    %% \fPropJudge{\fDbWorld;\varnothing}
    %%            {\fPGProp}
    %%            {\fDbPropSchema[\fPGProp]}\\
    \fLabelSJudge{\fDbElemSchema}{\fLbl{\fCyPatternAtom}}{\fDbElemSchema[\textsf{l}]}\\
    \fPropSJudge{\fDbElemSchema}{\fDbPropSchema}{\fDbElemSchema[\fDbPropSchema]}
  }
  {
    \fSAtomJudge{\fDbElemSchema}
               {\fCyPatternAtom}
               {\fDbElemSchema[\textsf{l}]\cap\fDbElemSchema[\fDbPropSchema]}
  }
  \and
    \infer[\RuleNameFont\FTyPNodeLabelPropOpen]
    {
      \fDbGNames \not\in \fDom{\fDbSchema}\\
      \fVar{\fCyPatternNode}\in\fAttr
    }
    {
      \fPatAtomJudge{\fDbWorld}
                 {\fCyPatternNode}
                 {[\fVar{\fCyPatternNode} \mapsto \fDbNodeType]} 
    }
    \and
    \infer[\RuleNameFont\FTyPNodeLabelPropClosed]
    {
      \fDbGNames \in \fDom{\fDbSchema}\\
      (\fDbNodeSchema, -)=\fDbSchema(\fDbGNames)\\
      \fSAtomJudge{\fDbNodeSchema}
                 {\fCyPatternNode}
                 {\fDbNodeSchema[N]}%
      % \\
      % \fDbNodeSchema\neq\varnothing
    }
    {
      \fPatAtomJudge{\fDbWorld}
                 {\fCyPatternNode}
                 {[\fVar{\fCyPatternNode} \mapsto \fDbNodeType[\fDbGNames,{\fDbNodeSchema[N]}]]} 
    }
    % \quad
    % \infer[\RuleNameFont\FTyPNodeLabelPropClosedFail]
    % {
    %   \fDbGNames \in \fDom{\fDbSchema}\\\\      
    %   \fSNodeJudge{\fDbWorld}
    %              {\fCyPatternNode}
    %              {\varnothing}
    % }
    % {
    %   \fPatAtomJudge{\fDbWorld}
    %              {\fCyPatternNode}
    %              {[\fVar{\fCyPatternNode} \mapsto \fEmptyT{\fDbNodeType}]} 
    % }
    \and
    \infer[\RuleNameFont\FTyPEdgeLabelPropOpen]
    {
      \fDbGNames \not\in \fDom{\fDbSchema}\\
      \fVar{\fCyPatternEdge} \in \fAttr
    }
    {
      \fPatAtomJudge{\fDbWorld}
                 {\fCyPatternEdge}
                 {[\fVar{\fCyPatternEdge} \mapsto \fDbEdgeType]} 
    }
    \and
    \infer[\RuleNameFont\FTyPEdgeLabelPropClosed]
    {
      \fDbGNames \in \fDom{\fDbSchema}\\
      (-,\fDbEdgeSchema)=\fDbSchema(\fDbGNames)\\     
      \fSAtomJudge{\fDbEdgeSchema}
                 {\fCyPatternEdge}
                 {\fDbEdgeSchema[E]}%
      %            \\
      % \fDbEdgeSchema\neq\varnothing
    }
    {
      \fPatAtomJudge{\fDbWorld}
                 {\fCyPatternEdge}
                 {[\fVar{\fCyPatternEdge} \mapsto \fDbEdgeType[\fDbGNames,{\fDbEdgeSchema[E]}]]}
    }
    % \and
    % \infer[\RuleNameFont\FTyPEdgeLabelPropClosedFail]
    % {
    %   \fDbGNames \in \fDom{\fDbSchema}\\
    %   \fCyEdgeVar\fAssign\fVar{\fCyPatternEdge}\\\\      
    %   \fSEdgeJudge{\fDbWorld}
    %              {\fCyPatternEdge}
    %              {\varnothing}
    % }
    % {
    %   \fPatAtomJudge{\fDbWorld}
    %              {\fCyPatternEdge}
    %              {[\fCyEdgeVar \mapsto \fEmptyT{\fDbEdgeType}]}
    % }
    %% \and
    %% \infer[\RuleNameFont\FTyPEdgeLabelPropDir]
    %% {
    %%   \fCyPatternEdge=\fCyElemEdgePattern{\fCyEdgeVar\,\fPGNodelabel\,\fPGProp\,\fCyDir}\\
    %%   \fCyDirection\in\{
    %%   \fCyEdgeRight{\fCyPatternEdge},\;
    %%   \fCyEdgeLeft{\fCyPatternEdge},\;
    %%   \fCyEdgeLeftRight{\fCyPatternEdge},\;
    %%   \fCyEdgeRightU{\fCyPatternEdge},\;
    %%   \fCyEdgeLeftU{\fCyPatternEdge},\;
    %%   \fCyEdgeNoneU{\fCyPatternEdge},\;
    %%   \fCyEdgeNone{\fCyPatternEdge}
    %%   \}\\
    %%   \fPatAtomJudge{\fDbWorld}
    %%              {\fCyPatternEdge}
    %%              {\fDbRSchema}
    %% }
    %% {
    %%   \fPatAtomJudge{\fDbWorld}
    %%              {\fCyDirection}
    %%              {\fDbRSchema}
    %% }
\end{mathpar}

  \end{RulesDisplayCustomFont}

\noindent
\FTySAtomLabelPropFresh{} resolves $\fCyPatternAtom$'s schema by intersecting 
its label-expression and property map filter results.
The atom-typing rules handle open and 
closed graphs separately.
Open graphs (\FTyPNodeLabelPropOpen) have no schema, so the atom's binding 
variable $\fVar{\fCyPatternNode}$ is assigned just 
the graph-indexed type~$\fDbNodeType$.
For closed graphs (\FTyPNodeLabelPropClosed), the variable is assigned
the schema-refined type \smash{$\fDbNodeType[\fDbGNames,{\fDbNodeSchema[N]}]$}.
If the atom's label-expression and property map combination does not 
satisfy the constraints of any node schema in 
$\fDbGNames$'s graph schema, \ie \smash{$\fDbNodeSchema[N]\!=\!\varnothing$}; the variable 
is assigned the annotated empty type \smash{$\fEmptyT{\fDbNodeType}$} 
since \smash{$\fDbNodeType[\fDbGNames,\varnothing]\!\triangleq\!\fEmptyT{\fDbNodeType}$} 
(\S~\ref{subsec:typing-nullability}).
Symmetric rules cover edge atoms.

%% \FTyPEdgeLabelPropDir{} lifts any
%% directional wrapper to the same schema, since directionality does not
%% affect the bound variable's type.

\MyPara{Path Pattern Typing (Level 2)}\label{sec:pattern-typing}%
Path pattern typing composes atoms into complete path patterns.
When composing a node--edge--node step 
$\fCyPatternNodes{1}\fCyPatternEdges{2}\fCyPatternNodes{2}$, the types assigned to 
the atoms---$\fCyPatternNodes{1}$, $\fCyPatternEdges{2}$, and $\fCyPatternNodes{3}$---in 
isolation may be overly permissive. This is because \emph{atom typing} (Level 1) types 
atoms with refined element schemas, with refinement based solely on their label-expressions 
and property maps.
While this ``may be'' sufficient for typing the node atoms $\fCyPatternNodes{1}$ and 
$\fCyPatternNodes{2}$, it is not the case for the edge atom $\fCyPatternEdges{2}$,
as edge schemas impose additional constraints beyond label sets and property maps: endpoint 
node schemas and edge directionality (Definition~\ref{def:prop-graph-schemas}). 
Moreover, refining $\fCyPatternEdges{2}$'s type based on the additional constraints 
usually requires refining the types of its endpoints $\fCyPatternNodes{1}$ and 
$\fCyPatternNodes{2}$ as well. 

We formalize path pattern typing such that type refinement for the composing 
node-edge-node step is done jointly, allowing typing to prune infeasible triples 
and thus assign more precise types.
For the path pattern $\fCyPattern\,\fCyDirection\,\fCyPatternNode$,
the judgement
$\fRefineJudge{\fDbWorld;\fDbRSchemas{{\scriptscriptstyle P}};\fDbRSchemas{{\scriptscriptstyle E}};\fDbRSchemas{{\scriptscriptstyle N}}}
              {\fCyPattern\,\fCyDirection\,\fCyPatternNode}
              {\fDbRSchemas{{\scriptscriptstyle P}}',\fDbRSchemas{{\scriptscriptstyle E}}',\fDbRSchemas{{\scriptscriptstyle N}}'}$
jointly refines the types of the binding variables corresponding to the 
rightmost node atom (\ie $\fTail{\cdot}$) of the preceding pattern $\fCyPattern$, the edge atom $\fCyDirection$, 
and the following node atom $\fCyPatternNode$; from their individual table schemas 
(typed in isolation)---$\fDbRSchemas{{\scriptscriptstyle P}}$, $\fDbRSchemas{{\scriptscriptstyle E}}$, and 
$\fDbRSchemas{{\scriptscriptstyle N}}$---to produce the type refined schemas 
$\fDbRSchemas{{\scriptscriptstyle P}}'$, $\fDbRSchemas{{\scriptscriptstyle E}}'$, and $\fDbRSchemas{{\scriptscriptstyle N}}'$.

\begin{myDef}[Type Preserving Update for Schema-Refined Types]\label{def:lift-update}
  Consider a graph site $\fDbGNames$ with graph schema
  $(\fDbNodeSchema,\fDbEdgeSchema)$.
  Let $\fSortVar$ be a schema-refined type $\fDbNodeType[\fDbGNames,-]$ 
  (or $\fDbEdgeType[\fDbGNames,-]$) or its list-lifted version; and 
  let $\fDbElemSchema$, $\fDbElemSchema'$ range over the subsets of 
  $\fDbNodeSchema$ (or $\fDbEdgeSchema$).
  We then define the type preserving update $\fRefineType{\fSortVar}{\fDbElemSchema}$ as:

{
  \small
  \setlength{\abovedisplayskip}{-2mm}%
  \setlength{\belowdisplayskip}{1mm}%
  \begin{align*}
  \fRefineType{\fSortVar}{\fDbElemSchema}\;\triangleq\;
  \begin{cases}
    \fListT\;\bigl(\fBaseType{\fSortVar}[\fDbElemSchema/\fDbElemSchema']\bigr)
      & \text{if } \fSortVar = \fListT\;\fDbNodeType[\fDbGNames,\fDbElemSchema']\;\textbf{or}\;
        \fListT\;\fDbEdgeType[\fDbGNames,\fDbElemSchema']\\[3pt]
    \fSortVar[\fDbElemSchema/\fDbElemSchema']
      & \text{if } \fSortVar = \fDbNodeType[\fDbGNames,\fDbElemSchema']\;\textbf{or}\;
        \fDbEdgeType[\fDbGNames,\fDbElemSchema']
  \end{cases}
  \end{align*}
}

\noindent%
The operation $\fSortVar[\fDbElemSchema/\fDbElemSchema']$ denotes replacing the 
schema set~$\fDbElemSchema'$ of the schema-refined type $\fSortVar$ with the set~$\fDbElemSchema$.
It yields the appropriate annotated empty type when $\fDbElemSchema=\varnothing$
(\eg $\fDbNodeType[\fDbGNames,\fDbNodeSchema'][\varnothing/\fDbNodeSchema']=\fEmptyT{\fDbNodeType}$).
\end{myDef}

%
% To keep the rules compact, we write
% $\fUpToNull{\fSortVar'}{\fSortVar}$ for 
% $\fSubType{\fSortVar}{\fSubType{\fSortVar'}{\fSortVar\fNullable}}$.
% The endpoint condition further relies on the orientation metafunction
% $\fDir{\cdot}$ (Figure~\ref{fig:gql-metafunctions}), which
% maps an edge-pattern direction to the set of edge orientations it
% admits.

\begin{RulesDisplayCustomFont}{\small}
\begin{mathpar}
  \mprset{vskip=1mm}
  {\footnotesize%
    \infer[\RuleNameFont\FTyPPatRefineOpen]
    {
      \fDbGNames \not\in \fDom{\fDbSchema}\\
      \fCompatJ{}{}(\fDbRSchemas{{\scriptscriptstyle P}},\fDbRSchemas{{\scriptscriptstyle E}},\fDbRSchemas{{\scriptscriptstyle N}})
    }
    {
      \fRefineJudge{\fDbWorld;\fDbRSchemas{{\scriptscriptstyle P}};\fDbRSchemas{{\scriptscriptstyle E}};\fDbRSchemas{{\scriptscriptstyle N}}}
                   {\fCyPattern\,\fCyDirection\,\fCyPatternNode}%
                   {\fDbRSchemas{{\scriptscriptstyle P}},\fDbRSchemas{{\scriptscriptstyle E}},\fDbRSchemas{{\scriptscriptstyle N}}}
    }    
  }
    \quad
    {\footnotesize%
    \fEndpointCond{\fDbNodeSchemaElem[1]}{\fDbEdgeSchemaElem[2]}{\fDbNodeSchemaElem[3]}\triangleq%
    \hspace{-4mm}\bigvee\limits_{\fCyDir\,\in\,\fDir{\fCyDirection}}%
    \begin{cases}
      \big(\fDbNodeSchemaElem[1],\fDbNodeSchemaElem[3],\fPGDir\!\!\big)=
      \big(\fSchEnd{\fDbEdgeSchemaElem[2]}{1},\fSchEnd{\fDbEdgeSchemaElem[2]}{2},\fSchDir{\fDbEdgeSchemaElem[2]}\big) & \text{if } \fCyDir=\,\fOrientR\\[-1pt]
      \big(\fDbNodeSchemaElem[1],\fDbNodeSchemaElem[3],\fPGDir\!\!\big)=
      \big(\fSchEnd{\fDbEdgeSchemaElem[2]}{2},\fSchEnd{\fDbEdgeSchemaElem[2]}{1},\fSchDir{\fDbEdgeSchemaElem[2]}\big) & \text{if } \fCyDir=\,\fOrientL\\[-1pt]
      \big(\fDbNodeSchemaElem[1],\fDbNodeSchemaElem[3],\fPGUDir\!\!\big)=
      \big(\fSchEnd{\fDbEdgeSchemaElem[2]}{1},\fSchEnd{\fDbEdgeSchemaElem[2]}{2},\fSchDir{\fDbEdgeSchemaElem[2]}\big) & \multirow{1}{*}{\text{else}}\\[-4pt]
      \multicolumn{1}{c}{\textbf{or}} & \\[-4pt]
      \big(\fDbNodeSchemaElem[1],\fDbNodeSchemaElem[3],\fPGUDir\!\!\big)=
      \big(\fSchEnd{\fDbEdgeSchemaElem[2]}{2},\fSchEnd{\fDbEdgeSchemaElem[2]}{1},\fSchDir{\fDbEdgeSchemaElem[2]}\big)&\\
    \end{cases}
    }
    \and
    % {\footnotesize%
    % \infer[\RuleNameFont\FTyPPatRefine]
    % {
    %   \begin{array}{@{}l@{}}
    %     \fUpToNull{\fBaseType{\fDbRSchemas{1}(\fCyNodeVars{1})}}{\fDbNodeType[{\fDbGNames,\fDbNodeSchema[1]}]}\cr
    %     \fUpToNull{\fBaseType{\fDbRSchemas{2}(\fCyEdgeVars{2})}}{\fDbEdgeType[{\fDbGNames,\fDbEdgeSchema[2]}]}\cr
    %     \fUpToNull{\fBaseType{\fDbRSchemas{3}(\fCyNodeVars{3})}}{\fDbNodeType[{\fDbGNames,\fDbNodeSchema[3]}]}  
    %   \end{array}\\
    %   \mathcal{S}\fAssign\{(\fDbNodeSchemaElem[1],\fDbEdgeSchemaElem[2],\fDbNodeSchemaElem[3])\in(\fDbNodeSchema[1]\times\fDbEdgeSchema[2]\times\fDbNodeSchema[3])\mid\fEndpointCond{\fDbNodeSchemaElem[1]}{\fDbEdgeSchemaElem[2]}{\fDbNodeSchemaElem[3]}\}\\
    %   % \fDbNodeSchema[1]'\fAssign\fElemAccess{1}{\mathcal{S}}\\
    %   % \fDbEdgeSchema[2]'\fAssign\fElemAccess{2}{\mathcal{S}}\\
    %   % \fDbNodeSchema[3]'\fAssign\fElemAccess{3}{\mathcal{S}}
    %   \fDbRSchema\fAssign (\fSchemaJoin{\fDbRSchemas{1}}{\fSchemaJoin{\fDbRSchemas{2}}{\fDbRSchemas{3}}})
    %   \left[\begin{array}{@{}l@{}}
    %     {\fCyNodeVars{1} \mapsto \fRefineType{\fDbRSchemas{1}(\fCyNodeVars{1})}{\fDbNodeSchema[1]'}},\cr
    %     {\fCyEdgeVars{2} \mapsto \fRefineType{\fDbRSchemas{2}(\fCyEdgeVars{2})}{\fDbEdgeSchema[2]'}},\cr
    %     {\fCyNodeVars{3} \mapsto \fRefineType{\fDbRSchemas{3}(\fCyNodeVars{3})}{\fDbNodeSchema[3]'}}
    %   \end{array}\right]
    % }
    % {
    %   \fRefineJudge{\fDbWorld;\fDbRSchemas{1};\fDbRSchemas{2};\fDbRSchemas{3}}
    %                {\fCyNodeVars{1},\fCyEdgeVars{2},\fCyNodeVars{3},\fCyDirection}
    %                {\fDbRSchema}
    % }
    % }
    \infer[\RuleNameFont\FTyPPatRefine]
    {
      \fDbGNames \in \fDom{\fDbSchema}\\
      \fCompatJ{}{}(\fDbRSchemas{{\scriptscriptstyle P}},\fDbRSchemas{{\scriptscriptstyle E}},\fDbRSchemas{{\scriptscriptstyle N}})\\
      (\fSortVars{1},\fSortVars{2},\fSortVars{3})\fAssign%
      (\fDbRSchemas{{\scriptscriptstyle P}}(\fTail{\fCyPattern}),\,
      \fDbRSchemas{{\scriptscriptstyle E}}(\fVar{\fCyDirection}),\,
      \fDbRSchemas{{\scriptscriptstyle N}}(\fVar{\fCyPatternNode}))\\\\
      \fUpToNull{\fBaseType{\fSortVars{1}}}{\fDbNodeType[{\fDbGNames,\fDbNodeSchema[1]}]}\quad
      \fUpToNull{\fBaseType{\fSortVars{2}}}{\fDbEdgeType[{\fDbGNames,\fDbEdgeSchema[2]}]}\quad
      \fUpToNull{\fBaseType{\fSortVars{3}}}{\fDbNodeType[{\fDbGNames,\fDbNodeSchema[3]}]}\qquad
      \mathcal{S}\fAssign\{(\fDbNodeSchemaElem[1],\fDbEdgeSchemaElem[2],\fDbNodeSchemaElem[3])\!\in\!(\fDbNodeSchema[1]\times\fDbEdgeSchema[2]\times\fDbNodeSchema[3])\mid\fEndpointCond{\fDbNodeSchemaElem[1]}{\fDbEdgeSchemaElem[2]}{\fDbNodeSchemaElem[3]}\}
    }
    {
      \fRefineJudge{\fDbWorld;\fDbRSchemas{{\scriptscriptstyle P}};\fDbRSchemas{{\scriptscriptstyle E}};\fDbRSchemas{{\scriptscriptstyle N}}}
                   {\fCyPattern\,\fCyDirection\,\fCyPatternNode}
                   {\fDbRSchemas{{\scriptscriptstyle P}}\big[{\fTail{\fCyPattern} \mapsto \fRefineType{\fSortVars{1}}{\fElemAccess{1}{\mathcal{S}}}}\big],\;
                    \fDbRSchemas{{\scriptscriptstyle E}}\big[{\fVar{\fCyDirection} \mapsto \fRefineType{\fSortVars{2}}{\fElemAccess{2}{\mathcal{S}}}}\big],\;
                    \fDbRSchemas{{\scriptscriptstyle N}}\big[{\fVar{\fCyPatternNode} \mapsto \fRefineType{\fSortVars{3}}{\fElemAccess{3}{\mathcal{S}}}}\big]%
                   }
    }
\end{mathpar}

\end{RulesDisplayCustomFont}

\noindent
\FTyPPatRefineOpen{} handles type refinement in open graphs and
simply checks pairwise join compatibility (Definition~\ref{def:record-schema-union}, 
via shorthand:\;$\fCompatJ{}{}(\dots)$) of the individual schemas 
corresponding to the preceding pattern $\fCyPattern$, edge atom $\fCyDirection$, and 
following node atom $\fCyPatternNode$; before returning them unchanged.
\FTyPPatRefine{} handles closed graphs. 
Let $\fSortVars{1}$, $\fSortVars{2}$, and $\fSortVars{3}$ correspond to 
the types assigned in isolation to the binding variables: 
$\fTail{\fCyPattern}$, $\fVar{\fCyDirection}$, and $\fVar{\fCyPatternNode}$, respectively.
We first check if the types are the appropriate schema-refined types up to nullability, 
using the shorthand: 
$\fUpToNull{\fSortVar'\!}{\!\fSortVar}\triangleq\fSubType{\fSortVar\!}{\fSubType{\fSortVar'\!}{\fSortVar\fNullable}}$.
Next, we form the Cartesian product of their schema sets, retaining only the triples
$(\fDbNodeSchemaElem[1],\fDbEdgeSchemaElem[2],\fDbNodeSchemaElem[3])$
that satisfy the \emph{additional constraints} imposed by edge schemas ($\fDbEdgeSchemaElem[2]$) via 
the endpoint condition $\fEndpointCond{\fDbNodeSchemaElem[1]}{\fDbEdgeSchemaElem[2]}{\fDbNodeSchemaElem[3]}$.
%
% If any atom carries an empty type ($\fDbNodeType[-,\varnothing]$ or
% $\fDbEdgeType[-,\varnothing]$), the Cartesian product is empty and all
% three variables are refined to their empty-type counterparts via
% Definition~\ref{def:lift-update}.
%
The endpoint condition holds when the node-edge-node schema triple is 
compatible under some orientation in $\fDir{\fCyDirection}$ (Figure~\ref{fig:gql-metafunctions}). 
Each orientation ($\fCyDir\in\{\fOrientR,\fOrientL,\fOrientU\}$) compares the endpoint node schemas $\fDbNodeSchemaElem[1]$ and 
$\fDbNodeSchemaElem[3]$ against that of the edge schema $\fDbEdgeSchemaElem[2]$ 
as a tuple, with ordering determined by the orientation; the edge directionality $\fSchDir{\fDbEdgeSchemaElem[2]}$
must match $\fDirDenote{\fCyDir}$.
Orientation ``$\fOrientR$'' takes the endpoint node schemas in order, 
``$\fOrientL$'' swaps them, but both require the edge to be directed; while ``$\fOrientU$'' compares them 
as an unordered pair for undirected edges. Finally, the filtered schema triples 
are used to update the variables' types via the type-preserving update from Definition~\ref{def:lift-update}.

\begin{myDef}[Quantifier Lift]%
  \label{def:quant-lift}
  Given a record schema~$\fDbRSchema$ and a quantifier $\fCyQuantifier$,
  the quantifier lift operation $\fLiftQuant{\fCyQuantifier}{\fDbRSchema}$
  is defined as: $\forall\fCyVar\in\fDom{\fDbRSchema}$:

  {
    \small
    \setlength{\abovedisplayskip}{-2mm}%
    \setlength{\belowdisplayskip}{1mm}%
    \begin{align*}
    \fLiftQuant{\fCyQuantifier}{\fDbRSchema}(\fCyVar) \;\triangleq\;
    \begin{cases}
      \fDbRSchema(\fCyVar)\fNullable
        & \text{if } \fCyQuantifier=\fCyQuantQues\\
      \fListT\;\fDbRSchema(\fCyVar)
        & \text{otherwise }
    \end{cases}
    \end{align*}
  }

% \noindent
\end{myDef}

Path patterns are described in \isoSectionShort{\isoPathPattern}, and 
\isoSyntxRulePg{\isoSyntxPatEdgeCount}{\isoPgPatEdgeCount} states that patterns 
may only be quantified if they have a minimum path length of one. We enforce this
during path pattern typing via the premise $\fMinE{\fCyPattern}>0$. \isoSyntxRulePg{\isoSyntxPatNodeCount}{\isoPgPatNodeCount}
imposes another restriction of path patterns having at least one node atom regardless of quantification. We account for this while
typing pattern lists.
%
% Binding variables within a quantified context may accumulate multiple bindings per match, and 
% so are lifted to list or nullable types,
% based on whether the quantifier is the optional quantifier~$\fCyQuantQues$ or not.

\begin{RulesDisplayCustomFont}{\small}
  \begin{mathpar}
    \infer[\RuleNameFont\FTyPNodeLabelPropLift]
    {
      \phantom{a}\\\\
      \fPatAtomJudge{\fDbWorld}
                 {\fCyPatternNode}
                 {\fDbRSchema}
    }
    {
      \fPatJudgeM{\fDbWorld}
                 {\fCyPatternNode}
                 {\fDbRSchema}
                 {\Diamond}
    }
    \and
    \infer[\RuleNameFont\FTyPPatPattern]
    {      
      \fPatJudgeM{\fDbWorld}
                 {\fCyPattern}
                 {\fDbRSchemas{{\scriptscriptstyle P}}}
                 {\Diamond}\\
      \fPatJudgeM{\fDbWorld}
                 {\fCyPatternNode}
                 {\fDbRSchemas{{\scriptscriptstyle N}}}
                 {\Diamond}\\
      \fRefineJudge{\fDbWorld;\fDbRSchemas{{\scriptscriptstyle P}};\fDbRSchemas{{\scriptscriptstyle E}};\fDbRSchemas{{\scriptscriptstyle N}}}
      {\fCyPattern\,\fCyDirection\,\fCyPatternNode}
      {\fDbRSchemas{{\scriptscriptstyle P}}',\fDbRSchemas{{\scriptscriptstyle E}}',\fDbRSchemas{{\scriptscriptstyle N}}'}
    }
    {
      \fPatJudgeP{\fDbWorld;\fDbRSchemas{{\scriptscriptstyle E}}}
                {\fCyPattern\,\fCyDirection\,\fCyPatternNode}
                {\fSchemaJoin{\fDbRSchemas{{\scriptscriptstyle P}}'}
                             {\fSchemaJoin{\fDbRSchemas{{\scriptscriptstyle E}}'}
                                          {\fDbRSchemas{{\scriptscriptstyle N}}'}
                             }
                }
                {\Diamond}
    }
    \and
    \infer[\RuleNameFont\FTyPPatEdge]
    {
      \fPatAtomJudge{\fDbWorld}
                    {\fCyDirection}
                    {\fDbRSchemas{{\scriptscriptstyle E}}}\\\\
      \fPatJudgeP{\fDbWorld;\fDbRSchemas{{\scriptscriptstyle E}}}
                {\fCyPattern\,\fCyDirection\,\fCyPatternNode}
                {\fDbRSchema}
                {\Diamond}
    }
    {
      \fPatJudgeM{\fDbWorld}
                {\fCyPattern\,\fCyDirection\,\fCyPatternNode}
                {\fDbRSchema}
                {\Diamond}
    }
    \quad
    \infer[\RuleNameFont\FTyPPatQuantEdge]
    {
      \fPatAtomJudge{\fDbWorld}
                    {\fCyDirection}
                    {\fDbRSchemas{{\scriptscriptstyle E}}}\\\\
      \fPatJudgeP{\fDbWorld;(\fLiftQuant{\fCyQuantifier}{\fDbRSchemas{{\scriptscriptstyle E}}})}
                {\fCyPattern\,\fCyDirection\,\fCyPatternNode}
                {\fDbRSchema}
                {\fGroupRef}
    }
    {
      \fPatJudgeM{\fDbWorld}
                {\fCyPattern\,\fCyDirection\fCyQuantifier\,\fCyPatternNode}
                {\fDbRSchema}
                {\fSingletonRef}
    }    
    \quad
    \infer[\RuleNameFont\FTyPPatParenPath]
    {
      \phantom{a}\\\\
      \fPatJudgeM{\fDbWorld}
                 {\fCyPattern}
                 {\fDbRSchema}
                 {\Diamond}
    }
    {
      \fPatJudgeM{\fDbWorld}
                 {\fCyParenthesis{\fCyPattern}}
                 {\fDbRSchema}
                 {\Diamond}
    }
    \quad
    \infer[\RuleNameFont\FTyPPatQuantPath]
    {
      \fMinE{\fCyPattern}>0\\\\
      \fPatJudgeM{\fDbWorld}
                 {\fCyParenthesis{\fCyPattern}}
                 {\fDbRSchema}
                 {\fGroupRef}
    }
    {
      \fPatJudgeM{\fDbWorld}
                 {\fCyParenthesis{\fCyPattern}\,\fCyQuantifier}
                 {\fLiftQuant{\fCyQuantifier}{\fDbRSchema}}
                 {\fSingletonRef}
    }
\end{mathpar}

\end{RulesDisplayCustomFont}

\noindent
\FTyPNodeLabelPropLift{} lifts a node atom as a path pattern.
\FTyPPatPattern{} is an auxiliary judgement that type refines the 
binding variables of a path pattern using the refinement rules 
described earlier.
\FTyPPatEdge{} introduces an edge atom, delegating to \FTyPPatPattern{}.
\FTyPPatQuantEdge{} handles quantified edges by first lifting their 
singleton table schemas through $\fLiftQuant{\fCyQuantifier}{(\cdot)}$ from 
Definition~\ref{def:quant-lift} before delegating to \FTyPPatPattern{} (similar to \FTyPPatEdge{}).
\FTyPPatParenPath{} types parenthesized path patterns by simply typing the enclosed pattern.
\FTyPPatQuantPath{} handles quantified path patterns by
checking if they have at least one edge atom, and then typing them, before
lifting their table schemas through $\fLiftQuant{\fCyQuantifier}{(\cdot)}$.

\MyPara{Pattern List Typing (Level 3)}\label{sec:pattern-expression-typing}%
Pattern list typing composes path patterns via conjunction
($\fCyPatternAnd{\fCyPatternExp}{\fCyPattern}$).
We account for the minimum node count restriction imposed by 
\isoSyntxRulePg{\isoSyntxPatNodeCount}{\isoPgPatNodeCount} here, via 
$\fMinN{\fCyPattern}>0$.
% Types assigned to shared variables must be schema-join compatible
% (Definition~\ref{def:record-schema-union}).

\begin{RulesDisplayCustomFont}{\small}
\begin{mathpar}
    \infer[\RuleNameFont\FTyPPatExpSingle]
    {
      \fMinN{\fCyPattern}>0\\
      \fPatJudgeM{\fDbWorld}
                {\fCyPattern}
                {\fDbRSchema}
                {\fSingletonRef}
    }
    {
      \fPatExpJudge{\fDbWorld}
                   {\fCyPattern}
                   {\fDbRSchema}
    }
    \and
    \infer[\RuleNameFont\FTyPPatExpAnd]
    {
      \fPatExpJudge{\fDbWorld}
                   {\fCyPatternExp}
                   {\fDbRSchemas{\scriptscriptstyle p}}\\
      \fPatExpJudge{\fDbWorld}
                 {\fCyPattern}
                 {\fDbRSchemas{\scriptscriptstyle P}}\\
      \fCompatJ{\fDbRSchemas{\scriptscriptstyle p}}{\fDbRSchemas{\scriptscriptstyle P}}          
    }
    {
      \fPatExpJudge{\fDbWorld}
                   {\fCyPatternAnd{\fCyPatternExp}{\fCyPattern}}
                   {\fSchemaJoin{\fDbRSchemas{\scriptscriptstyle p}}{\fDbRSchemas{\scriptscriptstyle P}}}
    }
    %% \and
    %% \infer[\RuleNameFont\FTyPPatExpOr]
    %% {
    %%   \fPatExpJudge{\fDbWorld}
    %%                {\fCyPatternExp}
    %%                {\fDbRSchemas{1}}\\
    %%   \fPatJudgeM{\fDbWorld}
    %%              {\fCyPattern}
    %%              {\fDbRSchemas{2}}
    %%              {\fSingletonRef}
    %% }
    %% {
    %%   \fPatExpJudge{\fDbWorld}
    %%                {\fCyPatternOr{\fCyPatternExp}{\fCyPattern}}
    %%                {\fSchemaNullUnion{\fDbRSchemas{1}}{\fDbRSchemas{2}}}
    %% }
\end{mathpar}

\end{RulesDisplayCustomFont}

\noindent
\FTyPPatExpSingle{} lifts a path pattern into a pattern list if it has 
at least one node atom. 
\FTyPPatExpAnd{} types a conjunction. It types both
sub-components---the preceding path pattern list $\fCyPatternExp$
and the following path pattern $\fCyPattern$---independently, and verifies
the join compatibility of their table schemas
$\fDbRSchemas{{\scriptscriptstyle p}}$ and
$\fDbRSchemas{{\scriptscriptstyle P}}$
before joining them to produce a single table schema for the
conjunction. Since path pattern lists are typed under
surrounding context in one forward left-to-right pass, only those
variables declared in $\fCyPatternExp$ that overlap with the ones in
$\fCyPattern$ are refined during the schema join
\(
\fSchemaJoin{\fDbRSchemas{\scriptscriptstyle p}}
            {\fDbRSchemas{\scriptscriptstyle P}}.
\)

\subsection{\lang Queries}%
\label{sec:typing-rules-queries}

\lang queries are either focused linear queries ($\fCyQuery$) that
match patterns against a single working graph and optionally filter
the matched bindings before projecting them via projection
expressions; or composite queries ($\fCyQueryExp$) that compose
several focused linear queries via composite operators ($\fCyCompOp$).

Query typing is relatively simple since much of the complexity has
been deliberately offloaded to typing: (1)~value expressions
(\S\ref{sec:typing-rules-expressions})---used for typing the predicate
and projection expressions in $\fWhere$ and $\fReturn$ clauses,
respectively; and (2)~patterns
(\S\ref{sec:typing-rules-patterns})---used for typing the patterns in
$\fMatch$ clauses.
We type \lang queries under three judgement forms,

\begin{TypingDisplay}
\[
\underset{\scriptsize\text{Projection}}{\fProjJudge{\fDbWorld;\fDbRSchema}{\fCyProjection}{\fDbRSchemas{\scriptscriptstyle \fCyProjection}}}
\qquad\qquad\quad
\underset{\scriptsize\text{Query}}{\fQueryJudge{\fDbWorldS}{\fCyQuery}{\fDbRSchema}}
\qquad\qquad\quad
\underset{\scriptsize\text{Composite Query}}{\fCompQueryJudge{\fDbWorldS}{\fCyQueryExp}{\fDbRSchema}}
\]
\vspace{-2mm}%
\end{TypingDisplay}

\noindent%
The \emph{projection judgement} types projection expressions by projecting an
upstream table schema in the context of a working graph. The projected
table schema is the schema of a focused linear query's result binding table.
The \emph{query judgement} assigns focused linear queries with the
schemas of their result binding tables.
The \emph{composite query judgement} types composite queries
similarly, but under the additional context of the composite
operators used for focused linear query composition.

\MyPara{Projection Typing}%
Projection expressions project upstream table schemas ($\fDbRSchema$) using
value expressions ($\fCyExp$).
Projection expressions are covered in \isoSectionShort{\isoQueryProjection},
and its \isoSyntxRulePg{\isoSyntxPrjAlias}{\isoPgPrjAlias} states that they
must be aliased via $\fAs$ unless they are just references to binding
variables ($\fCyVar\in\fDom{\fDbRSchema}$). We choose the fragment for
projection expressions and structure its typing rules to match the standard.

\begin{RulesDisplayCustomFont}{\small}
\begin{mathpar}
  \infer[\RuleNameFont\FTyProjAs]
  {
    \phantom{a}\\\\
    \fCyVar\in\fAttr\\
    \fExpJudgeParam[\fDbWorld;\fDbRSchema]
                   [\fCyExp][\fSortVar][1][\fGroupRef][\fDbVarUse]
  }
  {
    \fProjJudge{\fDbWorld;\fDbRSchema}
               {\fCyExp\;\fAs\;\fCyVar}
               {[\fCyVar \mapsto \fSortVar]}
  }
  \and
  \infer[\RuleNameFont\FTyProj]
  {
    \phantom{a}\\\\
    \fProjJudge{\fDbWorld;\fDbRSchema}
               {\fCyVar\;\fAs\;\fCyVar}
               {\fDbRSchemas{\scriptscriptstyle x}}
  }
  {
    \fProjJudge{\fDbWorld;\fDbRSchema}
               {\fCyVar}
               {\fDbRSchemas{\scriptscriptstyle x}}
  }
  \and
  \infer[\RuleNameFont\FTyProjList]
  {        
    \fProjJudge{\fDbWorld;\fDbRSchema}
               {\fCyProjection}
               {\fDbRSchemas{\scriptscriptstyle \fCyProjection}}\\
    \fProjJudge{\fDbWorld;\fDbRSchema}
               {\_\fCyProjection}
               {\fDbRSchemas{\scriptscriptstyle \_\fCyProjection}}\\\\
    \fDom{\fDbRSchemas{\scriptscriptstyle \fCyProjection}}\cap\fDom{\fDbRSchemas{\scriptscriptstyle \_\fCyProjection}}=\varnothing
  }
  {
    \fProjJudge{\fDbWorld;\fDbRSchema}
               {\fCyProjection\fCyComma\_\fCyProjection}
               {\fSchemaJoin{\fDbRSchemas{\scriptscriptstyle \fCyProjection}}{\fDbRSchemas{\scriptscriptstyle \_\fCyProjection}}}
  }
\end{mathpar}

\end{RulesDisplayCustomFont}

\noindent
\FTyProjAs{} covers value expressions generally, and so compulsorily
includes an alias. It types value expressions in the context of an
upstream table schema, and produces a singleton table schema with the
alias mapped to that type.
Value expressions that are just binding variable references are
handled by \FTyProj{}, which first rewrites them vacuously into
aliased projection expressions, before delegating them to
\FTyProjAs{}.
\FTyProjList{} covers projection composition by typing their atoms
left-to-right, joining their schemas only when they are disjoint.

\MyPara{Linear Query Typing}%
Linear queries first select a working graph via
$\fUse$, then introduce binding tables via $\fMatch$ and optionally
filter them via $\fWhere$, before projecting them using $\fReturn$.

\begin{RulesDisplayCustomFont}{\small}
\begin{mathpar}
  \infer[\RuleNameFont\FTyMatchWhereRet]
  {
    \fDbGNames\in\fDom{\fDbCatalog}\\  
    \fPatExpJudge{\fDbWorld}
                 {\fCyPatternExp}
                 {\fDbRSchemas{\scriptscriptstyle p}}\\\\
    \fExpJudgeParam[\fDbWorld;\fDbRSchemas{\scriptscriptstyle p}]
                   [\fCyPred][\fBool\fNullable][1][\fSingletonRef][\fDbVarUse]\\
    \fProjJudge{\fDbWorld;\fDbRSchemas{\scriptscriptstyle p}}
               {\fCyProjection}
               {\fDbRSchema}
  }
  {
    \fQueryJudge{\fDbWorldS}{\fUse\;\fDbGNames\;\fMatch\;\fCyPatternExp\;\fWhere\;\fCyPred\;\fReturn\;\fCyProjection}{\fDbRSchema}
  }
  \quad
  \infer[\RuleNameFont\FTyMatchRet]
  {
    \phantom{a}\\\\
    \fQueryJudge{\fDbWorldS}{\fUse\;\fDbGNames\;\fMatch\;\fCyPatternExp\;\fWhere\;\fTrue\;\fReturn\;\fCyProjection}{\fDbRSchema}
  }
  {
    \fQueryJudge{\fDbWorldS}{\fUse\;\fDbGNames\;\fMatch\;\fCyPatternExp\;\fReturn\;\fCyProjection}{\fDbRSchema}
  }
\end{mathpar}

\end{RulesDisplayCustomFont}

\noindent
\FTyMatchWhereRet{} checks if the specified graph $\fDbGNames$
is in the catalog $\fDbCatalog$, and types the pattern expression
$\fCyPatternExp$ of the $\fMatch$ clause to obtain its table schema
$\fDbRSchemas{\scriptscriptstyle p}$, which is used for typing the
predicate of the $\fWhere$ clause as a nullable Boolean in
singleton-reference context; and the projection expression of the
$\fReturn$ clause to obtain the result table schema $\fDbRSchema$.
\FTyMatchRet{} first rewrites with a tautological predicate~$\fTrue$
for $\fWhere$, and then delegates to \FTyMatchWhereRet{}.

\MyPara{Composite Query Typing}%
Composite queries compose linear queries under composite operators
such as $\fUnion$ and $\fExceptJ$, \etc with an optional qualifier
($\fDstnct$ v. $\fAll$) for choosing set or bag semantics for the
operation.
Composite queries are covered in \isoSectionShort{\isoCompositeQuery},
and its \isoSyntxRulePg{\isoSyntxCompOp}{\isoPgCompOp} states that
the composite operators used within a single composite query must
all be identical (including the qualifier that decides set v.
bag semantics). We design our typing rules to match this.

\begin{RulesDisplayCustomFont}{\small}
\begin{mathpar}
  \infer[\RuleNameFont\FTyCompQueryLift]
  {
    \fQueryJudge{\fDbWorldS}{\fCyQuery}{\fDbRSchema}\\
  }
  {
    \fCompQueryJudge{\fDbWorldS}{\fCyQuery}{\fDbRSchema}
  }
  \and
  \infer[\RuleNameFont\FTyCompQuery]
  {
    \fCompQueryJudge{\fDbWorldS}{\fCyQueryExp}{\fDbRSchemas{\scriptscriptstyle q}}\\
    \fCompQueryJudge{\fDbWorldS}{\fCyQuery}{\fDbRSchemas{\scriptscriptstyle Q}}\\
    \fCompatC{\fDbRSchemas{\scriptscriptstyle q}}{\fDbRSchemas{\scriptscriptstyle Q}}
  }
  {
    \fCompQueryJudge{\fDbWorldS}{\fCyQueryExp\;\fCyCompOp\;\fCyQuery}{\fDbRSchemas{\scriptscriptstyle q}\,\fCyCompOp\,\fDbRSchemas{\scriptscriptstyle Q}}
  }
\end{mathpar}

\end{RulesDisplayCustomFont}

\noindent
\FTyCompQueryLift{} lifts a linear query into the composite judgement.
\FTyCompQuery{} covers query composition by typing the linear query
atoms in the context of the composite operator $\fCyCompOp$, 
checking compatibility, and combining their table schemas
based on the operator's semantics.
Since the judgement carries the composite operator in its
context, it automatically enforces the \lang standard's requirement of
a single composite query using the same composite operator throughout.

\section{Small-Step Operational Semantics}
\label{sec:semantics-operational}

We defined well-formedness of \lang queries in \S\ref{sec:typing-rules} by 
structuring our typing rules compositionally, \ie lifting context from 
query-level down to pattern- and expression-level, and typing
them with binding table schemas.
We now formally define how these queries are executed.
While most prior works have formalized query language semantics 
denotationally in general, including \lang~\cite{deutsch2022graph},
we use a small-step operational approach. This allows us to explicitly 
show the complexities of interleaving states during query execution, arising
from the non-trivial semantics of \lang.

Since \lang allows queries to specify multiple ways of traversing graphs for 
pattern matching (\isoSectionShort{\isoPatternMatching}), we first fix the 
variant that we formalize.
Graph pattern matching semantics is configured primarily, and jointly, through 
two parameters: \emph{path modes} 
(\isoSectionShort{\isoPatternMatchingPathMode}) and \emph{match modes} 
(\isoSectionShort{\isoPatternMatchingMatchMode}).
The path modes impose constraints on matched edge/node frequencies for 
individual path patterns ($\fCyPattern$), while the match modes control 
whether these constraints are enforced across the entire pattern list ($\fCyPatternExp$).
Since every successful pattern match implies a new record containing bindings 
to the matched elements, the action of path and match modes can be understood as
filtering the binding tables of the composing path patterns and pattern lists, 
respectively.
% %
% Every record in a path pattern's binding table represents (always) a contiguous
% path in a graph, but this is not necessarily true for pattern lists.
%
Path modes filter a path pattern's records, which always represent 
contiguous paths in a graph, by:
(1)~\pathModeTrail{}---graph edges may not be bound more than once;
(2)~\pathModeAcyclic{}---graph nodes may not be bound more than once, 
\ie matched paths must be acyclic;
(3)~\pathModeSimple{}---similar to \pathModeAcyclic{} except, graph nodes 
bound to the head/tail variables of a path pattern may repeat, \ie matched paths 
may be cyclic; and
(4)~\pathModeWalk{}---no constraints.
Match modes filter a pattern list's records by:
(1)~\matchModeDiff{}---graph edges may not be bound more than once; and
(2)~\matchModeRepeat{}---no additional constraints.

We use the \pathModeTrail{} path mode and \matchModeDiff{} match mode
in our formalization, as they are the default in popular graph 
query languages like \cypher, and because they guarantee termination
of pattern matching, \ie graph traversal during pattern matching terminates (regardless of quantifiers) since
traversed paths cannot repeat edges and hence are 
upper-bounded by the graph's edge count.

% %
% This design exposes three challenges implicit in the standard:
% %
% \begin{itemize}[nosep,leftmargin=2em]
%   \item \emph{Three-valued null propagation.}
%     The standard identifies null with Unknown (\iso~\S4.16.2) and
%     mandates Kleene three-valued logic for all predicates
%     (\iso~\S20.20); a $\fWhere$ clause retains only records where the
%     condition evaluates to $\fTrue$ (\iso~\S16.13).

%   \item \emph{Trail-aware pattern composition.}
%     Under the \textsc{Trail} path mode (\iso~\S4.11.7), no edge may
%     appear more than once in a path---a global constraint that cuts
%     across local atom matches, requiring every join to check
%     edge-disjointness.

%   \item \emph{Finiteness of quantified paths.}
%     Unbounded quantifiers ($\fCyQuantStar$, $\fCyQuantPlus$) are
%     admissible only with a restrictive path mode (\iso~\S16.4);
%     under \textsc{Trail}, iteration is bounded by~$\fSize{\fPGEdge}$.
% \end{itemize}
% %
% \noindent
We start by introducing the runtime values and syntax used in our 
formalization (\S\ref{sec:semantics-runtime-dom}); before 
moving onto the actual \lang semantics, which similar to our typing 
rules, is organized into three layers mirroring the calculus in 
Figure~\ref{fig:gql-calculus}:
(1)~\emph{Value Expressions} (not shown due to space constraints);
(2)~\emph{Patterns}
(\S\ref{sec:sem-pattern})---graph pattern matching; and
(3)~\emph{Queries}
(\S\ref{sec:sem-op-query})---the full query evaluation pipeline.
%
%% All three share the runtime domain of
%% \S\ref{sec:semantics-runtime-dom} and define step relations
%% $\mathcal{E}\vdash t\,\fOpStep\,t'$ whose transitive--reflexive
%% closure $\fOpStepStar$ reduces a source term to its final value.

\subsection{Runtime Values, Execution Constructs, and Runtime Syntax}
\label{sec:semantics-runtime-dom}

\newcommand{\rtPanel}[1]{\fbox{\footnotesize #1}}
\newcommand{\rtRule}{\vspace{3pt}\hrule\vspace{3pt}}
\newcommand{\rtGloss}[1]{\text{\emph{\footnotesize #1}}}
\newcommand{\rtPanelTight}[1]{%
  \setlength{\fboxsep}{1.8pt}%
  \fbox{\footnotesize #1}%
}

\begin{figure}[t]
  \centering
  \footnotesize
  \setlength{\tabcolsep}{2.5pt}
  \setlength{\extrarowheight}{2pt}
  \begin{tabular}{@{}p{0.228\linewidth}|p{0.324\linewidth}|p{0.425\linewidth}@{}}
    \rtPanel{Runtime Values ($\fDbRVal\!\in\!\fValSet$)} &
    \rtPanel{Execution Constructs} &
    \rtPanelTight{Runtime Syntax}\\[6pt]

    $\begin{array}{@{}r@{\,}c@{\;}l@{\quad}l@{}}
       \fDbRVal & \fTermDef & \fCyExpConst{}      & \rtGloss{constant} \\
                & \mid      & \fNull              & \rtGloss{null} \\
                & \mid      & \fPGNodeElem        & \rtGloss{node} \\
                & \mid      & \fPGEdgeElem        & \rtGloss{edge} \\
                & \mid      & \fList{\fDbRVal,\dots,\fDbRVal}
                                                  & \rtGloss{finite list}
     \end{array}$
    &
    $\begin{array}{@{}l@{\colon}l@{\quad}l@{}}
       \fCyTuple      & \fAttr\rightharpoonup\fValSet
                      & \rtGloss{record} \\
       \fBagVar       & \fCyTuple\rightarrow\mathbb{Z}_{\geq 0}
                      & \rtGloss{bind. table} \\
       \fDbBindTable  & \fCyTuple\times\fPowerSet{\fPGEdge}\rightarrow\mathbb{Z}_{\geq 0}
                      & \rtGloss{trail table} \\
       \fPathTable    & \fCyTuple\!\times\!\fPGNode^{2}\!\times\!\fPowerSet{\fPGEdge}\rightarrow\mathbb{Z}_{\geq 0}
                      & \rtGloss{path table}
     \end{array}$
     &
     $\begin{array}{@{}r@{\;}l@{\quad}l@{}}
      \fRuntimeStx{\fCyPattern} &\fTermDef \fCyPatternNode
           \mid \fRuntimeStx{\fCyPattern}\,\fCyDirection\,\fRuntimeStx{\fCyPattern}
           \mid \fPatConcat{\fRuntimeStx{\fCyPattern}}{\fRuntimeStx{\fCyPattern}}
           \mid \fRuntimeStx{\fCyQuantifier}
           \mid \fPathTable
           & \rtGloss{path pattern} \\
       \fRuntimeStx{\fCyQuantifier}
         &\fTermDef \fQFrame{\fRuntimeStx{\fCyPattern}}{\fQuantVars}{\fCyQuantifier}
           \mid \fQState{\fPathTable[A]}{\fVisitPath}{\fQuantIt}
           \mid \fPathTable
         & \rtGloss{quantified path} \\
       \fRuntimeStx{\fCyPatternExp} &\fTermDef \fRuntimeStx{\fCyPattern}
           \mid \fDbBindTable
           \mid \fCyPatternAnd{\fRuntimeStx{\fCyPatternExp}}{\fRuntimeStx{\fCyPattern}}
           \mid \fCyPatternAnd{{\fDbBindTable}}{{\fDbBindTable}}
         & \rtGloss{pattern list} \\
       \fRuntimeStx{\fCyQuery} &\fTermDef \fCyQuery
           \mid\langle\fPropGraph,\fRuntimeStx{\fCyPatternExp},\fCyPred,\fCyProjection\rangle
           \mid\langle\fPropGraph,\fBagVar,\fCyProjection\rangle
         & \rtGloss{query} \\
       \fRuntimeStx{\fCyQueryExp} &\fTermDef \fRuntimeStx{\fCyQuery}
           \mid\fBagVar
           \mid \fRuntimeStx{\fCyQueryExp}\,\fCyCompOp\,\fRuntimeStx{\fCyQuery}
           \mid \fBagVar\,\fCyCompOp\,\fBagVar
         & \rtGloss{comp. query}
     \end{array}$
  \end{tabular}

  \rtRule

  \setlength{\abovedisplayskip}{0pt}%
  \setlength{\belowdisplayskip}{0pt}%

  \begin{minipage}[t]{0.578\linewidth}%
  \[
  \begin{array}{@{}r@{\,}c@{\,}l@{}}
    \multicolumn{3}{@{}l@{}}{\rtPanelTight{Lifting and Lowering between Execution Constructs}}\\[1.5mm]
    \fPTEdgeLift{\fPropGraph}{\fCyDirection}{\fBagVar}&\triangleq&
      \fBag{(\fCyTuple,\fPGNodeElems{\texttt{lft}},\fPGNodeElems{\texttt{rht}},\fSet{\fPGEdgeElem})
            \,\Bigg|\,
            \begin{array}{@{}l@{}}
            \fCyTuple\in\fBagVar,
            \fPGEdgeElem\fAssign\fCyTuple(\fVar{\fCyDirection}),\\
            (\fPGNodeElems{\texttt{lft}},\fPGNodeElems{\texttt{rht}})
              \in\fEdgeEnds{\fPropGraph}{\fCyDirection}{\fPGEdgeElem}
            \end{array}}\\[4mm]
    \fEdgeEnds{\fPropGraph}{\fCyDirection}{\fPGEdgeElem}&\triangleq&
        \hspace{-2mm}\bigcup\limits_{\substack{\fCyDir\,\in\,\fDir{\fCyDirection}\\
                            \fPGDirFunc(\fPGEdgeElem)\,=\,\fDirDenote{\fCyDir}}}
        \begin{cases}
          \fSet{\fPGSrcDstFunc(\fPGEdgeElem)}
            & \text{if } \fCyDir = \fOrientR \\[-2pt]
          \fSet{\bigl(
            \fElemAccess{2}{\fPGSrcDstFunc(\fPGEdgeElem)},
                \fElemAccess{1}{\fPGSrcDstFunc(\fPGEdgeElem)}\bigr)}
            & \text{if } \fCyDir = \fOrientL \\[-2pt]
          \fSet{(\fPGNodeElems{1},\fPGNodeElems{2}),
                (\fPGNodeElems{2},\fPGNodeElems{1})
                \big| \fSet{\fPGNodeElems{1},\fPGNodeElems{2}}\!=\!\fPGSrcDstFunc(\fPGEdgeElem)
                }
            & \text{if } \fCyDir = \fOrientU
        \end{cases}
  \end{array}
  \]%
  \end{minipage}%
  \hfill
  \begin{minipage}[t]{0.388\linewidth}%
  \[
    \begin{array}{@{}r@{\,}c@{\,}l@{}}
      \multicolumn{3}{@{}l@{}}{}\\[1.5mm]
      \fPTNodeLift{\fPropGraph}{\fCyPatternNode}{\fBagVar}&\triangleq&
      \fBag{(\fCyTuple,\fPGNodeElem,
             \fPGNodeElem,\varnothing)
             \;\Bigg|\;
             \begin{array}{@{}l@{}}
             \fCyTuple\in\fBagVar,\\
             \fPGNodeElem\fAssign\fCyTuple(\fVar{\fCyPatternNode})
             \end{array}
             }\\[6mm]
      \fPTDrop{\fPathTable}&\triangleq&
      \fBag{(\fCyTuple,\fVisitEdge)\mid
            (\fCyTuple,\fPGNodeElems{\texttt{lft}},\fPGNodeElems{\texttt{rht}},\fVisitEdge)
            \in\fPathTable}\\[2mm]
    \fBTDrop{\fDbBindTable}&\triangleq&
      \fBag{\fCyTuple\mid(\fCyTuple,\fVisitEdge)\in\fDbBindTable}
    \end{array}
    \]%
  \end{minipage}%

  \rtRule

  \begin{minipage}[t]{0.60\linewidth}%
  \[
  \begin{array}{@{}r@{\,}c@{\,}l@{}}
    \multicolumn{3}{@{}l@{}}{\rtPanelTight{Pattern Normalization}}\\[1.5mm]
    \fNormOf{\fCyPatternExp}&\triangleq&
    \begin{cases}
      \fAnonFill{\fCyPatternNode}
        & \text{if } \fCyPatternExp\!=\!\fCyPatternNode\\[-1pt]
      \fNormOf{\fCyPattern}\,\fAnonFill{\fCyDirection}\,
        \fAnonFill{\fCyPatternNode}
        & \text{if } \fCyPatternExp\!=\!
          \fCyPattern\,\fCyDirection\,\fCyPatternNode\\[-1pt]
      \fPatConcat{\fNormOf{\fCyPattern}}
        {\fPatConcat{\fQFrame{\fRuntimeStx{\fCyPattern}_{\fCyQuantifier}}
                             {\fVars{\fRuntimeStx{\fCyPattern}_{\fCyQuantifier}}}
                             {\fCyQuantifier}}
                    {\fAnonFill{\fCyPatternNode}}}
        & \text{if } \fCyPatternExp\!=\!
          \fCyPattern\,\fCyDirection\fCyQuantifier\,\fCyPatternNode\\[-2pt]
      \fNormOf{\fCyPattern}
        & \text{if } \fCyPatternExp\!=\!\fCyParenthesis{\fCyPattern}\\[-1pt]
      \fQFrame{\fNormOf{\fCyParenthesis{\fCyPattern}}}
              {\fVars{\fNormOf{\fCyParenthesis{\fCyPattern}}}}
              {\fCyQuantifier}
        & \text{if } \fCyPatternExp\!=\!
          \fCyParenthesis{\fCyPattern}\,\fCyQuantifier\\[-1pt]
      \fCyPatternAnd{\fNormOf{\fCyPatternExp'}}{\fNormOf{\fCyPattern}}
        & \text{if } \fCyPatternExp\!=\!
          \fCyPatternAnd{\fCyPatternExp'}{\fCyPattern}
    \end{cases}\\[1.5mm]
    \multicolumn{3}{@{}l@{}}{
      \text{where }
      \fRuntimeStx{\fCyPattern}_{\fCyQuantifier}\triangleq
        \fAnonFill{\fCyParenthesis{\epsilon\epsilon\epsilon}}\,\fAnonFill{\fCyDirection}\,
        \fAnonFill{\fCyParenthesis{\epsilon\epsilon\epsilon}}}
  \end{array}
  \]%
  \end{minipage}%
  \hfill\vrule\hfill
  \begin{minipage}[t]{0.37\linewidth}%
  \[
    \begin{array}{@{}r@{\,}c@{\,}l@{}}
      \multicolumn{3}{@{}l@{}}{\rtPanelTight{Quantifier Upper Bound}}\\[1mm]
      \fQHi{\fCyQuantifier}&\triangleq&
      \begin{cases}
        \fSize{\fPGEdge}
          & \text{if } \fCyQuantifier\!\in\!\fSet{\fCyQuantStar,\fCyQuantPlus}\\[-5pt]
        1 & \text{if } \fCyQuantifier\!=\!\fCyQuantQues\\[-5pt]
        i & \text{if } \fCyQuantifier\!=\!\fCyQuantExact{i}\\[-5pt]
        j & \text{if } \fCyQuantifier\!=\!\fCyQuantBound{i}{j}
      \end{cases}\\[5mm]
      \multicolumn{3}{@{}l@{}}{\rtPanelTight{Rename Anonymous Atoms}}\\[1mm]
      \fAnonFill{\fCyPatternAtom}&\triangleq&
      \begin{cases}
        \fCyPatternAtom
          & \text{if } \fVar{\fCyPatternAtom}\neq\epsilon\\[-4pt]
        \fCyPatternAtom[\fCyVars{a}/\epsilon]
          & \text{otherwise}
      \end{cases}\\[1.5mm]
      \multicolumn{3}{@{}l@{}}{
        \text{where } \fCyVars{a}\in\fAttrAnon \text{ fresh}}
    \end{array}
    \]%
  \end{minipage}%
  \caption{\FigCaptionGqlRuntime}
\end{figure}

Figure~\ref{fig:runtime-entities} summarizes the runtime 
values, execution constructs, runtime syntax, and metafunctions
used in our semantics formalization.
The \emph{runtime value domain} $\fValSet$ extends the
multi-sorted value universe $\fUniverse{}$ from
\S\ref{sec:prelims-pgm} with property graph 
elements---nodes ($\fPGNodeElem$) and edges ($\fPGEdgeElem$)---and 
finite lists.
The former ranges over the node/edge sets of all the graphs
in a database instance's catalog $\fDbCatalog$, while the latter is constructed
at runtime for accumulating bindings across quantified pattern
iterations. Lists are also paired with the standard \emph{list concatenation} 
operation ``$\fListConcat{}$'' for concatenating two lists.

\begin{myDef}[Records and Binding Tables]\label{def:records-binding-tables}
  A \emph{record} $\fCyTuple\colon\fAttr\rightharpoonup\fValSet,$ is a
  partial function from binding variables to runtime values.
  Two records are \emph{compatible}: $\fCyTuples{1}\asymp\fCyTuples{2}$, when they agree on their common
  domain; their natural join: $\fTupleJoin{\fCyTuples{1}}{\fCyTuples{2}}$
  is then their union as partial functions, and $\bot$ otherwise.
  A \emph{binding table} $\fBagVar: \fCyTuple\rightarrow\mathbb{Z}_{\geq 0},$ 
  is a bag of records $\fCyTuple$; their join is the join
  of all compatible record pairs.
  
\end{myDef}

\noindent%
The \lang standard specifies two \emph{execution constructs}: 
records ($\fCyTuple$) and binding tables ($\fBagVar$), which are
used throughout the entire query evaluation pipeline to hold 
intermediate as well as the final results (Definition~\ref{def:records-binding-tables}).
However, since we consider the \pathModeTrail{} path mode 
under the \matchModeDiff{} matching mode for our pattern matching semantics,
we introduce two more execution constructs to simplify formalizing the
additional constraints they impose. 
The first is the \emph{trail table} ($\fDbBindTable$),
which simplifies enforcing the \matchModeDiff{} matching mode constraint
across the entire pattern list. Similar to binding tables, trail tables
are also bags of records, but each record is also annotated with a set of
graph edges ($\fVisitEdge$) bound in that record. So edge-disjointness across trail tables
corresponding to different path patterns in a pattern list can be trivially 
enforced by checking the intersection of the bound-edge sets of their 
respective records.
The second is the \emph{path table} ($\fPathTable$), for simplifying 
enforcing the \pathModeTrail{} constraint when matching path patterns.
Unlike pattern lists, the bindings in a single path pattern match 
correspond to a single path in the working graph. So the path table builds 
on top of the trail table entries by annotating them with the graph nodes 
corresponding to the endpoints of the matched path. This simplifies 
path composition to a trivial endpoint continuity check, while 
the same bound-edge sets from the trail tables can be reused for the
\pathModeTrail{} constraint.

\begin{myDef}[Trail Tables]\label{def:trail-monoid}
  A \emph{trail table entry} is a tuple $(\fCyTuple,\fVisitEdge)$ of a 
  record $\fCyTuple$ and the edges $\fVisitEdge\in\fPowerSet{\fPGEdge}$ bound in 
  that record, \ie a bound-edge set.
  A \emph{trail table} $\fDbBindTable$ is a bag of trail table entries.
  We equip trail tables with a commutative monoid $(\mathsf{T},\;\fBTJoinOp\,,\;\fBTUnit)$,
  where $\fBTUnit\!\triangleq\!\fBag{([\,],\,\varnothing)}$; 
  for implicitly enforcing edge-disjointness when joining them. 
  The \emph{trail product} $\!\fBTTrailProd{}{}\!$ combines two trail table 
  entries only when their records join and their 
  bound-edge sets are disjoint, and $\fBTJoinOp$ lifts it to trail tables:
  
  {
    \small
    \setlength{\abovedisplayskip}{-3mm}%
    \setlength{\belowdisplayskip}{0mm}%
    \begin{align*}
      \fDbBindTable[1]\!\fBTJoinOp\!\fDbBindTable[2]
      \!\triangleq\!
      \fBag{
        \fBTTrailProd{t_1}{t_2}
        \Big|\;
        \begin{array}{@{}l@{}}
        t_1\!\in\!\fDbBindTable[1],\;
        t_2\!\in\!\fDbBindTable[2],\\
        \fBTTrailProd{t_1}{t_2}\neq\bot
        \end{array}
      }
      \quad
      \fBTTrailProd{(\fCyTuples{1},\fVisitEdge[1])\!}{\!(\fCyTuples{2},\fVisitEdge[2])}
      \!\triangleq\!
      \begin{cases}
        \multirow{2}{*}{$(\fTupleJoin{\fCyTuples{1}}{\fCyTuples{2}},
         \fVisitEdge[1]\cup\fVisitEdge[2])$}
          & \text{if }\fTupleJoin{\fCyTuples{1}}{\fCyTuples{2}}\neq\bot
            \textbf{ and }\\[-2pt]
          & \fVisitEdge[1]\cap\fVisitEdge[2]=\varnothing \\[-2pt]
        \bot & \text{otherwise}
      \end{cases}
    \end{align*}
  }

\end{myDef}

\begin{myDef}[Path Tables]\label{def:path-tables}
  A \emph{path entry} is a tuple
  $(\fCyTuple,\fPGNodeElems{\texttt{lft}},\fPGNodeElems{\texttt{rht}},\fVisitEdge)$
  containing a record, the matched path's endpoint nodes,
  and its bound-edges. A \emph{path table} $\fPathTable$ is a bag of
  path entries. 
  Path tables compose under $\fPTJoinRec{\oplus}$, which implicitly enforces
  endpoint continuity and edge-disjointness. We parameterize it with the 
  operation $\oplus$ that defines how to combine the records of the two entries:

  {
    \small
    \setlength{\abovedisplayskip}{-3mm}%
    \setlength{\belowdisplayskip}{0mm}%
    \begin{align*}
      \fPathTable[1]\fPTJoinRec{\oplus}\fPathTable[2]
      \!\triangleq\!
      \fBag{
        (\fCyTuples{1}\oplus\fCyTuples{2},\;
         \fPGNodeElems{\texttt{lft}1},\;
         \fPGNodeElems{\texttt{rht}2},\;
         \fVisitEdge[1]\cup\fVisitEdge[2])
        \;\Bigg|\;
        \begin{array}{@{}l@{}}
        (\fCyTuples{i},\;\fPGNodeElems{\texttt{lft}i},\;\fPGNodeElems{\texttt{rht}i},\;
          \fVisitEdge[i])\in\fPathTable[i]\;(i\in\{1,2\}),\\
        \fCyTuples{1}\oplus\fCyTuples{2}\neq\bot,\quad
        \fPGNodeElems{\texttt{rht}1}=\fPGNodeElems{\texttt{lft}2},\quad
        \fVisitEdge[1]\cap\fVisitEdge[2]=\varnothing
        \end{array}}
      .
    \end{align*}
  }

  \noindent%
  The parametrization is important because although the semantics of path 
  composition remains unchanged between quantified and unquantified path patterns
  (it only depends on endpoint continuity and edge-disjointness), how variables
  may bind to nodes/edges differs. Unquantified 
  variables cannot bind to more than one node/edge per record, while 
  quantified variables may bind to multiple depending on the 
  quantifier used. Parametrization allows us to enforce these differences by 
  defining how records combine during path composition.
  \emph{Path concatenation} $\fPTJoinOp\triangleq\fPTJoinRec{\Join}$
  instantiates $\oplus$ with the natural join of records, while 
  Definition~\ref{def:quant-record-composition} defines it
  for quantified path patterns.
\end{myDef}

\begin{myDef}[Quantified Path Repetition]\label{def:quant-record-composition}
  Let $\fQuantVars=\fVars{\fCyPattern}$ be the variables declared by a
  quantified path and let $\fCyQuantifier$ be its quantifier. The
  quantifier repetition is bounded below by
  $\fQLo{\fCyQuantifier}$ (Figure~\ref{fig:gql-metafunctions}) and above
  by $\fQHi{\fCyQuantifier}$ (Figure~\ref{fig:runtime-entities}), which
  caps the unbounded quantifiers by the number of graph edges (\pathModeTrail{} 
  constraint of edge-disjointness).
  Variables declared inside the quantifier accumulate bindings by the
  \emph{repetition extension}
  $\fQJoinOp{\fCyQuantifier}\triangleq\fPTJoinRec{\fQRecExt}$, starting
  from the \emph{zero-repetition table}
  $\fQZero{\fQuantVars}{\fCyQuantifier}$; for
  $\fCyVar\in\fQuantVars$:

  {
    \small
    \setlength{\abovedisplayskip}{-3mm}%
    \setlength{\belowdisplayskip}{0mm}%
    \begin{align*}
      % \fQZeroRec(\fCyVar)
      % \triangleq
      % \begin{cases}
      %   \fNull & \text{if }\fCyQuantifier=\fCyQuantQues\\
      %   \fList{} & \text{otherwise}
      % \end{cases}
      % \;
      \fQZero{\fQuantVars}{\fCyQuantifier}
      \!\triangleq\!
      \fBag{\left(\left[\fCyVar\mapsto
      \begin{cases}
        \fNull & \text{if }\fCyQuantifier=\fCyQuantQues\\[-2pt]
        \fList{} & \text{else}
      \end{cases}
      \right]_{\fCyVar\in\fQuantVars},\fPGNodeElem,\fPGNodeElem,\varnothing\right)
            \Bigg|\,\fPGNodeElem\in\fPGNode}
      \quad
      (\fCyTuples{1}\!\fQRecExt\!\fCyTuples{2})(\fCyVar)
      \!\triangleq\!
      \begin{cases}
        \fCyTuples{2}(\fCyVar)
          & \text{if }\fCyQuantifier=\fCyQuantQues\\[-2pt]
        \fListConcat{\fCyTuples{1}(\fCyVar)\!}{\!\fList{\fCyTuples{2}(\fCyVar)}}
          & \text{else}
      \end{cases}
    \end{align*}
  }

  \noindent%
  The resulting group representation is the runtime counterpart of the 
  quantifier lift $\fLiftQuant{\fCyQuantifier}{(\cdot)}$
  (Definition~\ref{def:quant-lift}). It is list-valued for every 
  quantifier other than $\fCyQuantQues$, which admits a single repetition 
  and keeps the binding itself on successful match, or $\fNull$ otherwise.
  The zero-repetition table is a path table pairing every node in the graph 
  with itself, so joining it with $\fPathTable$ under 
  $\fQJoinOp{\fCyQuantifier}$ leaves endpoints and bound-edge sets unchanged 
  and only lifts each record into its group representation.
  Therefore, $\fQZero{\fQuantVars}{\fCyQuantifier}\fQJoinOp{\fCyQuantifier}\fPathTable$
  is the group representation of one repetition, and repeating it
  $\fQuantIt-1$ times further, yields the matches that use exactly
  $\fQuantIt$ repetitions.
\end{myDef}

%% \subsection{Expression Semantics}
%% The expression step relation
%% $\fSemExpJudge{\fPropGraph;\fCyTuple}{\fCyExp}{\fCyExp'}$ reduces
%% value expressions left-to-right under a record~$\fCyTuple$ and
%% property graph~$\fPropGraph$.
%% %
%% Every compound form follows a uniform three-rule pattern: congruence
%% rules reduce sub-expressions left-to-right until all operands are
%% values, a computation rule produces the result, and a companion null
%% rule fires when operands fall outside the expected domain,
%% yielding~$\fNull$.
%% %
%% This discipline captures the \iso standard's null-propagation semantics
%% (\iso~\S4.16.2, \S20.20).%% , including Kleene three-valued logic for all
%% %% logical connectives and the identification of $\fNull$ with the truth
%% %% value Unknown.
%% %% %
%% Expression reduction rules are discussed in Appendix~\ref{sec:sem-op-expr}.

\subsection{Pattern Semantics}
\label{sec:sem-pattern}

Pattern matching is responsible for introducing bindings into 
the query pipeline. A $\fMatch$ clause enumerates every homomorphism of its
pattern into the working graph, which the standard specifies as a
multi-phase pipeline (\isoSectionShort{\isoPathEval})---each path
pattern is evaluated independently, followed by forcing agreement 
on shared variables via natural equijoins on the cross products of 
their binding tables, producing the final result table.
We formalize pattern matching using three main judgements.

\begin{TypingDisplay}
\[
\fSAtomJudgeN{\fCyPatternAtom}{\fBagVar}
\qquad\quad
\fSPathJudge{\fRuntimeStx{\fCyPattern}}{\fRuntimeStx{\fCyPattern}'}
\qquad\quad
% \fSQPathJudge{\fRuntimeStx{\fCyQuantifier}}{\fRuntimeStx{\fCyQuantifier}'}
% \quad\;\;
\fPEStep{\fRuntimeStx{\fCyPatternExp}}{\fRuntimeStx{\fCyPatternExp}'}
\]
\end{TypingDisplay}

\noindent
The \emph{atom judgement} matches a pattern atom $\fCyPatternAtom$, \ie
a node $\fCyPatternNode$ or an edge $\fCyDirection$ atom, in isolation
to produce a binding table $\fBagVar$ of singleton records. It is
big-step ($\fOpBigStep$) as it only depends on the working graph.
The \emph{path judgement} reduces path patterns
(Figure~\ref{fig:runtime-entities}) to path tables with endpoint 
continuity and edge-disjointness enforced.
The \emph{pattern list judgement} reduces pattern lists
to trail tables by composing path patterns using conjunction
($\fCyPatternAnd{\fCyPatternExp}{\fCyPattern}$), while enforcing 
edge-disjointness throughout.

\MyPara{Atom Matching (Level 1)}
The first step in atom matching is to determine if the labels of a graph 
element $\fPGElem$ (node or edge) satisfy the label expression $\fCyLabelExp$.
We use an auxiliary judgement $\fLabelOp{\fPGElem}{\fCyLabelExp}{\fCyExpBConst{}}$
for this, which evaluates to a Boolean 
$\fCyExpBConst{}\in\fSet{\fTrue,\fFalse}$ indicating whether the labels of 
$\fPGElem$ satisfy $\fCyLabelExp$.

\begin{RulesDisplayCustomFont}{\small}
  \begin{mathpar}
  \infer[\RuleNameFont\FSLabelEmpty]
  {
    \phantom{x}
  }
  {
    \fLabelOp{\fPGElem}{\varepsilon}{\fTrue}
  }
  \and
  \infer[\RuleNameFont\FSLabelAtom]
  {
    \fCyLabel\in\fPGLabelFunc(\fPGElem)
  }
  {
    \fLabelOp{\fPGElem}{\fCyLabel}{\fTrue}
  }
  \and
  \infer[\RuleNameFont\FSLabelAtomFail]
  {
    \fCyLabel\notin\fPGLabelFunc(\fPGElem)
  }
  {
    \fLabelOp{\fPGElem}{\fCyLabel}{\fFalse}
  }
  \and
  \infer[\RuleNameFont\FSLabelWild]
  {
    \fPGLabelFunc(\fPGElem)\neq\varnothing
  }
  {
    \fLabelOp{\fPGElem}{\fCyLWild}{\fTrue}
  }
  \and
  \infer[\RuleNameFont\FSLabelWildFail]
  {
    \fPGLabelFunc(\fPGElem)=\varnothing
  }
  {
    \fLabelOp{\fPGElem}{\fCyLWild}{\fFalse}
  }
  \\
  \infer[\RuleNameFont\FSLabelNeg]
  {
    \fLabelOp{\fPGElem}{\fCyLabelExp}{\fCyExpBConst{}}
  }
  {
    \fLabelOp{\fPGElem}{\fCyLNeg{\fCyLabelExp}}{\neg\fCyExpBConst{}}
  }
  \and
  \infer[\RuleNameFont\FSLabelAnd]
  {
    \fLabelOp{\fPGElem}{\fCyLabelExp_1}{\fCyExpBConst{1}}\\
    \fLabelOp{\fPGElem}{\fCyLabelExp_2}{\fCyExpBConst{2}}
  }
  {
    \fLabelOp{\fPGElem}{\fCyLAnd{\fCyLabelExp_1}{\fCyLabelExp_2}}
             {\fCyExpBConst{1}\land\fCyExpBConst{2}}
  }
  \and
  \infer[\RuleNameFont\FSLabelOr]
  {
    \fLabelOp{\fPGElem}{\fCyLabelExp_1}{\fCyExpBConst{1}}\\
    \fLabelOp{\fPGElem}{\fCyLabelExp_2}{\fCyExpBConst{2}}
  }
  {
    \fLabelOp{\fPGElem}{\fCyLOr{\fCyLabelExp_1}{\fCyLabelExp_2}}
             {\fCyExpBConst{1}\lor\fCyExpBConst{2}}
  }
\end{mathpar}

\end{RulesDisplayCustomFont}

\noindent
The rules follow the same structure as their counterparts in atom typing (\S\ref{sec:atom-typing}),
which filter schemas.

\begin{RulesDisplayCustomFont}{\small}
  \begin{mathpar}
  \infer[\RuleNameFont\FSPatAtomNode]
  {
    \fBagVar\!\fAssign\!\fBag{
      [\fVar{\fCyPatternNode}\mapsto\fPGNodeElem]
      \Bigg|\;
      \begin{array}{@{}l@{}}
      \fPGNodeElem\!\in\!\fPGNode,\\
      \fProp{\fCyPatternNode}\subseteq\fPGPropFunc(\fPGNodeElem),\\
      \fLabelOp{\fPGNodeElem}{\fLbl{\fCyPatternNode}}{\fTrue}
      \end{array}
     }
  }
  {
    \fSAtomJudgeN{\fCyPatternNode}
                 {\fBagVar}
  }
  \and
  \infer[\RuleNameFont\FSPatAtomEdge]
  {
    \fBagVar\!\fAssign\!\fBag{[\fVar{\fCyDirection}\mapsto\fPGEdgeElem]
      \Bigg|\;
      \begin{array}{@{}l@{}}
      \fPGEdgeElem\!\in\!\fPGEdge,\,
      \fProp{\fCyDirection}\subseteq\fPGPropFunc(\fPGEdgeElem),\\
      \fLabelOp{\fPGEdgeElem}{\fLbl{\fCyDirection}}{\fTrue},\\
        \fPGDirFunc(\fPGEdgeElem)\in\fDirDenote{\fDir{\fCyDirection}}
      \end{array}
     }
  }
  {
    \fSAtomJudgeE{\fCyDirection}
                 {\fBagVar}
  }
\end{mathpar}

\end{RulesDisplayCustomFont}

\noindent
\FSPatAtomNode{} collects every graph node whose property map is a superset of
the atom's property map and whose labels satisfy its label expression,
binding each of them to $\fVar{\fCyPatternNode}$ in a singleton record.
\FSPatAtomEdge{} is the symmetric rule for edges, with an additional 
predicate for directionality.

\MyPara{Path Pattern Matching (Level 2)}
Path pattern matching (\isoSectionShort{\isoPathPattern}) reduces a
path pattern to a path table (Definition~\ref{def:path-tables}), leveraging
the latter's implicit enforcement of endpoint continuity and edge-disjointness 
constraints for simplifying the formalization of the former's semantic rules.

\begin{RulesDisplayCustomFont}{\small}
  \begin{mathpar}
  \infer[\RuleNameFont\FSSegNode]
  {
    \fSAtomJudgeN{\fCyPatternNode}
                 {\fBagVar}
  }
  {
    \fSPathJudge{\fCyPatternNode}
                {\fPTNodeLift{\fPropGraph}{\fCyPatternNode}{\fBagVar}}
  }
  \;\;\;
  \infer[\RuleNameFont\FSSegEdgeStepL]
  {
    \fSPathJudge{\fRuntimeStx{\fCyPattern}}
                {\fRuntimeStx{\fCyPattern}'}
  }
  {
    \fSPathJudge{\fRuntimeStx{\fCyPattern}\,\fCyDirection\,\fCyPatternNode}
                {\fRuntimeStx{\fCyPattern}'\fCyDirection\,\fCyPatternNode}
  }
  \;\;\;
  \infer[\RuleNameFont\FSSegEdgeStepR]
  {
    \fSPathJudge{\fCyPatternNode}
                 {\fPathTable[2]}
  }
  {
    \fSPathJudge{\fPathTable[1]\fCyDirection\,\fCyPatternNode}
                {\fPathTable[1]\fCyDirection\,
                 \fPathTable[2]}
  }
  \;\;\;
  \infer[\RuleNameFont\FSSegEdge]
  {
    \fSAtomJudgeE{\fCyDirection}
                 {\fBagVar}
  }
  {
    \fSPathJudge{\fPathTable[1]\fCyDirection\,\fPathTable[2]}
                {{\fPathTable[1]}\!\fPTJoinOp
                 {\fPTEdgeLift{\fPropGraph}{\fCyDirection}{\fBagVar}}\fPTJoinOp
                 {\fPathTable[2]}
                }
  }
\end{mathpar}

\end{RulesDisplayCustomFont}

\noindent
\FSSegNode{} lifts the binding table ($\fBagVar$) with singleton records from a 
node atom ($\fCyPatternNode$) match into a path table with empty bound-edges,
\ie $\fVisitEdge=\varnothing$; using the 
metafunction $\fPTNodeLift{\fPropGraph}{\fCyPatternNode}{\fBagVar}$ defined 
in Figure~\ref{fig:runtime-entities}.
\FSSegEdgeStepL{} and \FSSegEdgeStepR{} reduce the two sides of
$\fRuntimeStx{\fCyPattern}\,\fCyDirection\,\fCyPatternNode$.
\FSSegEdge{} applies once both are path tables by lifting the
binding table of the edge atom ($\fCyDirection$) using the 
edge-symmetrical metafunction $\fPTEdgeLift{\fPropGraph}{\fCyDirection}{\fBagVar}$,
and then joining the path tables via path concatenation $\fPTJoinOp$.
This is the runtime counterpart of the endpoint condition
$\fEndpointCond{\cdot}{\cdot}{\cdot}$ from
\S\ref{sec:pattern-typing}; handling edge orientation via 
$\fEdgeEnds{\fPropGraph}{\fCyDirection}{\fPGEdgeElem}.$

\begin{RulesDisplay}
  \begin{mathpar}
  \infer[\RuleNameFont\FSPathConcatStepL]
  {
    \fSPathJudge{\fRuntimeStx{\fCyPatterns{1}}}
                {{\fRuntimeStx{\fCyPatterns{1}}}'}
  }
  {
    \fSPathJudge{\fPatConcat{\fRuntimeStx{\fCyPatterns{1}}}{\fRuntimeStx{\fCyPatterns{2}}}}
                {\fPatConcat{\fRuntimeStx{\fCyPatterns{1}}'}{\fRuntimeStx{\fCyPatterns{2}}}}
  }
  \quad
  \infer[\RuleNameFont\FSPathConcatStepR]
  {
    \fSPathJudge{\fRuntimeStx{\fCyPatterns{2}}}
                {{\fRuntimeStx{\fCyPatterns{2}}}'}
  }
  {
    \fSPathJudge{
                 \fPatConcat{\fPathTable[1]}
                            {\fRuntimeStx{\fCyPatterns{2}}}
                }
                {
                 \fPatConcat{\fPathTable[1]}
                            {{\fRuntimeStx{\fCyPatterns{2}}}'}
                }
  }
  \quad
  \infer[\RuleNameFont\FSPathConcat]
  {
    \fPathTable\fAssign
    \fPathTable[1]\fPTJoinOp\fPathTable[2]
  }
  {
    \fSPathJudge{
                 \fPatConcat{\fPathTable[1]}{\fPathTable[2]}
                }
                {\fPathTable}
  }
  % \\
  % \infer[\RuleNameFont\FSPathParen]
  % {
  %   \fSPathJudge{\fRuntimeStx{\fCyPattern}}
  %               {\fRuntimeStx{\fCyPattern}'}
  % }
  % {
  %   \fSPathJudge{\fCyParenthesis{\fRuntimeStx{\fCyPattern}}}
  %               {\fCyParenthesis{\fRuntimeStx{\fCyPattern}'}}
  % }
  % \and
  % \infer[\RuleNameFont\FSPathParenDone]
  % {
  %   \phantom{a}
  % }
  % {
  %   \fSPathJudge{\fCyParenthesis{\fPathTable}}
  %               {\fPathTable}
  % }
  \quad
  \infer[\RuleNameFont\FSPathQPathStep]
  {
    \fSPathJudge{\fRuntimeStx{\fCyPattern}}
                {\fRuntimeStx{\fCyPattern}'}
  }
  {
    \fSPathJudge{\fQFrame{\fRuntimeStx{\fCyPattern}}{\fQuantVars}{\fCyQuantifier}}
                {\fQFrame{\fRuntimeStx{\fCyPattern}'}{\fQuantVars}{\fCyQuantifier}}
  }
  \quad
  \infer[\RuleNameFont\FSPathQPathEval]
  {
    \fSQPathJudge{\fRuntimeStx{\fCyQuantifier}}
                 {\fRuntimeStx{\fCyQuantifier}'}
  }
  {
    \fSPathJudge{\fRuntimeStx{\fCyQuantifier}}
                {\fRuntimeStx{\fCyQuantifier}'}
  }
\end{mathpar}

\end{RulesDisplay}

\noindent
\FSPathConcatStepL{}, \FSPathConcatStepR{}, and \FSPathConcat{}
reduce path patterns connected by \lang's concatenation operator $\fPatConcat{\!}{\!}$.
This operator allows connecting path patterns at their endpoint node atoms by 
enforcing that the variables corresponding to these endpoint node atoms must agree, \ie are the same.
We do not consider this in the source calculus (Figure~\ref{fig:gql-calculus})
but introduce it in the runtime calculus (Figure~\ref{fig:runtime-entities}) for 
conveniently desugaring quantified edges into quantified path patterns with 
anonymous (variable declaration absent) endpoint node atoms.
\FSPathQPathStep{} similarly steps the path inside a quantifier frame
$\fQFrame{\cdot}{\fQuantVars}{\fCyQuantifier}$
while preserving the frame itself, and \FSPathQPathEval{} delegates it 
to the auxiliary quantified path judgement once the path pattern
is fully reduced to a path table.

\lang's quantified paths are parallel to regular-expression repetition.
Their semantics from the standard can be interpreted using our 
execution constructs, as the union
\smash{$\fPathTable^{\fQLo{\fCyQuantifier}}\cup\cdots\cup
\fPathTable^{\fQHi{\fCyQuantifier}}$}
of the path tables obtained by iterating the inner path table $\fPathTable$
between the quantifier's bounds.
We compute this union with a frontier. A quantified path runtime term
$\fQState{\fPathTable[A]}{\fVisitPath}{\fQuantIt}$ carries the
\emph{accumulator} $\fPathTable[A]$ of the matched bindings, accumulated 
starting from repetition $\fQLo{\fCyQuantifier}$ to repetition $\fQuantIt-1$; 
while the \emph{frontier} $\fVisitPath$ contains only those bindings matched 
using exactly $\fQuantIt$ repetitions.

\begin{RulesDisplay}
  \begin{mathpar}
  \infer[\RuleNameFont\FSQPathInit]
  {
    \fPathTable[A]\fAssign
      \mathsf{if}\;\fQLo{\fCyQuantifier}=0\;\mathsf{then}\;
      \fQZero{\fQuantVars}{\fCyQuantifier}\;\mathsf{else}\;\fBag{}
    \qquad
    \fVisitPath\fAssign
      \fQZero{\fQuantVars}{\fCyQuantifier}\fQJoinOp{\fCyQuantifier}\fPathTable
  }
  {
    \fSQPathJudge{\fQFrame{\fPathTable}{\fQuantVars}{\fCyQuantifier}}
                 {\fQState{\fPathTable[A]}{\fVisitPath}{1}}
  }
  \and
  \infer[\RuleNameFont\FSQPathFinish]
  {
    \fQuantIt>\fQHi{\fCyQuantifier}\;\lor\;
    \fVisitPath=\fBag{}
  }
  {
    \fSQPathJudge{\fQState{\fPathTable[A]}{\fVisitPath}{\fQuantIt}}
                 {\fPathTable[A]}
  }
  \and
  \infer[\RuleNameFont\FSQPathIter]
  {
    \fQuantIt\leq\fQHi{\fCyQuantifier}
    \qquad
    \fVisitPath\neq\fBag{}\\
    \fPathTable[A']\fAssign
      \mathsf{if}\;\fQLo{\fCyQuantifier}\leq\fQuantIt\;\mathsf{then}\;
      \fPathTable[A]\uplus\fVisitPath\;\mathsf{else}\;\fPathTable[A]\\
    \fVisitPath'\fAssign
      \mathsf{if}\;\fQuantIt<\fQHi{\fCyQuantifier}\;\mathsf{then}\;
      \fVisitPath\fQJoinOp{\fCyQuantifier}\fPathTable
      \;\mathsf{else}\;\fBag{}
  }
  {
    \fSQPathJudge{\fQState{\fPathTable[A]}{\fVisitPath}{\fQuantIt}}
                 {\fQState{\fPathTable[A']}{\fVisitPath'}{\fQuantIt+1}}
  }
\end{mathpar}

\end{RulesDisplay}

\noindent
\FSQPathInit{} initiates the iteration by computing the one-repetition 
frontier from the zero-repetition table by extending it with the inner path table
(Definition~\ref{def:quant-record-composition}), followed by instantiating 
the accumulator path table $\fPathTable[A]$ with the zero-repetition table itself
iff the quantifier lower bound $\fQLo{\fCyQuantifier}$ is zero. This enables
supporting \lang's \emph{empty match} feature for quantifiers.
\FSQPathIter{} applies as long as the counter tracking the number of 
repetitions $\fQuantIt$ is not greater than the quantifier 
upper bound $\fQHi{\fCyQuantifier}$, and as long as the current frontier 
$\fVisitPath$ is not empty. In each iteration, it accumulates the current 
frontier $\fVisitPath$ into the accumulator $\fPathTable[A]$ (
if $\fQuantIt\geq\fQLo{\fCyQuantifier}$), and then computes the next 
frontier $\fVisitPath'$ by extending the current one with the inner path table
$\fPathTable$. Since all extensions happen using the repetition extension  
operator $\fQJoinOp{\fCyQuantifier}$, endpoint continuity and edge-disjointness 
are implicitly enforced.

\MyPara{Pattern List Matching (Level 3)}
Pattern list matching composes path patterns via conjunction
($\fCyPatternAnd{\fCyPatternExp}{\fCyPattern}$), whose cross product
followed by equijoin semantics is captured by our trail-monoid's join
$\fBTJoinOp$ (Definition~\ref{def:trail-monoid}).
It joins on the shared variables while enforcing 
\matchModeDiff{} at the same time.

\begin{RulesDisplayCustomFont}{\small}
  \begin{mathpar}
  \infer[\RuleNameFont\FSPatExpSingle]
  {
    \fSPathJudgeC{\fRuntimeStx{\fCyPattern}}{\fRuntimeStx{\fCyPattern}'}
  }
  {
    \fPEStep{\fRuntimeStx{\fCyPattern}}{\fRuntimeStx{\fCyPattern}'}
  }
  \quad
  \infer[\RuleNameFont\FSPatExpDone]
  {
    \fDbBindTable\fAssign\fPTDrop{\fPathTable}
  }
  {
    \fPEStep{\fPathTable}{\fDbBindTable}
  }
  \quad
  \infer[\RuleNameFont\FSPatExpAndL]
  {
    \fPEStep{\fRuntimeStx{\fCyPatternExp}}{\fRuntimeStx{\fCyPatternExp}'}
  }
  {
    \fPEStep{\fCyPatternAnd{\fRuntimeStx{\fCyPatternExp}}{\fRuntimeStx{\fCyPattern}}}
             {\fCyPatternAnd{\fRuntimeStx{\fCyPatternExp}'}{\fRuntimeStx{\fCyPattern}}}
  }
  \quad
  \infer[\RuleNameFont\FSPatExpAndR]
  {
    \fPEStep{\fRuntimeStx{\fCyPattern}}{\fRuntimeStx{\fCyPattern}'}
  }
  {
    \fPEStep{\fCyPatternAnd{\fDbBindTable[1]}{\fRuntimeStx{\fCyPattern}}}
             {\fCyPatternAnd{\fDbBindTable[1]}{\fRuntimeStx{\fCyPattern}'}}
  }
  \quad
  \infer[\RuleNameFont\FSPatExpAnd]
  {
    \fDbBindTable\fAssign
      \fDbBindTable[1]\fBTJoinOp\fDbBindTable[2]
  }
  {
    \fPEStep{\fCyPatternAnd{\fDbBindTable[1]}{\fDbBindTable[2]}}
            {\fDbBindTable}
  }
  %% \quad
  %% \infer[\RuleNameFont\FSPatExpOrL]
  %% {
  %%   \fPEStep{\fPropGraph}{\fCyPatternExp}{\fCyPatternExp'}
  %% }
  %% {
  %%   \fPEStep{\fPropGraph}
  %%            {\fCyPatternOr{\fCyPatternExp}{\fCyPattern}}
  %%            {\fCyPatternOr{\fCyPatternExp'}{\fCyPattern}}
  %% }
  %% \quad
  %% \infer[\RuleNameFont\FSPatExpOrR]
  %% {
  %%   \fSPathJudgeC{\fPropGraph}{\fCyPattern}{\fCyPattern'}
  %% }
  %% {
  %%   \fPEStep{\fPropGraph}
  %%            {\fCyPatternOr{\fDbBindTable}{\fCyPattern}}
  %%            {\fCyPatternOr{\fDbBindTable}{\fCyPattern'}}
  %% }
  %% \quad
  %% \infer[\RuleNameFont\FSPatExpOr]
  %% {
  %%   \phantom{x}
  %% }
  %% {
  %%   \fPEStep{\fPropGraph}
  %%            {\fCyPatternOr{\fDbBindTable[1]}{(\fDbBindTable[2],\fCyNodeVar)}}
  %%            {\fDbBindTable[1]\uplus\fDbBindTable[2]}
  %% }
\end{mathpar}

\end{RulesDisplayCustomFont}

\noindent
\FSPatExpSingle{} lifts a single path pattern step into the pattern
list judgement, and \FSPatExpDone{} lowers the reduced path table
of a path pattern into a trail table by projecting only the record and 
bound-edge set attributes via $\fPTDrop{\cdot}$ from Figure~\ref{fig:runtime-entities}.
\FSPatExpAndL{} and \FSPatExpAndR{} reduce either side of a
conjunction left-to-right, while \FSPatExpAnd{} joins their trail tables 
when fully reduced.

\subsection{Query Semantics}
\label{sec:sem-op-query}

Formalizing query evaluation semantics is yet again relatively simple 
since most of the complexity has been delegated to value expression and 
pattern matching semantics, with only the formalization of the flow of the binding 
tables between query clauses remaining.
We use two judgement forms:

\begin{TypingDisplay}
\[
\underset{\scriptsize\text{Linear Query}}%
  {\fQLinStep{\fDbCatalog}{\fRuntimeStx{\fCyQuery}}
                       {\fRuntimeStx{\fCyQuery}'}}
\qquad\qquad\qquad\qquad
\underset{\scriptsize\text{Composite Query}}%
         {\fQStep{\fDbCatalog}{\fRuntimeStx{\fCyQueryExp}}
                              {\fRuntimeStx{\fCyQueryExp}'}}
\]
\vspace{-2mm}%
\end{TypingDisplay}

\noindent%
The \emph{linear query judgement} rewrites linear queries in the 
context of the database catalog $\fDbCatalog$ into one of two 
intermediate tuples: (1)~$\langle\fPropGraph,\fRuntimeStx{\fCyPatternExp},\fCyPred,\fCyProjection\rangle$
---for formalizing $\fMatch$ and $\fWhere$ clause evaluation, and 
(2)~$\langle\fPropGraph,\fBagVar,\fCyProjection\rangle$---for 
formalizing $\fReturn$ clause evaluation.
The \emph{composite query judgement} formalizes composite query evaluation
where the semantics 
depend on the composite operator $\fCyCompOp$.

\MyPara{Linear Query Evaluation}%
Focused linear queries first select a working graph, 
then match a pattern against it, optionally filter the matched bindings,
before projecting to generate the result.

\begin{RulesDisplayCustomFont}{\small}
  \begin{mathpar}
  \infer[\RuleNameFont\FSQueryUse]
  {
    \fDbGNames\in\fDom{\fDbCatalog}\\
    \fPropGraph = \fDbCatalog(\fDbGNames)\\
    \fRuntimeStx{\fCyPatternExp}\fAssign\fNormOf{\fCyPatternExp}
  }
  {
    \fQLinStep{\fDbCatalog}
      {\fUse\;\fDbGNames\;\fMatch\;\fCyPatternExp\;
       \fWhere\;\fCyPred\;\fReturn\;\fCyProjection}
      {\langle\fPropGraph,\fRuntimeStx{\fCyPatternExp},\fCyPred,\fCyProjection\rangle}
  }
  % \and
  % \infer[\RuleNameFont\FSQueryUseNoWhere]
  % {
  %   \fDbGNames\in\fDom{\fDbCatalog}\\
  %   \fPropGraph = \fDbCatalog(\fDbGNames)\\
  %   \fNormalizePat{\fCyPatternExp}{\bar \fCyPatternExp}{\fDbVarUse}
  % }
  % {
  %   \fQLinStep{\fDbCatalog}
  %     {\fUse\;\fDbGNames\;\fMatch\;\fCyPatternExp\;\fReturn\;\fCyProjection}
  %     {\langle\fPropGraph,\bar \fCyPatternExp,\fTrue,\fCyProjection\rangle}
  % }
  \and
  \infer[\RuleNameFont\FSQueryMatch]
  {
    \fPEStepStar{\fRuntimeStx{\fCyPatternExp}}{\fDbBindTable}
  }
  {
    \fQLinStep{\fDbCatalog}
      {\langle\fPropGraph,\fRuntimeStx{\fCyPatternExp},\fCyPred,\fCyProjection\rangle}
      {\langle\fPropGraph,\fDbBindTable,\fCyPred,\fCyProjection\rangle}
  }
  \and
  \infer[\RuleNameFont\FSQueryMatchDone]
  {
    \fBagVar \fAssign \fBag{
      \fCyTuple \in \fAnonDrop{\fAttrAnon}{\fBTDrop{\fDbBindTable}} \mid
      \fPropGraph;\,\fCyTuple\vdash
        \fCyPred\,\fOpStepStar\,\fTrue
    }
  }
  {
    \fQLinStep{\fDbCatalog}
      {\langle\fPropGraph,\fDbBindTable,\fCyPred,\fCyProjection\rangle}
      {\langle\fPropGraph,\fBagVar,\fCyProjection\rangle}
  }
  \and
  \infer[\RuleNameFont\FSQueryProject]
  {
    \fBagVar' \fAssign \fBag{
      [\fCyVars{i}\!\mapsto\!\fDbRVals{i}]
      %% _{
      %%   (\fCyExps{i}\,\fAs\,\fCyVars{i})\in\fCyProjection}
      \;\Big|\;
      \begin{array}{@{}l@{}}
        \fCyTuple\!\in\!\fBagVar,\;
        \forall (\fCyExps{i}\;\fAs\;\fCyVars{i})\!\in\!\mu.\\
        \fPropGraph;\,\fCyTuple\vdash
        \fCyExps{i}\,\fOpStepStar\,\fDbRVals{i}
      \end{array}
    }
  }
  {
    \fQLinStep{\fDbCatalog}
      {\langle\fPropGraph,\fBagVar,\mu\rangle}
      {\fBagVar'}
  }
\end{mathpar}

\end{RulesDisplayCustomFont}

\noindent
\FSQueryUse{} resolves the graph $\fDbGNames$ against the catalog
$\fDbCatalog$ and normalizes the pattern list of the $\fMatch$ clause
using $\fNormOf{\cdot}$ (Figure~\ref{fig:runtime-entities}), which
expands its quantified edges and names its anonymous pattern atoms, 
before rewriting the query to an intermediate tuple. Only the
rule for queries with a $\fWhere$ clause is shown since those without can 
be rewritten into one.
\FSQueryMatch{} reduces the pattern list to a trail table using the
multi-step closure $\fOpStepStar$ of the pattern list judgement.
\FSQueryMatchDone{} filters this trail table by lowering it to a binding
table via $\fBTDrop{\cdot}$ from Figure~\ref{fig:runtime-entities}, before 
dropping the variables introduced by $\fNormOf{\cdot}$ using 
$\fAnonDrop{\fAttrAnon}{\cdot}$ (as they cannot be referenced in any clause); and finally retaining only the records
for which the predicate reduces to $\fTrue$.
\FSQueryProject{} projects the optionally filtered binding table using the projection list
$\fCyProjection$ to generate the final binding table.

\MyPara{Composite Query Evaluation}%
Composite queries compose linear queries via composite operators.

\begin{RulesDisplayCustomFont}{\small}
  \begin{mathpar}
  \infer[\RuleNameFont\FSCompQueryLift]
  {
    \fQLinStep{\fDbCatalog}{\fRuntimeStx{\fCyQuery}}{\fRuntimeStx{\fCyQuery}'}
  }
  {
    \fQStep{\fDbCatalog}
      {\fRuntimeStx{\fCyQuery}}
      {\fRuntimeStx{\fCyQuery}'}
  }
  \and
  \infer[\RuleNameFont\FSCompQueryL]
  {
    \fQStep{\fDbCatalog}{\fRuntimeStx{\fCyQueryExp}}{\fRuntimeStx{\fCyQueryExp}'}
  }
  {
    \fQStep{\fDbCatalog}
      {\fRuntimeStx{\fCyQueryExp}\;\fCyCompOp\;\fRuntimeStx{\fCyQuery}}
      {\fRuntimeStx{\fCyQueryExp}'\;\fCyCompOp\;\fRuntimeStx{\fCyQuery}}
  }
  \and
  \infer[\RuleNameFont\FSCompQueryR]
  {
    \fQStep{\fDbCatalog}{\fRuntimeStx{\fCyQuery}}{\fRuntimeStx{\fCyQuery}'}
  }
  {
    \fQStep{\fDbCatalog}
      {\fBagVar[1]\;\fCyCompOp\;\fRuntimeStx{\fCyQuery}}
      {\fBagVar[1]\;\fCyCompOp\;\fRuntimeStx{\fCyQuery}'}
  }
  \and
  \infer[\RuleNameFont\FSCompQueryUnion]
  {
    \phantom{x}
  }
  {
    \fQStep{\fDbCatalog}
      {\fBagVar[1]\;\fUnion\;\fBagVar[2]}
      {\fBagVar[1]\uplus\fBagVar[2]}
  }
  % \and
  % \infer[\RuleNameFont\FSCompQueryOtherwise]
  % {
  %   \fSup{\fBagVar[1]}\neq\varnothing
  % }
  % {
  %   \fQStep{\fDbCatalog}
  %     {\fBagVar[1]\;\fOtherwise\;\fBagVar[2]}
  %     {\fBagVar[1]}
  % }
  % \and
  % \infer[\RuleNameFont\FSCompQueryOtherwiseEmpty]
  % {
  %   \fSup{\fBagVar[1]}=\varnothing
  % }
  % {
  %   \fQStep{\fDbCatalog}
  %     {\fBagVar[1]\;\fOtherwise\;\fBagVar[2]}
  %     {\fBagVar[2]}
  % }
  % \and
  % %% \infer[\RuleNameFont\FSCompQueryExceptD]
  % %% {
  % %%   \phantom{x}
  % %% }
  % %% {
  % %%   \fQStep{\fDbCatalog}
  % %%     {\fBagVar[1]\;\fExceptD\;\fBagVar[2]}
  % %%     {\mathsf{dist}(\fBagVar[1]\setminus\fBagVar[2])}
  % %% }
  % %% \and
  % \infer[\RuleNameFont\FSCompQueryExceptA]
  % {
  %   \phantom{x}
  % }
  % {
  %   \fQStep{\fDbCatalog}
  %     {\fBagVar[1]\;\fExceptA\;\fBagVar[2]}
  %     {\fBagVar[1]\setminus\fBagVar[2]}
  % }
  % \and
  % %% \infer[\RuleNameFont\FSCompQueryIntersectD]
  % %% {
  % %%   \phantom{x}
  % %% }
  % %% {
  % %%   \fQStep{\fDbCatalog}
  % %%     {\fBagVar[1]\;\fIntersectD\;\fBagVar[2]}
  % %%     {\mathsf{dist}(\fBagVar[1]\cap\fBagVar[2])}
  % %% }
  % %% \and
  % \infer[\RuleNameFont\FSCompQueryIntersectA]
  % {
  %   \phantom{x}
  % }
  % {
  %   \fQStep{\fDbCatalog}
  %     {\fBagVar[1]\;\fIntersectA\;\fBagVar[2]}
  %     {\fBagVar[1]\cap\fBagVar[2]}
  % }
\end{mathpar}

\end{RulesDisplayCustomFont}

\noindent
\FSCompQueryLift{} lifts linear query stepping into composite query judgement.
\FSCompQueryL{} and \FSCompQueryR{} reduce the operands left-to-right until 
both operands are binding tables.
\FSCompQueryUnion{} computes bag union preserving multiplicities.
Since the judgement is indexed by the composite operator $\fCyCompOp$,
the congruence rules thread the same operator through both operands, 
enforcing the \lang standard's requirement that composite queries 
use a single composite operator across all their linear queries.

\section{Type Soundness}
\label{sec:metatheory}

We prove \emph{type soundness}, \ie every well-formed query produces 
a binding table when evaluated to completion; that conforms to the 
schema assigned to the query statically by the typing rules.
The typing judgments assign types to \emph{source terms}, while the
semantics reduce \emph{runtime terms} containing intermediate forms.
Auxiliary conformance and configuration-typing relations bridge this
gap.

\begin{myDef}[Value Typing and Conformance]\label{def:value-typing}
  A runtime value $\fDbRVal\in\fValSet$ \emph{inhabits} a type
  $\fSortVar$, written $\fValType{\fDbRVal}{\fSortVar}$, when the value
  is a member of the type's semantic domain: scalars inhabit their base
  type, $\fNull$ inhabits $\fNullable$ and every nullable type
  $\fSortVar\fNullable$, graph elements inhabit their graph-indexed
  identity types (up to subtyping), and lists inhabit
  $\fListT\;\fSortVar$ when every element does. Value typing is closed
  under subtyping: $\fValType{\fDbRVal}{\fSortVar}$ and
  $\fSubType{\fSortVar}{\fSortVars{2}}$ imply
  $\fValType{\fDbRVal}{\fSortVars{2}}$.
  A record $\fCyTuple$ then conforms to a record schema
  $\fDbRSchema$, written $\fConform{\fCyTuple}{\fDbRSchema}$, when
  \(
    \fDom{\fCyTuple}=\fDom{\fDbRSchema}
  \)
  and
  \(
    \forall\fCyVar\in\fDom{\fDbRSchema}.\fCyTuple(\fCyVar)\!:\fSortVar.\,
    \fSubType{\fSortVar\!}{\fDbRSchema(\fCyVar)},
  \)
  and a binding or trail table conforms to $\fDbRSchema$ when every
  record in its support does.
  A path table $\fPathTable$ conforms to $\fDbRSchema$ when every entry
  $(\fCyTuple,\fPGNodeElems{\texttt{lft}},\fPGNodeElems{\texttt{rht}},
  \fVisitEdge)$ has
  $\fConform{\fCyTuple}{\fDbRSchema}$,
  $\fPGNodeElems{\texttt{lft}},\fPGNodeElems{\texttt{rht}}\in\fPGNode$, and
  $\fVisitEdge\subseteq\fPGEdge$; we write
  $\fConform{\cdot}{\fDbRSchema}$ for all of these.
\end{myDef}

\MyPara{Configuration Typing}%
A database world $\fDbWorldS=(\fDbCatalog,\fDbSchema)$ is
\emph{well-formed}, written $\fWFWorld$, when every closed graph site's
instance conforms to its declared graph schema
(Definition~\ref{def:conformance-relation}).
%
% Two further relations, defined in Appendix~\ref{sec:metatheory-aux},
% carry typing over to what the semantics reduces.
% %
% \emph{Normalization typing} assigns $\fNormOf{\fCyPatternExp}$ an
% \emph{internal schema} $\fDbRSchemas{\mathsf{rt}}$ that types the
% attributes generated by $\fAnonFill{\cdot}$ alongside the named ones,
% giving each generated attribute the type of its atom under the
% quantifier lift (Definition~\ref{def:quant-lift}) of every enclosing
% quantified frame; since generated attributes are drawn from
% $\fAttrAnon$, erasing them recovers the source schema $\fDbRSchema$, so
% normalization is type preserving.
%
\emph{Configuration typing}---written
$\fPConfigType{\fRuntimeStx{\fCyPatternExp}}
              {\fDbRSchemas{\mathsf{rt}}}{\fDbRSchema}$
for a pattern configuration and
$\fConfigType{\fRuntimeStx{\fCyQueryExp}}{\fDbRSchema}$ for a query
configuration---closes the source typing rules under the runtime forms
of Figure~\ref{fig:runtime-entities}, requiring every table a
configuration has materialized to conform to the internal schema of the
subterm it replaced, and the clauses that remain to compose that schema
into $\fDbRSchema$.

\begin{myTheorem}[Expression Soundness]\label{thm:expr-sound}
  Suppose $\fWFWorld$, $\fPropGraph=\fDbCatalog(\fDbGNames)$, and
  $\fConform{\fCyTuple}{\fDbRSchema}$.
  \begin{enumerate}[nosep,leftmargin=2em]
    \item \emph{(Progress)}\;
      If $\fCyExp\notin\fValSet$, then there exists $\fCyExp'$ such that
      $\fEval{\fPropGraph}{\fCyTuple}{\fCyExp}{\fCyExp'}$.
    \item \emph{(Soundness)}\;
      If\; $\fExpJudgeParam[\fDbWorld;\fDbRSchema][\fCyExp][\fSortVar]$
      is derivable for some $\Box$, $\Diamond$, $\fDbVarUse$, and
      $\fPropGraph;\,\fCyTuple\vdash\fCyExp\,\fOpStepStar\,\fDbRVal$
      with $\fDbRVal\in\fValSet$, then
      $\fValType{\fDbRVal}{\fSortVars{1}}$ for some
      $\fSubType{\fSortVars{1}}{\fSortVar}$.
  \end{enumerate}
\vspace{-2.5mm}%
\begin{proof}[Proof sketch]
  Progress: every computation rule has a complementary null rule, so a
  non-value either has a reducible operand or contracts.
  Soundness: by induction on $\fCyExp$, splitting the reduction
  sequence into subterm evaluations along the left-to-right congruence
  rules. Computation rules return values at the operator's result type,
  while the null rules---which fire exactly when an operand's dynamic
  refinement fails at
  runtime---return $\fNull$ at type $\fNullable$. Both are subtypes of
  the nullable types the rules assign, which is why the conclusion is up
  to subtyping.
\end{proof}
\end{myTheorem}

\begin{myTheorem}[Pattern Configuration Safety]\label{thm:pattern-config-safety}
  Suppose $\fWFWorld$,
  $\fPropGraph=\fDbCatalog(\fDbGNames)$, and
  \(
    \fPConfigType{\fRuntimeStx{\fCyPatternExp}}
      {\fDbRSchemas{\mathsf{rt}}}{\fDbRSchema}.
  \)
  \begin{enumerate}[nosep,leftmargin=2em]
    \item \emph{(Progress)}\;
      Either $\fRuntimeStx{\fCyPatternExp}=\fDbBindTable$, or there
      exists $\fRuntimeStx{\fCyPatternExp}'$ such that
      $\fPEStep{\fRuntimeStx{\fCyPatternExp}}
               {\fRuntimeStx{\fCyPatternExp}'}$.
    \item \emph{(Preservation)}\;
      If $\fPEStep{\fRuntimeStx{\fCyPatternExp}}
                  {\fRuntimeStx{\fCyPatternExp}'}$, then
      \(
        \fPConfigType{\fRuntimeStx{\fCyPatternExp}'}
          {\fDbRSchemas{\mathsf{rt}}}{\fDbRSchema}.
      \)
    \item \emph{(Terminal conformance)}\;
      If $\fRuntimeStx{\fCyPatternExp}=\fDbBindTable$, then
      $\fConform{\fAnonErase{\fDbBindTable}}{\fDbRSchema}$.
  \end{enumerate}
\vspace{-2.5mm}%
\begin{proof}[Proof sketch]
  By induction on the configuration-typing derivation. The path cases
  use atom conformance and closure of the path-table operations under
  compatible record joins, and the conjunction case trail-join
  conformance and its schema-compatibility premise; trail and endpoint
  checks only remove entries. The quantified case uses an inner
  induction on~$\fQuantIt$: the zero-repetition table
  $\fQZero{\fQuantVars}{\fCyQuantifier}$ inhabits the lifted schema, and
  each application of $\fQJoinOp{\fCyQuantifier}$ preserves it since
  $\fQRecExt$ extends a list-typed (or, for $\fCyQuantQues$, nullable)
  binding by one element of the inner schema's type. Terminal
  conformance then follows since $\fBTDrop{\cdot}$ and
  $\fAnonDrop{\fAttrAnon}{\cdot}$ only erase annotations/attributes that $\fDbRSchemas{\mathsf{rt}}$ adds.
\end{proof}
\end{myTheorem}

\begin{myCorol}[Pattern Soundness]\label{thm:pattern-sound}
  If $\fWFWorld$,
  $\fPropGraph=\fDbCatalog(\fDbGNames)$,
  $\fPatExpJudge{\fDbWorld}{\fCyPatternExp}{\fDbRSchema}$, and
  \(
    \fPEStepStar{\fNormOf{\fCyPatternExp}}{\fDbBindTable},
  \)
  then
  $\fConform{\fAnonErase{\fDbBindTable}}{\fDbRSchema}$.
\vspace{-2.5mm}%
\begin{proof}[Proof sketch]
  Normalization typing supplies the internal schema
  $\fDbRSchemas{\mathsf{rt}}$; iterated preservation and terminal
  conformance from Theorem~\ref{thm:pattern-config-safety} then give the
  result.
\end{proof}
\end{myCorol}

\begin{myTheorem}[Query Type Soundness]\label{thm:query-sound}
  If $\fWFWorld$,
  $\fQueryJudge{\fDbWorldS}{\fCyQuery}{\fDbRSchema}$, and
  $\fDbCatalog\vdash\fCyQuery\;\fOpStepStar\;\fBagVar$,
  then $\fConform{\fBagVar}{\fDbRSchema}$.
\vspace{-2.5mm}%
\begin{proof}[Proof sketch]
  Progress and preservation for configurations, then iteration.
  \emph{Progress}: well-formedness of $\fDbWorldS$ lets a source query
  resolve its graph site; a pattern-stage configuration steps its
  pattern list, and transitions to a 3-tuple once that list has reduced
  to a trail table and every predicate to a value---both terminate,
  since each pattern rule either reduces a subterm to a table or
  advances a counter capped by $\fQHi{\fCyQuantifier}\leq\fSize{\fPGEdge}$,
  and each expression step strictly decreases the number of non-value
  subterms; a 3-tuple projects (by
  Theorem~\ref{thm:expr-sound}); a final bag is terminal.
  \emph{Preservation}: graph resolution (\FSQueryUse) produces a
  well-typed pattern-stage configuration; pattern stepping
  (\FSQueryMatch) preserves it by
  Theorem~\ref{thm:pattern-config-safety}; filtering (\FSQueryMatchDone)
  lowers the trail table, erases the generated
  attributes---unreferenceable, since no source query mentions an
  attribute in $\fAttrAnon$---and selects a sub-bag, so terminal
  conformance gives a 3-tuple well-typed at $\fDbRSchemas{1}$; and
  projection (\FSQueryProject) applies Theorem~\ref{thm:expr-sound} per
  expression.
  Iterating yields $\fConfigType{\fBagVar}{\fDbRSchema}$, \ie
  $\fConform{\fBagVar}{\fDbRSchema}$.
\end{proof}
\end{myTheorem}

\begin{myCorol}[Composite Query Soundness]\label{cor:composite-sound}
  If $\fWFWorld$,
  $\fCompQueryJudge{\fDbWorldS}{\fCyQueryExp}{\fDbRSchema}$, and
  $\fDbCatalog\vdash_{\fCyCompOp}\fCyQueryExp\;\fOpStepStar\;\fBagVar$,
  then $\fConform{\fBagVar}{\fDbRSchema}$.
\vspace{-2.5mm}%
\begin{proof}[Proof sketch]
  By induction on the composite-query typing derivation.
  The base case is Theorem~\ref{thm:query-sound}; in the inductive case
  \FSCompQueryLift{}, \FSCompQueryL{}, and \FSCompQueryR{} reduce the
  operands to bags conforming to their operand schemas, and every
  composite operator either merges the two bags or selects one of
  them, so the result conforms to the
  composite schema that \FTyCompQuery{} assigns.
\end{proof}
\end{myCorol}

\section{Mechanized Implementation}
\label{sec:mechanization}

We implemented a complete mechanization in Lean~4 comprising
more than 23{,}000 lines across 14 modules, following
the \lang calculus in Figure~\ref{fig:gql-calculus}.  Table~\ref{tab:modules} summarizes the module
layout; the core formalization and its soundness proofs (excluding
tests and benchmarks) total roughly 21{,}300 lines.

The sort hierarchy is encoded as inductive types: \texttt{BaseSort} for
$T_0$, \texttt{ExtSort} for $T_1$ (including graph-scoped and
schema-refined types), and \texttt{GSort} for full sorts~$\fSortVar$,
while values use a mutual inductive
(\texttt{Value}/\texttt{ValueList}) for kernel-derived
\texttt{DecidableEq} over nested lists. Subtyping is an inductive
proposition with 17~constructors matching the paper's rules;
\textsc{S-Union-Congruence} is derived as a theorem from primitives due
to a Lean kernel restriction on nested inductives.
% All absorption laws (Lemma~4.1) and the bounding chain (Lemma~4.2) are
% proved without \texttt{sorry}.
The typing rules of Section~\ref{sec:typing-rules} are encoded as eleven inductive
propositions totaling 57~constructors across expression, predicate,
property-constraint, atom, refinement, pattern, pattern-expression,
projection, projection-list, query, and single-operator composite
typing.

\begin{wraptable}{r}{0.52\linewidth}
\vspace{-1mm}
\centering
\TableFont
\caption{Module structure of the Lean~4 mechanization.}
\label{tab:modules}
\begin{tabular}{llr}
\toprule
\textbf{Module} & \textbf{Formalization} & \textbf{LoC} \\
\midrule
\texttt{Sorts.lean}        & Sec.~\hyperref[sec:type-system]{\ref*{sec:type-system}} ($\fMSortPG$, $\fMSortGQL$, $\fMSortGQLN$)          & 254 \\
\texttt{Values.lean}       & Sec.~\hyperref[sec:prelims-pgm]{\ref*{sec:prelims-pgm}} (values, records) & 362 \\
\texttt{Syntax.lean}       & Fig.~\hyperref[fig:gql-calculus]{\ref*{fig:gql-calculus}} (expr., patterns, queries)    & 291 \\
\texttt{Schema.lean}       & Def.~\hyperref[def:data-model]{\ref*{def:data-model}}--\hyperref[def:conformance-relation]{\ref*{def:conformance-relation}} (property graphs)& 290 \\
\texttt{RecordSchema.lean} & Def.~\hyperref[def:type-intersect]{\ref*{def:type-intersect}}--\hyperref[def:record-schema-union]{\ref*{def:record-schema-union}}, Sec.~\hyperref[sec:typing-rules]{\ref*{sec:typing-rules}} (record ops) & 375 \\
\texttt{Subtyping.lean}    & Sec.~\hyperref[subsec:typing-subtyping]{\ref*{subsec:typing-subtyping}} (subtyping) & 297 \\
\texttt{Typing.lean}       & Sec.~\hyperref[sec:typing-rules]{\ref*{sec:typing-rules}} (typing judgments)               & 1{,}133 \\
\texttt{Semantics.lean}    & Sec.~\hyperref[sec:semantics-operational]{\ref*{sec:semantics-operational}} (executable semantics)  & 638 \\
\texttt{SmallStep.lean}    & Sec.~\hyperref[sec:semantics-operational]{\ref*{sec:semantics-operational}} (step relations) & 2{,}267 \\
\texttt{Metatheory.lean}   & Sec.~\hyperref[sec:metatheory]{\ref*{sec:metatheory}} (soundness proofs)                & 14{,}057 \\
\texttt{TypeChecker.lean}  & Sec.~\hyperref[sec:typing-rules]{\ref*{sec:typing-rules}} (executable checker)      & 1{,}364 \\
\midrule
\texttt{Test.lean}         & Unit tests (276 assertions)               & 1{,}252 \\
\texttt{Examples.lean}     & Worked examples (12 files)  & 404 \\
\texttt{LDBCBench.lean}    & LDBC SNB integration tests & 671 \\
% \texttt{Counterexample.lean} & Machine-checked counterexample          & 99 \\
\midrule
\multicolumn{2}{l}{\textbf{Total}}                                & \textbf{23{,}655} \\
\bottomrule
\end{tabular}
\end{wraptable}
The semantics module provides executable definitions for all reduction
rules, including direction-aware endpoint conditions,
frontier-based quantified path iteration, and set operations with
duplicate elimination. Beyond the executable definitions, the mechanization proves the full
soundness development of Section~6 (\texttt{Metatheory.lean}), the
small-step layer with progress, preservation, and a proven-equivalent
deterministic interpreter behind a user-facing engine flag
(\texttt{SmallStep.lean}), and a certified executable type checker
(\texttt{TypeChecker.lean}): \texttt{inferQuery} computes a result
schema for any query in the fragment, \texttt{inferQuery\_sound} shows
every accepted query is well typed in the declarative system, and
composition with query type soundness guarantees that an accepted
query's result table conforms to the inferred schema.  All results
depend only on the standard Lean axioms (\texttt{propext},
\texttt{Classical.choice}, \texttt{Quot.sound}), with zero
\texttt{sorry} and zero user-declared axioms.

\begin{wraptable}{r}{0.39\linewidth}
\vspace{-3mm}
\centering
\TableFont
\caption{Sample LDBC SNB queries type checked and evaluated using \tool.}
\label{tab:ldbc}
\begin{tabular}{@{}lll@{}}
\toprule
\textbf{Query} & \textbf{Pattern} & \textbf{Features} \\
\midrule
IS1 & 1 directed edge    & label, prop constraint \\
IS3 & 1 undirected edge  & undirected direction \\
IS4 & single node        & prop constraint \\
IS5 & 1 directed edge    & label filter \\
IC8 & 3-hop chain        & conjunction, left dir. \\
IC2 & 2-hop + WHERE      & conj., WHERE filter \\
\bottomrule
\end{tabular}
\end{wraptable}

The artifact's tests are layered: 276 unit assertions cover each
computational unit and rule, 30 worked examples mirror the paper's
definitions one by one, including negative cases, and the six
expressible queries of the LDBC SNB Interactive v2
workload~\cite{ldbc-snb}, translated into \lang{} as \texttt{Query}
terms, form an end-to-end integration test: each is type checked by the
certified checker, evaluated against golden results, checked to conform
to its inferred schema, and executed on both engines with bit-for-bit
agreement.  Table~\ref{tab:ldbc} summarizes the queries from the 
SNB workload, evaluated and type checked using \tool{}. In total, the
build checks 350 \texttt{native\_decide} assertions.  
Supporting other queries requires features outside the considered 
fragment with most leaning towards \lang's relational core, \ie data transformations.
% The remaining
% LDBC queries require features outside the formalized
% fragment---primarily \texttt{WITH}-based pipelining, \texttt{OPTIONAL
% MATCH}, and \texttt{shortestPath}.

%% All stated theorems
%% are fully proved.  The mechanization focuses on the core
%% \texttt{$\fMatch$\,...\,$\fWhere$\,...\,$\fReturn$} fragment with path patterns,
%% property constraints, label expressions, set operations, and bounded
%% quantified paths; extending to additional \lang features is future work.

\section{Related Work}
\label{sec:related-work}
% We now discuss prior work most closely related to \tool{}.

\MyPara{Graph Query Languages}
Early graph query formalisms centered on \rdf and
\sparql~\cite{perez2009sparql,sparql11,losemann2013sparql,kostylev2015sparql},
which operate on triple-based data and require indirect encodings for
the richer metadata found in property graphs.
Traversal-based languages such as
Gremlin~\cite{rodriguez2015gremlin} offer fine-grained navigational
control at the expense of declarative reasoning. In contrast, 
pattern-matching languages such as \cypher~\cite{francis2018cypher},
PGQL~\cite{van2016pgql}, and G-CORE~\cite{angles2018gcore}; provide
a more declarative interface.
\lang~\cite{deutsch2022graph,iso2024gql} is
the first standardized graph query language for property graphs, 
standardized as \iso in 2024.
It unifies ideas from its predecessors and incorporates 
graph schemas as part of the language specification.
Yet the standard's entirely informal specification of \lang's semantics, 
is not conducive to formal reasoning; leading to potential 
inconsistencies and ambiguities when implementing it, or even 
in the standard's specification itself.
\citet{Francis2023gpc} distilled \lang's pattern
matching into a Graph Pattern Calculus (GPC) equipped with typing
rules and a
denotational semantics under set semantics, and \citet{gheerbrant2025gql}
defined Core GQL and Core PGQ as concise
formal models focused on understanding its expressivity.
Our work is complementary to these efforts. We provide a small-step
operational semantics, a schema-aware type system with a
machine-checked soundness proof; all under bag semantics. None of these are
covered by GPC or Core GQL.

\MyPara{Language Mechanization}
Mechanizing language specifications in proof assistants has
offered several cross-domain benefits.
\citet{bodin2014trusted} mechanized ECMAScript\,5 in \coq, producing
both a formal semantics and an extracted reference interpreter.
\citet{de2024coq} gave a comprehensive \coq mechanization of
JavaScript regular expressions per ECMA-262, uncovering errors in
previous formalisms.
\citet{seassau2025ocaml} formalized a substantial OCaml subset in Rocq.
But \lang's specification has not been mechanized before, making 
\tool{} the first to mechanize it and thus providing the 
first formal foundation for rigorously reasoning about its semantics.

\MyPara{Formal Semantics of Query Languages}
Relational query languages have been formalized extensively.
\citet{guagliardo2017formal} gave a formal
semantics for a large \sql fragment, while ~\citet{benzaken2019coq} and ~\citet{chu2017hottsql} provided mechanized \coq treatments of \sql
with nulls, aggregations, and bag semantics.
For graph query languages, Regular Path Queries and their
extensions~\cite{barcelo2013querying} primarily focus 
on providing a clean compact theoretical foundation, leaving out
several key features of \lang, \cypher, \etc such as properties, 
schemas, and bag semantics.
Formal semantics exist for core \cypher
constructs~\cite{francis2018cypher,angles2017foundations}, 
and~\citet{ye2025flexible} proposed a gradual-typing calculus for \lang's
path patterns, but none of these are machine-checked.
~\citet{diaz2020graphcoql} mechanized GraphQL in \coq,
but GraphQL is a tree-shaped API language and does not contain 
a graph pattern matching fragment.
Our work is the first mechanized formalization of
\lang in a proof assistant (\lean).

\section{Conclusion}
\label{sec:conclusion}

We have presented \tool, a formal semantics for a substantial fragment
of the \iso Graph Query Language.
Our development provides three interlocking contributions: a type
system that exploits graph schemas for static refinement while
tracking three-valued nullability and heterogeneous union types; a
small-step operational semantics that makes explicit the interplay
among Kleene three-valued logic, trail-aware pattern composition,
quantified-path iteration, and the clause-by-clause query pipeline;
and a layered type soundness argument establishing that well-typed
queries produce results conforming to their declared schemas. \tool{}
provides the first bridge between \lang's specification and a 
mechanized implementation.

\section*{Acknowledgments}

We thank Cheng Ding, Ivan Grigorik, Linghan Zhong, Lara Marinov, 
and the anonymous reviewers for helpful feedback and discussions.
This work was supported in part by the U.S. National Science
Foundation (NSF) Nos.~CCF-2217696, CCF-2313027, CCF-2403036; 
and an Amazon Research Award (Fall 2025).
Any opinions, findings, and conclusions or recommendations expressed 
in this material are those of the authors and do not necessarily reflect 
the views of the NSF or Amazon.

\section*{Data Availability Statement}

% We will import our mechanization of cypher query from LDBC SNB
% Interactive v2 workload~\cite{ldbc-snb} 
% using Lean4.

We mechanized our small-step formalization of the semantics of GQL in \lean.
The artifact supporting this paper is available on Zenodo~\cite{thimmaiah2026mgqlArtifact}.
It also contains a big-step semantics formalization.

\newpage
\bibliographystyle{ACM-Reference-Format}
\bibliography{bib}
%\newpage
%\input{appendix}

\end{document}